\PassOptionsToClass{reprint,nofootinbib}{revtex4-2}

\documentclass[amsmath,amssymb,aps,prl,floatfix]{revtex4-2}

\usepackage{ifthen}
\usepackage{iftex}

\usepackage{adjustbox}
\usepackage{theorem}
\usepackage{ascmac}
\usepackage{bbm}
\usepackage{bm}
\usepackage{braket}
\usepackage{dcolumn}
\usepackage{enumerate}
\usepackage{enumitem}
\usepackage{framed}
\usepackage{tikz}
\usetikzlibrary{arrows.meta,calc}
\usepackage{hyperref}
\usepackage{longtable}
\usepackage{mathtools}
\usepackage{multirow}
\usepackage{stackengine}
\usepackage{txfonts}
\usepackage{xparse}

\usepackage{color}
\usepackage[most]{tcolorbox}

\newcommand{\hypercolor}{blue}
\hypersetup{
  colorlinks,
  pdftitle={Two Operational Principles Single Out Quantum Theory},
  pdfauthor={Kenji Nakahira},
  citecolor=\hypercolor,
  linkcolor=\hypercolor,
  urlcolor=\hypercolor,
  bookmarksopen,
  bookmarksopenlevel=4,
  bookmarksnumbered
}

\usepackage{settings_common_paper}
\usepackage{settings}

\newtheorem{fpscounterexample}{Counterexample}

\newcommand{\titlename}{Two Operational Principles Single Out Quantum Theory}

\begin{document}

\title{\titlename}
\author{Kenji Nakahira}
\affiliation{%
 Quantum Information Science Research Center, Quantum ICT Research Institute, Tamagawa University,
 Machida, Tokyo 194-8610, Japan
}

\begin{abstract}
 Within a finite-dimensional causal operational-probabilistic theory (OPT), we show that adding just two postulates,
 local equivalence and purification, suffices to derive the mathematical structure of quantum theory as a whole,
 without separately assuming quantum rules for states, measurements, composites, or transformations.
 This resolves the question raised by Chiribella, D'Ariano, and Perinotti (CDP)
 [Phys.~Rev.~A \textbf{81}, 062348 (2010)].
 The key step is to establish homogeneity of state cones by combining infinitesimal cone-preserving deformations
 obtained from these postulates with geometric optimization.
 CDP purification requires every normalized state to be purifiable and purifications of the same state
 on any common ancilla to be related by a reversible channel on that ancilla.
 Our alternative assigns each system one operationally equivalent ancilla that purifies all its normalized states,
 requiring such reversible equivalence only there.
 Together with local equivalence,
 this condition also yields quantum theory.
\end{abstract}

\maketitle

The longstanding goal of quantum-theory reconstruction is to derive its mathematical structure
from a few clear and intuitive physical principles.
Special relativity provides a benchmark: Einstein's relativity principle and the constancy of light
speed have clear physical meanings yet yield rich kinematics.
By contrast, the standard axioms of quantum theory prescribe precise mathematical structures:
Hilbert spaces for systems, density matrices
$\rho$ for states, Born-rule probabilities $\Tr(E\rho)$ for positive operators $E$,
positive-operator-valued measures for measurements, tensor products for composites, and
completely positive trace-preserving maps for channels
\cite{Neu-1932,Mac-1963,Cho-1975,Bus-Lah-1996,Nie-Chu-2000}.
Yet these rules tell us how to use quantum theory, not why its states are described by
complex density matrices.
Can this state-space structure, the Born rule, tensor products,
and complete positivity all be derived from more basic operational principles?
This question is not new: von Neumann himself later wrote that he no longer believed
absolutely in Hilbert space
\cite{Neu-Letters-2005}.
The issue is not quantum theory's empirical success, but whether its mathematical structure follows from
intuitively understandable principles.

Quantum-theory reconstructions draw on operational or information-theoretic postulates
\cite{Har-2011,Har-2013,Har-2016,Mas-Mul-2011,DeLaTorre-Mas-Sho-Mul-2012} and
conditions involving Jordan algebras and filters
\cite{Wil-2012,Wil-2012-axioms,Bar-Mul-Udu-2014,Wet-2018}.
Categorical and diagrammatic methods describe quantum processes
\cite{Abr-Coe-2004,Coe-2010,Coe-2014,Coe-Kis-2017,Sel-Sca-Coe-2021}, and quantum logic offers
another axiomatic approach \cite{Bir-Neu-1936}.
Purification and local tomography have also been used to constrain alternatives to the Born rule
\cite{Gal-Mas-2017,Gal-Mas-2018}.
Operational-probabilistic theories (OPTs) offer a natural framework \cite{Pla-2021}:
without assuming Hilbert spaces, they treat states, transformations, measurements, and measurement-outcome
probabilities as basic concepts and ask which physical principles single out quantum theory
\cite{Har-2001,Bar-2007,Dak-Bru-2009,Chi-Dar-Per-2010,Har-2011,Mas-Mul-2011,Chi-Dar-Per-2011,Wil-2012,
 Wil-2012-axioms,DeLaTorre-Mas-Sho-Mul-2012,Har-2013,Bar-Gae-Wil-2013,Jan-Lal-2013,Bar-Mul-Udu-2014,
 Har-2016,Dar-Chi-Per-2017,Wri-2021,Bar-Udu-Wet-2023,Mul-2021,Gia-Chi-2026}.
In a purification-based approach to reconstruction, Chiribella, D'Ariano, and Perinotti (CDP)
proposed a postulate that we call \termdef{CDP purification}: every normalized state admits a
purification, and such purifications are essentially unique.
In 2010, they derived far-reaching consequences from this postulate and local discriminability---
the determination of composite states from local-measurement statistics---and asked whether
these two postulates alone characterize quantum theory \cite{Chi-Dar-Per-2010}.
In 2011, they reconstructed quantum theory using three further postulates \cite{Chi-Dar-Per-2011}.
Later work gave partial results \cite{Chi-Sca-2015,Bar-Lee-Sca-Sel-2017,Gal-Mas-2018,Lee-Sel-2018},
but the two-postulate question remained open.

Here we resolve this question using the finite-dimensional causal OPT framework of
Refs.~\cite{Chi-Dar-Per-2010,Chi-Dar-Per-2011,Dar-Chi-Per-2017}.
We add only local discriminability and CDP purification, as in Ref.~\cite{Chi-Dar-Per-2010}
(Fig.~\ref{fig:main}); we express the former as the equivalent postulate of
\termdef{local equivalence}.
Unlike representative reconstructions using four to six postulates, ours adds only these two
to the OPT framework, without separately assuming mathematical conditions such as
Jordan-algebraic structure, state-cone self-duality, or measurement no-restriction.
A key step is to derive homogeneity of state cones from these postulates.
Together with other consequences of the postulates, this leads, through the Jordan-algebraic classification,
to the familiar state spaces of complex density matrices.
The remaining steps recover the quantum rules for composites, measurements, and transformations.
CDP purification requires essential uniqueness: any two purifications of the same state
on a common purifying system are related by a reversible channel on that system.
This requirement applies to every common purifying system on which such purifications exist,
and may therefore seem demanding.
We therefore introduce a new postulate, fixed-purifying-system (FPS) purification.
For each system $A$, it designates a system $A'$ on which every normalized state of $A$
admits a purification.
Essential uniqueness is required only for purifications on this designated $A'$:
any two purifications of the same state on $AA'$ are related by a reversible channel on $A'$.
A counterexample shows that local equivalence and FPS purification do not suffice to
characterize quantum theory.
We also introduce equivalent-system (ES) purification, a new postulate requiring $A'$ to be
equivalent to $A$.
Local equivalence and ES purification together yield quantum theory.

\begin{figure}[t]
 \centering
 \includegraphics{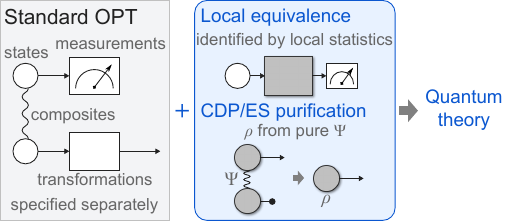}
 \caption{An OPT specifies states, measurements, composite systems, and transformations.
 Local equivalence and CDP purification (Theorem~\ref{thm:CDP-quantum}) together determine all four
 structures to be those of quantum theory.
 The same conclusion holds for ES purification (Theorem~\ref{thm:ES-quantum}), introduced below.}
 \label{fig:main}
\end{figure}

\noindent\textit{OPT and the two postulates---}
An OPT describes experiments involving preparations, transformations, and measurements, together with
probabilities of their outcomes, without assuming Hilbert spaces.
Here we briefly introduce the basic notions and notation used in the main text; details are given
in Appendix~\ref{sm:sec:OPT}.
A transformation corresponds to an individual experimental outcome and represents the physical process
associated with that outcome.
In quantum theory, transformations from an $n_A$-level system $A$ to an $n_B$-level system $B$
are represented by completely positive trace-nonincreasing maps from
$\Complex^{n_A \times n_A}$ to $\Complex^{n_B \times n_B}$.
The set of transformations from a system $A$ to a system $B$ is denoted by $\Trans(A,B)$.
Let $I$ be a distinguished system; elements of $\St(A) \coloneqq \Trans(I,A)$ are called \termdef{states} of $A$,
and elements of $\Eff(A) \coloneqq \Trans(A,I)$ are called \termdef{effects} of $A$.
States correspond to particular outcomes of preparation experiments, while effects correspond to particular
outcomes of measurement experiments.
The set $\Trans(I,I)$ is identified with the real interval $[0,1]$, whose elements represent probabilities.
In quantum theory, states of $A$ are represented by positive semidefinite $n_A \times n_A$ matrices
of trace at most one, while effects are represented by maps $X \mapsto \tr(EX)$,
where both $E$ and $I_{n_A}-E$ are positive semidefinite.
Here $I_{n_A}$ denotes the identity matrix of order $n_A$.

Transformations can be composed in parallel: $f \in \Trans(A,B)$ and $g \in \Trans(C,D)$ give $f \ot g \in
\Trans(AC,BD)$, where $AC$ denotes the composite of $A$ and $C$.
When $B = C$, they can also be composed sequentially, giving $g \c f \in \Trans(A,D)$.
When $f$ is a state, we also write $g(f)$ for their sequential composition.
Each system $A$ has a unique effect $e_A$ combining all possible measurement outcomes.
A state $\rho$ satisfying $e_A(\rho) = 1$ is called \termdef{normalized}; the set of all such states is denoted
by $\StN(A)$.
In quantum theory, $e_A$ is the trace functional and normalized states are density matrices.
A transformation taking every normalized state to a normalized state is called a \termdef{channel}, and a channel
possessing an inverse channel is \termdef{reversible}.
The identity channel on $A$ is written $\id_A$.
For a normalized state $\rho$ and an effect $a$, the number $a(\rho)$ is the probability of the
corresponding measurement outcome.
A normalized state is \termdef{pure} if it cannot be written as a nontrivial mixture of distinct normalized states;
the set of such states of $A$ is $\StNP(A)$.
The convex cone and real vector space spanned by $\Trans(A,B)$ are denoted by $\Trans_+(A,B)$ and
$\Trans_\Real(A,B)$, respectively.
In particular, $\St_+(A) \coloneqq \Trans_+(I,A)$ and $\St_\Real(A) \coloneqq \Trans_\Real(I,A)$ are the
cone and real vector space spanned by $\St(A)$.

The first requirement says that channels are identified by experimentally accessible input-output statistics.
Intuitively, local equivalence says that local input-output statistics completely determine a channel.
\begin{postulate}[Local equivalence] \label{postulate:LE}
 For any systems $A$ and $B$, if two channels $\Lambda$ and $\cE$ from $A$ to $B$ satisfy
 $a(\Lambda(\rho)) = a(\cE(\rho))$ for all $\rho \in \StN(A)$ and all $a \in \Eff(B)$,
 then $\Lambda = \cE$.
\end{postulate}
If this requirement failed, then two devices with identical statistics in all direct input-output experiments could
nevertheless behave differently when connected to an external system.
Then identifying a channel would require probing correlations with an environment, rather than only the
input-output behavior of the device itself.
Local equivalence rules out such locally invisible differences and makes channel identity directly accessible in
the laboratory.
Local equivalence is equivalent to \termdef{local discriminability}:
for any two distinct normalized states of a composite system $\AB$,
some product effect $a \ot b$, with $a \in \Eff(A)$ and $b \in \Eff(B)$, gives different probabilities
\cite{Nak-2020}.

The second postulate is \termdef{CDP purification}.
Purification asks whether the mixedness of a state can be understood by regarding it as the marginal of a
larger normalized pure state.
Formally, a purification of $\rho \in \StN(A)$ is a normalized pure state $\Psi \in \StNP(AB)$ satisfying
$(\id_A \ot e_B)(\Psi) = \rho$, with $B$ called the \termdef{purifying system}.
In quantum theory, the channel $\id_A \ot e_B$ is represented by the partial trace over $B$.
CDP purification further demands essential uniqueness for every common purifying system
\cite{Chi-Dar-Per-2010}.

\begin{postulate}[CDP purification] \label{postulate:CDP}
 Every normalized state of every system $A$ has a purification.
 Moreover, for any system $B$, any two purifications $\Psi,\Phi \in \StNP(AB)$ of the same state are
 related by a reversible channel $U$ on $B$: $\Phi = (\id_A \ot U)(\Psi)$.
\end{postulate}

Neither postulate alone suffices: classical theory satisfies local equivalence but not CDP purification,
whereas real quantum theory satisfies CDP purification but not local equivalence
\cite{Stu-1960,Ale-Bor-Woo-2013}.

\noindent\textit{Main theorem---}
We now state our main theorem and outline its proof, using quantum-theoretic examples to convey the idea
without presupposing familiarity with OPTs.
An outline of the proofs appears after the references, with complete arguments in
Appendices~\ref{sm:sec:CDP_ES_basic}--\ref{sm:sec:simple_complex}.

\begin{thm} \label{thm:CDP-quantum}
 An OPT satisfying local equivalence and CDP purification is quantum theory.
\end{thm}

\begin{proof}

 In quantum theory, for a system $A$ with $n_A$ levels,
 $\StN(A)$ is isomorphic to the set of $n_A \times n_A$ density matrices.
 The corresponding state cone $\St_+(A)$ is the positive cone of a Euclidean Jordan algebra (EJA)
 \cite{Jor-Neu-Wig-1934}.
 The positive cones of EJAs, in turn, are known to fall into several types
 \cite{Koe-1957,Vin-1960}.
 Several approaches to deriving quantum theory seek to recover EJA structure and make use of this classification
 \cite{Ara-1980,Bar-Wil-2014,Wil-2018,Wet-2019,Bar-Gra-Wil-2020},
 and we adopt this approach.
 Suppose first that each state cone $\St_+(A)$ is homogeneous.
 A cone is \termdef{homogeneous} if its interior is nonempty and any two of its interior points are connected
 by a \termdef{cone automorphism}, that is, an invertible linear map preserving the cone.
 This is a strong geometric condition that is not satisfied by general convex cones
 \footnote{For example, even in three dimensions, the general OPT framework admits any closed and
 pointed convex cone spanning the space as a state cone.
 By contrast, homogeneous state cones are restricted, up to linear isomorphism, to just two types:
 the simplicial cone $\{ (x,y,z) \in \Real^3 \mid x,y,z \ge 0 \}$ and
 the Lorentz cone $\{ (x,y,z) \in \Real^3 \mid \sqrt{x^2 + y^2} \le z \}$
 \cite{Kan-Tsu-1974}}.
 By Ref.~\cite{Bar-Udu-Wet-2023}, homogeneity and the existence of reversible transformations
 between any two pure states imply that every $\St_+(A)$ is the positive cone of an EJA.
 With purification and the structure of composite systems, the EJA classification implies that,
 for some positive integer $r_A$, $\StN(A)$ is isomorphic to the set of $r_A \times r_A$ density matrices.
 The final step recovers the quantum rules for composite systems and the full set of quantum transformations,
 drawing on Ref.~\cite{Chi-Dar-Per-2010}.
 Both local equivalence and CDP purification are used in each of these steps.
 It therefore remains to prove homogeneity.
 This is perhaps the most difficult part of the argument, and we focus on it below.

 Fix a system $A$, and let $\cA_A$ be the set of cone automorphisms of $\St_+(A)$.
 To establish homogeneity, it suffices to construct, for every interior point $z$ of $\St_+(A)$,
 an automorphism $g_z \in \cA_A$ mapping $z$ to a common interior point.
 Indeed, for any two interior points $x$ and $y$, the automorphism $g_y^{-1} \c g_x \in \cA_A$
 maps $x$ to $y$.
 A natural target is the centroid $\omega_A$ of the normalized state space $\StN(A)$.
 For example, in quantum theory, given an $n_A \times n_A$ positive-definite matrix $\rho$,
 set $W \coloneqq (n_A\rho)^{-1/2}$ and $g_\rho(Z) \coloneqq WZW^\dag$.
 Then $g_\rho$ is a cone automorphism that maps $\rho$ to the centroid
 $\omega_A = I_{n_A}/n_A$.
 But for a general OPT the matrix operation $(n_A\rho)^{-1/2}$ is unavailable,
 so a different construction is needed.
 The guiding idea comes from convex geometry.
 For an interior point $z$, consider bounded sections $\cN$ of the cone by hyperplanes passing through $z$.
 Among these sections, a section minimizing the volume $V_\cN$ of the region $\{ p x \mid 0 \le p \le 1,~x
 \in \cN \}$ between the section and the origin has its centroid at $z$ \cite[Sec.~4]{Izm-2023}.
 If an automorphism $g$ carries that section to $\StN(A)$, it also carries centroids to centroids and hence
 satisfies $g(z) = \omega_A$ (Fig.~\ref{fig:cone}).
 The quantum automorphism $g_\rho$ above has precisely this property, but such an automorphism
 need not exist in general.
 Instead, consider cone automorphisms $g$ for which the section
 $\cN_g \coloneqq \{ v \in \St_+(A) \mid e_A(g(v)) = 1 \}$ contains $z$.
 Minimize the associated volume $V_{\cN_g}$, and let $g_z$ be an optimizing automorphism.
 Write $\omega_z \coloneqq g_z(z)$.
 Showing that $\omega_z$ coincides with the centroid $\omega_A$ would suffice.
 In fact, as noted above, it is enough to show that $\omega_z$ is a common interior point independent of $z$.

 \begin{figure}[t]
 \centering
 \includegraphics{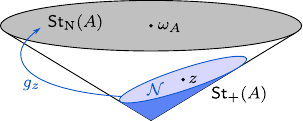}
 \caption{The state cone $\St_+(A)$, the section $\StN(A)$ of normalized states,
 and a section $\cN$ through an interior point $z$.
 The figure schematically depicts the case in which the volume of the region
 $\{ p x \mid 0 \le p \le 1, ~x \in \cN \}$ between the section and the origin
 is minimal among bounded sections through $z$; in this case, the centroid of $\cN$ is $z$.
 If, for every interior point $z$, there exists a cone automorphism $g_z$
 mapping such a section $\cN$ onto $\StN(A)$,
 then $\omega_z \coloneqq g_z(z)$ coincides with the centroid $\omega_A$ of $\StN(A)$,
 and $\St_+(A)$ is homogeneous.}
 \label{fig:cone}
\end{figure}

 The following infinitesimal condition suffices to show that $\omega_z$ is independent of $z$:
 for every $a$ in the dual space $\St_\Real(A)^*$, there exist
 an element $X_a \in \Trans_\Real(A,A)$,
 elements $P_{n,\varepsilon} \in \Trans_\Real(A,A)$
 for both $\varepsilon \in \{ -1,1 \}$, and positive numbers $s_n$ tending to zero such that
 each $P_{n,\varepsilon}$ preserves the state cone $\St_+(A)$ and
 \begin{alignat}{1}
  e_A \c X_a = a, \quad \frac{P_{n,\varepsilon} - \id_A}{s_n} \to \varepsilon X_a
 \end{alignat}
 as $n \to \infty$.
 Indeed, for $t \ge 0$,
 $e^{t(P_{n,\varepsilon} - \id_A)/s_n} \to e^{\varepsilon tX_a}$.
 Since each $P_{n,\varepsilon}$ preserves the cone, so does the exponential on the left; consequently both
 $e^{tX_a}$ and $e^{-tX_a}$ preserve the cone and are inverse to one another.
 Thus $e^{tX_a} \in \cA_A$ for every real $t$.
 Fix $a$ and define $\lambda_z(t) \coloneqq (e_A \c e^{tX_a})(\omega_z)$.
 Then $\lambda_z'(0) = (e_A \c X_a)(\omega_z) = a(\omega_z)$.
 The minimizing property of $g_z$ and the volume calculation together imply that $\lambda_z'(0)$ is
 independent of $z$.
 Hence $a(\omega_z)$ is independent of $z$ for every $a \in \St_\Real(A)^*$, forcing $\omega_z$ itself
 to be independent of $z$.
 The quantum example makes the infinitesimal condition explicit.
 Every $a \in \St_\Real(A)^*$ has the form $a(\rho) = \tr(H_a\rho)$ for some Hermitian matrix $H_a$.
 Define $X_a(\rho) \coloneqq (H_a\rho + \rho H_a)/2$; then $\tr(X_a(\rho)) = a(\rho)$.
 For $s_n \to 0$, setting
 $P_{n,\varepsilon}(\rho) \coloneqq \sqrt{I_{n_A} + \varepsilon s_nH_a}\,\rho\,
 \sqrt{I_{n_A} + \varepsilon s_nH_a}$
 gives the required two-sided limit for sufficiently large $n$.
 Their limiting automorphisms are $e^{tX_a}(\rho) = e^{tH_a/2}\rho e^{tH_a/2}$.
 In a general OPT, however, explicitly constructing $X_a$ and $P_{n,\varepsilon}$ for every functional $a$
 is far from straightforward.

 As detailed in Appendix~\ref{sm:sec:homogeneity},
 rather than constructing $X_a$, $P_{n,\varepsilon}$, and $s_n$ directly
 for each system $A$ and each $a \in \St_\Real(A)^*$, it suffices to find a single system $C$,
 an element $Y \in \Trans_\Real(C,C)$, elements $T_{n,\varepsilon} \in \Trans_+(C,C)$
 for $\varepsilon \in \{ -1,1 \}$, and positive numbers $r_n$ tending to zero, such that
 $e_C \c Y$ is nonconstant on $\StN(C)$ and
 $(T_{n,\varepsilon} - \id_C)/r_n \to \varepsilon Y$.
 Here $Y$, $T_{n,\varepsilon}$, and $r_n$ correspond to
 $X_a$, $P_{n,\varepsilon}$, and $s_n$, respectively.
 Once these have been obtained, the proof proceeds as follows.
 For each $a \in \St_\Real(A)^*$, one can construct a cone-preserving real-linear map
 $\Gamma_a \colon \Trans_\Real(C,C) \to \Trans_\Real(A,A)$, a real number $c_e$, and a
 positive number $c_a$ such that
 \begin{alignat}{1}
  e_A \c \Gamma_a(Y) = c_e e_A + c_a a, \quad
  \frac{\Gamma_a(T_{n,\varepsilon}) - \id_A}{r_n} \to \varepsilon\Gamma_a(Y).
 \end{alignat}
 By defining $X_a$ and $P_{n,\varepsilon}$ as suitable affine transformations of
 $\Gamma_a(Y)$ and $\Gamma_a(T_{n,\varepsilon})$, respectively,
 we can show that they satisfy the conditions stated above.
 This completes the route to homogeneity.

\end{proof}

\noindent\textit{FPS purification---}
To some readers, the uniqueness condition in CDP purification may seem particularly demanding.
For each normalized state $\rho \in \StN(A)$, it must hold for every system $B$ on which $\rho$
admits a purification.
A naive experimental verification would therefore require examining all possible purifying systems,
which could be a formidable task.
In quantum theory, these include all systems $B$ with at least as many levels as $A$.
This motivates the following alternative.

\begin{postulate}[Fixed-purifying-system (FPS) purification] \label{postulate:FPS}
 For each system $A$, there exists a designated system $A'$ such that every $\rho \in \StN(A)$ has a
 purification on $AA'$.
 Moreover, any two purifications of the same state on this pair are related by a reversible channel on $A'$.
\end{postulate}

Unlike CDP purification, FPS purification designates for each system $A$ one system $A'$ on which
every normalized state of $A$ admits a purification, although other purifying systems may also exist.
Uniqueness is required only for purifications on the designated system $A'$.
This is a natural way to fix a purifying system, but local equivalence and FPS purification do not
suffice to derive quantum theory.

\begin{fpscounterexample} \label{thm:FPS-counterexample}
 There exists an OPT satisfying local equivalence and FPS purification that is not quantum theory.
\end{fpscounterexample}

\begin{proof}
 We sketch the counterexample; see Appendix~\ref{sm:sec:FPS} for details.
 As in quantum theory, for each system $A$, its normalized state space $\StN(A)$
 is isomorphic to the set of $n_A \times n_A$ density matrices,
 where $n_A$ is a positive integer, and $n_\AB = n_A n_B$ holds
 for all systems $A$ and $B$.
 However, unlike quantum theory, a composite system $AB$ admits no nonseparable effects.
 For example, when $n_A = n_B = 2$, let $\ket{\Phi} \coloneqq (\ket{00} + \ket{11})/\sqrt2$.
 Quantum theory admits an effect on $AB$ represented by the functional
 $\rho \mapsto \braket{\Phi|\rho|\Phi}$, whereas this OPT does not.
 As verified in Appendix~\ref{sm:sec:FPS},
 an OPT satisfying local equivalence and FPS purification can be constructed with
 these properties.

\end{proof}

\noindent\textit{ES purification---}
Since FPS purification does not suffice to derive quantum theory, we now add the requirement that its
designated purifying system be equivalent to the original system.
Here two systems $X$ and $Y$ are \termdef{equivalent} if there exists a reversible channel from $X$ to $Y$.
This leads to the following postulate.
\begin{postulate}[ES purification] \label{postulate:ES}
 For each system $A$, there exists a system $\tA$ equivalent to $A$ such that every $\rho \in \StN(A)$
 has a purification on $A\tA$.
 Moreover, any two purifications of the same state on this pair are related by a reversible channel on $\tA$.
\end{postulate}
Compared with CDP purification, ES purification requires the existence of an equivalent purifying system but
relaxes uniqueness to that designated system.
Ref.~\cite{Chi-Dar-Per-2010} shows, using CDP purification, that irreversibility can be
understood as an apparent loss of information when an environment is discarded, rather than as erasure of information.
More precisely, in an OPT satisfying local equivalence, CDP purification is equivalent to a generalized
Stinespring-dilation property: every channel $N \in \Trans(A,B)$ can be realized using a reversible channel
$W \in \Trans(EA,FB)$, a normalized pure state on the input ancillary system $E$, and
the deterministic effect $e_F \in \Eff(F)$, with essential uniqueness of the realization.
An analogous equivalence holds for ES purification, with the output environment $F$ chosen solely
from the input and output systems $A,B$ (see Sec.~\ref{sm:subsec:CDP_ES_basic-equivalence}).
Thus ES purification shares the conservation-of-information interpretation of CDP purification
\cite{Chi-Dar-Per-2010}: irreversibility arises by discarding an unobserved environment, not by erasing information.
As with CDP purification, ES purification yields the basic properties needed to establish,
among other results, homogeneity and the density-matrix characterization of state spaces.
The following theorem can therefore be proved by the same method as for CDP purification
(see Appendices~\ref{sm:sec:CDP_ES_basic}--\ref{sm:sec:simple_complex}).

\begin{thm} \label{thm:ES-quantum}
 An OPT satisfying local equivalence and ES purification is quantum theory.
\end{thm}

\noindent\textit{Conclusions---}
Within the stated finite-dimensional causal OPT framework, we have shown that local equivalence and CDP
purification suffice to derive quantum theory.
This resolves the question raised by CDP in 2010 \cite{Chi-Dar-Per-2010} and shows that the
three additional postulates used in their 2011 reconstruction \cite{Chi-Dar-Per-2011} are not indispensable
for this derivation.
The central step is the derivation of homogeneity of each state cone from these two postulates.
Specifically, by combining two-sided cone-preserving perturbations of the identity with the minimization
of section volumes, we show that the image $\omega_z$ of an interior point $z$ under an optimizing map
is independent of $z$.
We have also studied weaker uniqueness requirements for purification.
FPS purification, which imposes uniqueness only on a designated purifying system, does not suffice, as our
counterexample demonstrates.
By contrast, ES purification additionally requires the designated purifying system to be equivalent to
the original system and, together with local equivalence, does imply quantum theory.
The significance of this result is that, within the stated OPT framework, quantum theory can be
characterized by just two postulates: local equivalence and either CDP or ES purification.
This provides a unified understanding of its various mathematical rules as consequences of a small number of
operational principles rather than independent assumptions.

We thank O.~Hirota, M.~Sohma, T.~S.~Usuda, K.~Kato, and M.~Erba for insightful discussions.
The author used ChatGPT \cite{OpenAI-GPT6} to assist with checking the logical consistency of proofs
and improving the clarity of the manuscript; all output was independently verified by the author.

\bibliographystyle{apsrev4-2}

\clearpage
\section*{Proof outline}

\noindent%
For readers familiar with operational-probabilistic theories (OPTs),
quantum theory reconstruction, and Euclidean Jordan algebras (EJAs),
we outline the proofs of Theorems~\ref{thm:CDP-quantum} and~\ref{thm:ES-quantum}.
Complete proofs are provided in Appendices~\ref{sm:sec:CDP_ES_basic}--\ref{sm:sec:simple_complex}.

We first define three properties used in the proofs.
The theory satisfies \termdef{pure-state transitivity} if, for every system $A$,
any two normalized pure states of $A$ are related by a reversible channel on $A$.
For $x, y \in \St_\Real(A)$, let $x \le y$ mean $y - x \in \St_+(A)$.
The theory satisfies \termdef{refinement realization} if, for any systems $A$, $B$, and $C$,
every $\omega \in \StN(A) \cap \Int \St_+(A)$, purification $\Psi \in \StNP(AB)$ of $\omega$,
and state $\sigma \in \St(CB)$ satisfying $(e_C \ot \id_B)(\sigma) \le (e_A \ot \id_B)(\Psi)$,
there exists a transformation $K \in \Trans(A,C)$ such that $(K \ot \id_B)(\Psi) = \sigma$.
The theory has the \termdef{quantum state-space property} if, for each system $A$,
there is a positive integer $r_A$ for which $\StN(A) \cong \Den_{r_A}$.
Here $\Den_{r_A}$ is the set of $r_A \times r_A$ density matrices,
and $\cong$ denotes an isomorphism of convex sets.
Here \termdef{quantum theory} means an OPT whose state spaces admit
linear density-matrix representations of sizes $n_A$, with
$n_{AB}=n_A n_B$, and whose transformations are exactly
the completely positive trace-nonincreasing (CPTNI) maps in these
representations (Definition~\ref{sm:def:quantum_theory}).
Write $\St_+(A)^*$ for the dual cone of $\St_+(A)$
and $\St_\Real(A)^*$ for the dual space of $\St_\Real(A)$.
$\cA_A$ denotes the set of cone automorphisms of $\St_+(A)$.

\begin{figure}[!b]
 \centering
 \begin{tikzpicture}[>=Latex, every node/.style = {font = \footnotesize},
   block/.style = {draw, rounded corners = 1pt, align = center, inner sep = 1pt,
    text width = 19mm, minimum height = 5mm},
   arr/.style = {-{Latex[length=1.7mm,width=1.2mm]}, semithick}]
  \node[block, text width = 33mm] (root) at (0,0) {CDP/ES purification};
  \node[block, text width = 19mm] (hom) at (-3.20,-1.55) {homogeneity};
  \node[block, minimum height = 9mm] (pst) at (-1.07,-1.55) {pure-state\\transitivity};
  \node[block, minimum height = 9mm] (pur) at (1.07,-1.55) {existence of\\purification};
  \node[block, minimum height = 9mm] (refine) at (3.20,-1.55) {refinement\\realization};
  \node[block, text width = 33mm] (sym) at (-2.13,-3.10) {symmetric cones [E1]};
  \node[block, text width = 41mm] (den) at (0,-4.30) {quantum state-space property};
  \node[block, text width = 33mm] (qt) at (2.25,-5.50) {quantum theory [E2]};
  \draw[arr] ($(root.south) + (-9mm,0)$) -- (hom.north);
  \draw[arr] ($(root.south) + (-3mm,0)$) -- (pst.north);
  \draw[arr] ($(root.south) + (3mm,0)$) -- (pur.north);
  \draw[arr] ($(root.south) + (9mm,0)$) -- (refine.north);
  \draw[arr] (hom.south) -- ($(sym.north) + (-6mm,0)$);
  \draw[arr] (pst.south) -- ($(sym.north) + (6mm,0)$);
  \draw[arr] (sym.south) -- ($(den.north) + (-8mm,0)$);
  \draw[arr] (pur.south) -- (pur.south |- den.north);
  \draw[arr] ($(den.south) + (5mm,0)$) -- ($(qt.north) + (-6mm,0)$);
  \draw[arr] (refine.south) -- (refine.south |- qt.north);
 \end{tikzpicture}
 \caption{Proof roadmap for deriving quantum theory from CDP or ES purification
 in OPTs satisfying local equivalence.
 Arrows indicate logical dependencies; multiple incoming arrows signify that
 the corresponding premises are used jointly.}
 \label{em:fig:proof-roadmap}
\end{figure}

We assume local equivalence in the finite-dimensional causal OPT framework of
Refs.~\cite{Chi-Dar-Per-2010,Chi-Dar-Per-2011,Dar-Chi-Per-2017}
(see Appendix~\ref{sm:sec:OPT}).
Either CDP or ES purification implies \termdef{existence of purification}
(every normalized state has a purification), pure-state transitivity,
and refinement realization.
Figure~\ref{em:fig:proof-roadmap} summarizes the remaining steps, which use the following two lemmas.

\begin{lemma}[E1] \label{em:lemma:E1}
 If every state cone $\St_+(A)$ is homogeneous and the theory satisfies
 pure-state transitivity, then every state cone is symmetric.
\end{lemma}
\begin{proof}
 Reversible channels are normalized order isomorphisms, so homogeneity and
 pure-state transitivity give self-duality by Theorem~25 of
 Ref.~\cite{Bar-Udu-Wet-2023}, and thus symmetry.
\end{proof}

\begin{lemma}[E2] \label{em:lemma:E2}
 The quantum state-space property and refinement realization together imply quantum theory.
\end{lemma}
\begin{proof}
 Local tomography and the quantum state-space property give $r_{AB}=r_A r_B$.
 Purity of product states constrains the matrix representations
 of composites to tensor-product forms~\cite{Xu-Zhe-Fos-2015}; associativity and reversible swaps
 ensure their consistency (Lemma~\ref{sm:lemma:complex_monoidal_tensorator}).
 The Choi correspondence characterizes CPTNI maps by bipartite states
 $\sigma$ satisfying $(e_B \ot \id_A)(\sigma) \le \beta_A$,
 where $\beta_A$ is the maximally mixed state of $A$.
 Refinement realization supplies all such maps; conversely, the other
 outcomes of a test containing any transformation enforce the
 same marginal bound (Appendix~\ref{sm:sec:E1_E2}).
\end{proof}

By the Koecher--Vinberg theorem~\cite{Koe-1957,Vin-1960}, every symmetric state cone
is an EJA cone.
We choose the EJA structure so that its trace is $e_A$~\cite{Far-1994};
thus normalized pure states are primitive idempotents.
Write $u_A$ for its unit, $\b$ for its Jordan product,
and $\braket{x,y}=\tr_A(x\b y)$ for its trace inner product.

\begin{proposition} \label{em:pro:quantum_state_space_property}
 If every state cone is symmetric and every normalized state admits a purification,
 then the quantum state-space property holds.
\end{proposition}
\begin{proof}
 We first rule out a nontrivial direct-sum decomposition of $\St_+(A)$.
 Suppose, for contradiction, that $\St_+(A) = S_1 \oplus S_2$ with nonzero cones $S_i$.
 Choose a normalized state $\chi$ with positive weight in both summands
 and a purification $\Psi \in \StNP(AB)$ of $\chi$.
 Let $\pi_1$ project onto $S_1$, and set
 $h \coloneqq (e_A \c \pi_1) \ot e_B$ and $t \coloneqq h(\Psi)$.
 Then $h, e_\AB - h \in \St_+(\AB)^*$ and $0 < t < 1$.
 By self-duality, $h(w) = \braket{q,w}$ for some $\zero \le q \le u_\AB$.
 For $x_i \in S_i$ and $y \in \St_+(B)$,
 $h(x_1 \ot y) = e_\AB(x_1 \ot y)$ and $h(x_2 \ot y) = 0$.
 Product-state spanning gives $\Psi = z_1 + z_2$, where each $z_i$ is a linear combination of products
 with $A$ factor in $S_i$; moreover, $z_1 \b z_2 = \zero$, $q \b z_1 = z_1$,
 and $q \b z_2 = \zero$ (see Lemma~\ref{sm:lemma:EJA-effect_separation_orthogonality}).
 As $\Psi$ is a primitive idempotent,
 $\Psi = \Psi^2 = z_1^2 + z_2^2$, with $z_i^2 \in \St_+(AB)$.
 Associativity of the trace inner product yields
 $h(z_1^2) = e_\AB(z_1^2)$ and $h(z_2^2) = 0$, hence $e_\AB(z_1^2) = t$.
 Since $\zero \le z_1^2 \le \Psi$, purity gives $z_1^2 = t\Psi$,
 so $t = h(z_1^2) = th(\Psi) = t^2$, contradicting $0 < t < 1$.

 Let $r_A$ denote the EJA rank and $d_A \coloneqq \dim \St_\Real(A)$.
 Local equivalence (equivalently, local tomography) gives $d_\AB = d_A d_B$.
 By Lemma~\ref{sm:lemma:jordan_frame_product},
 products of Jordan frames form a Jordan frame of the composite,
 giving $r_\AB = r_A r_B$.
 Thus $AA$ has rank $r_A^2$ and dimension $d_A^2$.
 Comparing these relations with the classification of simple EJAs
 leaves only complex matrix state spaces, proving the quantum state-space property
 (Appendix~\ref{sm:sec:simple_complex});
 see also Refs.~\cite{Ara-1980,Wet-2019,Nie-2020,Sel-Sca-Coe-2021}.
\end{proof}

By Lemmas~\ref{em:lemma:E1} and~\ref{em:lemma:E2} and
Proposition~\ref{em:pro:quantum_state_space_property},
it remains to establish homogeneity.
We begin with the following lemma; see Lemma~\ref{sm:lemma:homogeneity-nonconstant_two_sided_filter}
for its proof.

\begin{lemma} \label{em:lemma:homogeneity-nonconstant_two_sided_filter}
 There exist a system $C$, an element $Y \in \Trans_\Real(C,C)$,
 a sequence $\{ T_{n,\varepsilon} \in \Trans_+(C,C) \}_{n=1}^\infty$
 for each $\varepsilon \in \{ -1,1 \}$, and a sequence of positive numbers $\{ r_n \}_{n=1}^\infty$
 such that $e_C \c Y$ is nonconstant on $\StN(C)$,
 $(T_{n,\varepsilon} - \id_C)/r_n \to \varepsilon Y$, and $r_n \to 0$ as $n \to \infty$.
\end{lemma}

We first establish two more lemmas.

\begin{lemma} \label{em:lemma:two_sided_positive_tangent}
 For every system $A$ and $a \in \St_\Real(A)^*$,
 there exist $X_a \in \Trans_\Real(A,A)$,
 a sequence $\{ P_{n,\varepsilon} \in \Trans_+(A,A) \}_{n=1}^\infty$
 for each $\varepsilon \in \{ -1,1 \}$, and a sequence of positive numbers $\{ s_n \}_{n=1}^\infty$
 such that $e_A \c X_a = a$, $s_n \to 0$, and
 $(P_{n,\varepsilon} - \id_A)/s_n \to \varepsilon X_a$ as $n \to \infty$.
\end{lemma}
\begin{proof}
 Let $C$, $Y$, $T_{n,\varepsilon}$, and $r_n$ be as in
 Lemma~\ref{em:lemma:homogeneity-nonconstant_two_sided_filter}.
 Fix any system $A$ and $a \in \St_\Real(A)^*$.
 Choose $\zeta_\pm \in \StN(C)$ such that
 $e_C(Y(\zeta_+)) > e_C(Y(\zeta_-))$, and define
 $q_\pm \coloneqq e_C(Y(\zeta_\pm))$,
 $c \coloneqq (q_+ + q_-)/2$, and $\lambda \coloneqq (q_+ - q_-)/2 > 0$.
 Since effects separate states, mixing finitely many measurements
 gives one whose effects span $\St_\Real(A)^*$; classical
 postprocessing yields, for small $\delta > 0$, a binary measurement
 $\{ f_a, e_A - f_a \}$ with $f_a \coloneqq (e_A + \delta a)/2$.
 By the conditioned-test and coarse-graining assumptions, one obtains the channel
 $N_a \coloneqq \zeta_+ \c f_a + \zeta_- \c (e_A - f_a) \in \Trans(A,C)$,
 which prepares $\zeta_+$ or $\zeta_-$ according to the measurement outcome.
 Setting $c_e \coloneqq c$ and $c_a \coloneqq \lambda\delta > 0$ gives
 $e_C \c Y \c N_a = c_e e_A + c_a a$.
 Under local equivalence and either CDP or ES purification,
 Lemma~\ref{sm:lemma:CDP_ES-reversible_dilation} gives
 systems $B,F$, a normalized pure state $\alpha \in \StNP(B)$,
 and a reversible channel $W \in \Trans(AB,CF)$ satisfying
 $N_a = (\id_C \ot e_F) \c W \c (\id_A \ot \alpha)$.
 Define a real-linear map
 $\Gamma_a \colon \Trans_\Real(C,C) \to \Trans_\Real(A,A)$ by
 \begin{alignat}{1}
  \Gamma_a(T)
  &\coloneqq (\id_A \ot e_B) \c W^{-1} \c (T \ot \id_F)
  \c W \c (\id_A \ot \alpha).
  \label{em:eq:positive_lifting_map}
 \end{alignat}
 Reversibility and normalization give $\Gamma_a(\id_C) = \id_A$;
 closure under composition and tensor products gives
 $\Gamma_a(\Trans_+(C,C)) \subseteq \Trans_+(A,A)$.
 Also, $e_A \c \Gamma_a(Y) = e_C \c Y \c N_a = c_e e_A + c_a a$.
 Put
 \begin{alignat}{1}
  X_a &\coloneqq \frac{\Gamma_a(Y) - c_e\id_A}{c_a}.
 \end{alignat}
 This gives $e_A \c X_a = (e_A \c \Gamma_a(Y) - c_e e_A) / c_a = a$.
 For each $\varepsilon \in \{ 1,-1 \}$, set
 \begin{alignat}{1}
  P_{n,\varepsilon} &\coloneqq (1 - \varepsilon c_e r_n)\Gamma_a(T_{n,\varepsilon}), \qquad
  s_n \coloneqq c_a r_n > 0.
 \end{alignat}
 Since $r_n \to 0$, the factor $1 - \varepsilon c_e r_n$ is positive for all sufficiently large $n$.
 Discarding the finitely many exceptional terms and reindexing yields
 $P_{n,\varepsilon} \in \Trans_+(A,A)$ and $s_n \to 0$.
 Linearity of $\Gamma_a$ and the limit $T_{n,\varepsilon} \to \id_C$ now give
 \begin{alignat}{1}
  \frac{P_{n,\varepsilon} - \id_A}{s_n} &= \frac{1}{c_a}
  \left[\Gamma_a\left(\frac{T_{n,\varepsilon} - \id_C}{r_n}\right)
  - \varepsilon c_e\Gamma_a(T_{n,\varepsilon})\right] \to \varepsilon X_a.
 \end{alignat}
\end{proof}

\begin{lemma} \label{em:lemma:exponential_cone_automorphisms}
 For every system $A$ and $a \in \St_\Real(A)^*$,
 there exists $X_a \in \Trans_\Real(A,A)$ satisfying
 $e_A \c X_a = a$ and $e^{tX_a} \in \cA_A$ for every $t \in \Real$.
\end{lemma}
\begin{proof}
 Take $X_a$, $P_{n,\varepsilon}$, and $s_n$ from
 Lemma~\ref{em:lemma:two_sided_positive_tangent}.
 Fix $t \ge 0$ and set $h_n \coloneqq t/s_n \ge 0$.
 Since $P_{n,\varepsilon} \in \Trans_+(A,A)$, the power $P_{n,\varepsilon}^i$ preserves
 the state cone $\St_+(A)$ for every nonnegative integer $i$.
 Therefore, the exponential
 \begin{alignat}{1}
  e^{h_n(P_{n,\varepsilon} - \id_A)} = e^{-h_n}\sum_{i=0}^\infty \frac{h_n^i}{i!}P_{n,\varepsilon}^i
 \end{alignat}
 also preserves $\St_+(A)$.
 As $n \to \infty$, it converges to $e^{\varepsilon tX_a}$.
 Taking both signs of $\varepsilon$ gives the mutually inverse cone-preserving maps
 $e^{tX_a}$ and $e^{-tX_a}$.
 Thus both belong to $\cA_A$ for $t \ge 0$, and hence for all $t \in \Real$.
\end{proof}

We now prove homogeneity by determinant maximization.
\begin{proposition} \label{em:pro:cone_homogeneity}
 For every system $A$, the state cone $\St_+(A)$ is homogeneous.
\end{proposition}
\begin{proof}
 Write $d_A \coloneqq \dim \St_\Real(A)$; the case $d_A = 1$ is immediate.
 Fix $z \in \Int \St_+(A)$ and set
 $\cD_z \coloneqq \{ g \in \cA_A \mid e_A(g(z)) = 1 \}$.
 Each $g \in \cD_z$ maps the section
 $\cN_g \coloneqq \{ v \in \St_+(A) \mid e_A(g(v)) = 1 \}$
 onto $\StN(A)$.
 Since $g$ maps the region between $\cN_g$ and the origin onto that between
 $\StN(A)$ and the origin, their volumes differ by a factor of $|\det g|^{-1}$.
 Thus minimizing this volume over $\cD_z$ amounts to choosing
 \begin{alignat}{1}
  g_z &\in \argmax_{g \in \cD_z} |\det g|.
 \end{alignat}
 Since $z$ is interior, $\cD_z$ is uniformly bounded
 (Lemma~\ref{sm:lemma:homogeneity-positive_map_bound}),
 so a maximizing sequence has a convergent subsequence.
 Let $g_z$ denote its limit.
 Since $e_A(z)^{-1} \id_A \in \cD_z$, we have $|\det g_z| > 0$.
 By continuity of inversion and closedness of the state cone,
 both $g_z$ and $g_z^{-1}$ preserve the cone, so $g_z \in \cD_z$.
 It remains to show that $\omega_z \coloneqq g_z(z)$ is independent of $z$.
 Indeed, this would imply $g_x(x) = g_y(y)$ for all $x,y \in \Int \St_+(A)$,
 and $g_y^{-1} \c g_x \in \cA_A$ would map $x$ to $y$.

 For a real-linear map $X$ on $\St_\Real(A)$, let $\tr X$ denote its trace.
 Fix $a \in \St_\Real(A)^*$ and choose $X_a$ from Lemma~\ref{em:lemma:exponential_cone_automorphisms},
 independently of $z$.
 For each $t \in \Real$, set
 $\lambda_z(t) \coloneqq e_A(e^{tX_a}(\omega_z)) > 0$.
 As $e^{tX_a},g_z \in \cA_A$,
 the map $h_{z,t} \coloneqq \lambda_z(t)^{-1}e^{tX_a} \c g_z$
 lies in $\cD_z$.
 Moreover,
 $|\det h_{z,t}| = \lambda_z(t)^{-d_A}\det(e^{tX_a})|\det g_z|$.
 Maximality of $|\det g_z|$, together with
 $\det(e^{tX_a}) = e^{t\tr X_a}$, therefore gives
 \begin{alignat}{1}
  t\tr X_a - d_A\log\lambda_z(t) &\le 0, \quad \forall t \in \Real.
 \end{alignat}
 Since $\lambda_z(0) = e_A(\omega_z) = 1$, the left-hand side has maximum $0$ at $t = 0$,
 so its derivative there vanishes: $\tr X_a - d_A\lambda_z'(0) = 0$.
 As $\lambda_z'(0) = e_A(X_a(\omega_z)) = a(\omega_z)$,
 it follows that $a(\omega_z) = d_A^{-1} \tr X_a$.
 Since $X_a$ was chosen independently of $z$, so is the right-hand side.
 Consequently, $a(\omega_z)$ is independent of $z$ for every $a \in \St_\Real(A)^*$,
 which implies that $\omega_z$ is independent of $z$.
\end{proof}

Lemma~\ref{sm:lemma:homogeneity_criterion} identifies the common image
with the centroid and establishes four equivalent criteria for homogeneity.
These equivalent criteria may be useful more generally for establishing homogeneity
in other OPTs.

\clearpage

\onecolumngrid
\makeatletter
\onecolumn@grid@setup
\let\set@footnotewidth\set@footnotewidth@one
\let\compose@footnotes\compose@footnotes@one
\makeatother
\section*{Appendices}
\appendix
\counterwithout{equation}{section}
\renewcommand{\theequation}{\arabic{equation}}
\setcounter{secnumdepth}{3}
\setcounter{tocdepth}{2}
\noindent
The following appendices give complete proofs of Theorems~\ref{thm:CDP-quantum}
 and~\ref{thm:ES-quantum} and establish the framework used in the main text.
Appendix~\ref{sm:sec:OPT} presents operational-probabilistic theories;
Appendix~\ref{sm:sec:postulates} states the postulates and defines quantum theory;
Appendix~\ref{sm:sec:CDP_ES_basic} derives the basic consequences of CDP/ES purification.
Appendices~\ref{sm:sec:E1_E2}--\ref{sm:sec:simple_complex} complete the reconstruction,
with a proof roadmap in Sec.~\ref{sm:subsec:postulates-roadmap}.
Appendix~\ref{sm:sec:FPS} gives the FPS counterexample.
Appendices~\ref{sm:sec:probability_product}--\ref{sm:sec:infinite} contain
additional proofs, a comparison with Ref.~\cite{Chi-Dar-Per-2011}, and a discussion
of extensions to infinite dimensions.

\phantomsection
\begingroup
\renewcommand{\tocname}{Contents}
\small
\newif\ifappendixcontents
\appendixcontentsfalse
\let\originalcontentsline\contentsline
\def\contentsline#1#2#3#4{%
 \ifappendixcontents\originalcontentsline{#1}{#2}{#3}{#4}\fi}
\makeatletter
\let\originalappendixtoc\appendix@toc
\def\appendix@toc{\appendixcontentstrue\originalappendixtoc}
\makeatother
\tableofcontents
\endgroup
\clearpage

\section{Operational-probabilistic theories (OPTs)} \label{sm:sec:OPT}

A theory satisfying all the assumptions stated in this appendix is called an
\termdef{operational-probabilistic theory} (OPT).
Apart from local discriminability, which is not assumed here but imposed separately as
local equivalence in Sec.~\ref{sm:subsec:postulates-LE}, this framework is essentially
that of causal operational-probabilistic theories described in Chaps.~II--VI
of Ref.~\cite{Chi-Dar-Per-2010};
we slightly modify its presentation.
An OPT generalizes the operational and probabilistic features of quantum theory and is closely related to the
framework commonly called generalized probabilistic theories \cite{Dak-Bru-2009,Pla-2021,Mul-2021}.

Let $\Natural$ denote the set of positive integers.
Let $\Real$ and $\Complex$ denote the sets of real and complex numbers, respectively.
Let $\Realp$ denote the set of nonnegative real numbers, and let $[0,1]$ denote the set of real numbers
between $0$ and $1$, inclusive.
Let $\Complex^n$ and $\Complex^{n \times n}$ denote, respectively, the sets of complex column vectors of length $n$
and complex square matrices of order $n$.
For each positive integer $n$, let $\Her_n$ and $\Her_n^+$ denote the real vector space of $n \times n$ complex
Hermitian matrices and its cone of positive semidefinite matrices, respectively.
We denote the trace on $\Her_n$ by $\tr_n$, or simply by $\tr$.
The set of density matrices of order $n$ is
$\Den_n \coloneqq \{ \rho \in \Her_n^+ \mid \tr \rho = 1 \}$.
The real vector space spanned by $\Den_n$ is $\Her_n$.
The identity matrix of order $n$ is denoted by $I_n$.
Unless stated otherwise, all real vector spaces considered in this paper
are finite dimensional.

For a real vector space $\V$ and a convex cone $\cC \subseteq \V$, its \termdef{dual cone} is
$\cC^* \coloneqq \{ f \in \V^* \mid f(x) \ge 0 ~(\forall x \in \cC) \}$,
where $\V^*$ is the dual vector space of $\V$; $\cC^*$ is itself a convex cone.
The symbol $\zero$ denotes the zero element of the real vector space determined by context.
For a subset $S$ of $\V$, let $\Span S$ and $\conv S$ denote its real linear span and convex hull,
respectively, and let $\Int S$ denote its interior in $\V$.
Moreover, $\cone S$ denotes the convex cone of all finite nonnegative linear combinations of elements of $S$.
A subset $F$ of a convex cone $K \subseteq \V$ is called a \termdef{face} of $K$ if $F$ is a convex cone
and $u,v \in K$ with $u + v \in F$ implies $u,v \in F$.
For $x \in K\setminus\{ \zero \}$, the ray $\Realp x \coloneqq \{ c x \mid c \in \Realp \}$ is an
\termdef{extreme ray} of $K$ if it is a face of $K$, or equivalently,
if $x = y + z$ for $y,z \in K$ implies $y,z \in \Realp x$.
We then also say that $x$ lies on an extreme ray of $K$.
For a real vector space $\V$, let $\GL(\V)$ denote the group of invertible real-linear maps on $\V$.
For convex sets $X$ and $Y$ spanning real vector spaces $\V$ and $\W$, respectively, a linear isomorphism
$f \colon \V \to \W$ satisfying $f(X) = Y$ is, by abuse of terminology, called a
\termdef{linear isomorphism} from $X$ onto $Y$.
If one exists, $X$ and $Y$ are \termdef{linearly isomorphic}, or simply \termdef{isomorphic},
written $X \cong Y$.
In particular, an isomorphism from $X$ onto itself is a \termdef{linear automorphism}, and
$\Aut(X) \coloneqq \{ g \in \GL(\V) \mid g(X) = X \}$ denotes the set of all such automorphisms.
For a complex square matrix $M$ of order $n$, define $\Ad_M \colon \Her_n \to \Her_n$ by
$\Ad_M(X) \coloneqq M X M^\dag$, where $\dag$ denotes conjugate transpose.
When $V$ is unitary, $\Ad_V$ is called \termdef{unitary conjugation}.
We have $\Ad_V \in \Aut(\Den_n)$, and hence $\tr_n \c \Ad_V = \tr_n$.
For each positive integer $n$, fix the standard basis of $\Complex^n$ and consider the corresponding transposition map
$\T \colon \Her_n \ni X \mapsto X^\T \in \Her_n$.

\subsection{Framework of OPTs} \label{sm:subsec:OPT-framework}

We assume that there are objects called \termdef{systems}, denoted by symbols such as $A$, $B$, $C$, $\dots$.
Let $\System$ be the set of all systems.
For any systems $A,B$, we consider a collection of \termdef{tests} from $A$ to $B$.
A test $\cF$ from $A$ to $B$ is a finite collection $\{ f_i \}_{i=1}^k$, where $k$ may depend on the test,
and each element $f_i$ is called a \termdef{transformation}.
Intuitively, performing the test $\cF$ produces a classical outcome $i \in \{ 1,\dots,k \}$
and applies the corresponding transformation $f_i$.
Although Ref.~\cite{Chi-Dar-Per-2010} allows more general outcome sets, we restrict to finite-outcome tests
to simplify notation.
All constructions and proofs below use only finite-outcome tests.
The set of all transformations from $A$ to $B$ is denoted by $\Trans(A,B)$.
A transformation $f$ is called \termdef{deterministic}, or a \termdef{channel}, if the singleton collection $\{ f \}$
is a test.
A test from $A$ to $A$ is often called a test on $A$, and the same convention is used for transformations and channels.

There is a special system $I$, called the \termdef{trivial system}, such that every probability distribution
(i.e., every collection $\{ p_i \in \Realp \}_{i=1}^k$ satisfying $\sum_{i=1}^k p_i = 1$)
is a test from $I$ to $I$, and there are no other tests from $I$ to $I$.
Hence $\Trans(I,I) = [0,1]$.
Elements of $\Trans(I,I)$, i.e., real numbers between $0$ and $1$, are called \termdef{probabilities}.
We define
\begin{alignat}{1}
 \St(A) &\coloneqq \Trans(I,A), \quad \Eff(A) \coloneqq \Trans(A,I).
\end{alignat}
Elements of $\St(A)$ are called \termdef{states} of $A$, and elements of $\Eff(A)$ are called \termdef{effects} of $A$.
A test from $A$ to $I$ is called a \termdef{measurement} on $A$.
Intuitively, performing a measurement $\{ a_i \}_{i=1}^k$ produces one of the outcomes $1,\dots,k$ and applies
the corresponding effect $a_i$.
By definition, $a$ is an effect of $A$ if and only if there exists a measurement on $A$ containing $a$.
Diagrammatically, a transformation $f \in \Trans(A,B)$, a state $\rho \in \St(A)$, an effect $a \in \Eff(A)$,
and a probability $p$ are represented as follows:
\begin{alignat}{1}
 \includegraphics[scale=1.0,alt={f, \rho, a, p}]{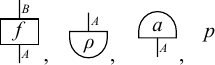}.
 \label{sm:eq:OPT-component}
\end{alignat}
Thus transformations are represented by boxes and systems by wires.
The wire for the trivial system $I$ and the frame of a probability box are often omitted, as in this diagram.
\begin{ex} \label{sm:ex:quant_basic}
 In quantum theory, $\Trans(A,B)$ is the set of completely positive (CP) trace-nonincreasing maps
 from $\Complex^{n_A \times n_A}$ to $\Complex^{n_B \times n_B}$.
 For each system $A$ in quantum theory, $n_A$ is called the \termdef{level} of $A$.
 A collection $\{ f_i \in \Trans(A,B) \}_{i=1}^k$ is a test if and only if $\sum_{i=1}^k f_i$ is trace preserving.
 In particular, $f \in \Trans(A,B)$ is a channel if and only if $f$ is trace preserving.
 We have $n_I = 1$, and each probability $p$ is identified with the map
 $\Complex^{1 \times 1} \ni z \mapsto p z \in \Complex^{1 \times 1}$.
 Identifying a positive semidefinite matrix $\rho \in \Her_{n_A}^+$
 with the CP map $\Complex^{1 \times 1} \ni z \mapsto z \rho \in \Complex^{n_A \times n_A}$, $\St(A)$ can be regarded
 as the set of positive semidefinite matrices of order $n_A$ with trace at most $1$.
 Thus $\rho \in \St(A)$ is deterministic if and only if $\tr \rho = 1$; note that, in the present terminology,
 matrices with $\tr \rho < 1$ are also called states.
 Similarly, identifying a positive semidefinite matrix $E \in \Her_{n_A}^+$ with
 the CP map $\Complex^{n_A \times n_A} \ni X \mapsto \tr(EX) \in \Complex^{1 \times 1}$,
 $\Eff(A)$ can be regarded as the set of positive semidefinite matrices of order $n_A$
 whose largest eigenvalue is at most $1$.
 In particular, $E$ is deterministic if and only if $E = I_{n_A}$.
\end{ex}

Tests admit sequential and parallel composition as follows.
We assume that, for all systems $A$, $B$, and $C$, any test $\cF \coloneqq \{ f_i \}_{i=1}^k$
from $A$ to $B$ and any test $\cG \coloneqq \{ g_j \}_{j=1}^l$ from $B$ to $C$ can be composed sequentially.
The composite is a test from $A$ to $C$ with $kl$ outcomes, denoted by
$\cG \c \cF \coloneqq \{ g_j \c f_i \}_{(i,j) = (1,1)}^{(k,l)}$.
Here $g_j \c f_i$ is the composite of $f_i$ and $g_j$, and is represented by vertically connecting the boxes:
\begin{alignat}{1}
 \includegraphics[scale=1.0,alt={g_j \c f_i}]{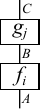}.
 \label{sm:eq:OPT-process_composition}
\end{alignat}
We also assume that any two systems $A$ and $B$ can be connected in parallel.
The resulting system is called the \termdef{composite system} of $A$ and $B$ and is denoted by $AB$.
We assume that, for systems $A$, $B$, $A'$, and $B'$, any test
$\cF \coloneqq \{ f_i \}_{i=1}^k$ from $A$ to $B$ and any test
$\cH \coloneqq \{ h_j \}_{j=1}^l$ from $A'$ to $B'$ can be connected in parallel.
The parallel composition is a test from $AA'$ to $BB'$ with $kl$ outcomes, denoted by
$\cF \ot \cH \coloneqq \{ f_i \ot h_j \}_{(i,j) = (1,1)}^{(k,l)}$.
Here $f_i \ot h_j$ denotes the parallel composition of $f_i$ and $h_j$, represented by placing the boxes side by side:
\begin{alignat}{1}
 \includegraphics[scale=1.0,alt={f_i \ot h_j}]{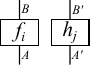}.
 \label{sm:eq:OPT-process_parallel_composition}
\end{alignat}

Therefore $g \c f \in \Trans(A,C)$ and $f \ot h \in \Trans(AA',BB')$ for all $f \in \Trans(A,B)$, $g \in \Trans(B,C)$,
and $h \in \Trans(A',B')$.
If $f$ and $g$ are deterministic, then $g \c f$ is deterministic; if $f$ and $h$ are deterministic, then
$f \ot h$ is deterministic.
For a state $\rho \in \St(A)$ and a transformation $f \in \Trans(A,B)$, the composite $f \c \rho \in \St(B)$ is
often written as $f(\rho)$.
In particular, for a state $\rho \in \St(A)$ and an effect $a \in \Eff(A)$, the composite $a(\rho)$ is an element
of $\Trans(I,I)$ and hence is a probability.
\begin{ex}
 In quantum theory, sequential composition is given by composition of maps,
 and parallel composition by their tensor product.
 In particular, the composite of a state $\rho \in \St(A)$ and an effect $E \in \Eff(A)$ is the probability
 $\tr(E \rho)$.
\end{ex}

Sequential and parallel composition are assumed to be compatible in the following sense.
For each system $A$, there is a channel $\id_A$ on $A$, called the \termdef{identity channel}.
For any test $\cF$ from $A$ to a system $B$, $\cF \c \{ \id_A \} = \cF$ holds, and for any test $\cG$
from $B$ to $A$, $\{ \id_A \} \c \cG = \cG$ holds.
Thus $f \c \id_A = f = \id_B \c f$ for every $f \in \Trans(A,B)$.
Diagrammatically,
\begin{alignat}{1}
 \includegraphics[scale=1.0,alt={f \c \id_A = f = \id_B \c f}]{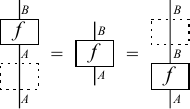},
 \label{sm:eq:OPT-process_identity}
\end{alignat}
where identity channels are drawn as wires.
The dashed boxes in the diagram are only aids for readability; intuitively, system wires can be stretched freely.
In particular, since $\{ \id_I \}$ is a probability distribution, $\id_I = 1$.
We assume that, for any sequentially composable tests $\cF$, $\cG$, and $\cH$,
sequential composition is associative: $\cH \c (\cG \c \cF) = (\cH \c \cG) \c \cF$.
Hence $h \c (g \c f) = (h \c g) \c f$ for transformations:
\begin{alignat}{1}
 \includegraphics[scale=1.0,alt={h \c (g \c f) = (h \c g) \c f}]{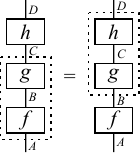}.
 \label{sm:eq:OPT-process_composition_associative}
\end{alignat}
The dashed lines play the role of parentheses and can be omitted.
We assume the following unit and associativity laws for parallel composition of systems and tests.
First, $AI = A = IA$ and $A(BC) = (AB)C$.
Second, $\id_\AB = \id_A \ot \id_B$ and $\id_I \ot f = f = f \ot \id_I$.
Finally, for any tests $\cF$, $\cG$, and $\cH$, parallel composition is associative:
$\cF \ot (\cG \ot \cH) = (\cF \ot \cG) \ot \cH$.
Hence $f \ot (g \ot h) = (f \ot g) \ot h$ for transformations:
\begin{alignat}{1}
 \includegraphics[scale=1.0,alt={f \ot (g \ot h) = (f \ot g) \ot h}]{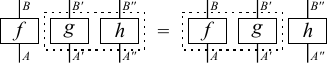}.
 \label{sm:eq:OPT-process_parallel_composition_associative}
\end{alignat}
We also assume the interchange law: for any tests for which the relevant sequential compositions are defined,
$(\cG \ot \cG') \c (\cF \ot \cF') = (\cG \c \cF) \ot (\cG' \c \cF')$.
Hence $(g \ot g') \c (f \ot f') = (g \c f) \ot (g' \c f')$ for transformations:
\begin{alignat}{1}
 \includegraphics[scale=1.0,alt={(g \ot g') \c (f \ot f') = (g \c f) \ot (g' \c f')}]{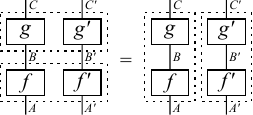}.
 \label{sm:eq:OPT-process_composition_ot}
\end{alignat}
Equations~\eqref{sm:eq:OPT-process_composition_associative},
\eqref{sm:eq:OPT-process_parallel_composition_associative}, and \eqref{sm:eq:OPT-process_composition_ot} imply
that, for a fixed wiring diagram involving three or more transformations, the resulting transformation does not depend
on how the specified sequential and parallel compositions are parenthesized and evaluated.

The following lemma is proved in Appendix~\ref{sm:sec:probability_product}.
\begin{lemma} \label{sm:lemma:OPT-probability_product}
 The sequential composition $p \c q$ and parallel composition $p \ot q$ of probabilities $p,q$
 both equal their ordinary product $pq$ as real numbers.
\end{lemma}

For each $A$, we assume that at least one deterministic effect exists, and we denote one such effect by $e_A$.
The singleton $\{ e_A \}$ is a measurement, and we depict the effect $e_A$ as
\begin{alignat}{1}
 \includegraphics[scale=1.0,alt={e_A}]{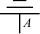}.
 \label{sm:eq:OPT-deterministic_effect}
\end{alignat}
Lemma~\ref{sm:lemma:basic-deterministic_effect_uniqueness} below shows that the deterministic effect of $A$ is unique.
In this sense, this paper deals only with causal theories in the sense of Ref.~\cite{Chi-Dar-Per-2010}.
A state of $A$ satisfying $e_A(\rho) = 1$ is called \termdef{normalized}; the set of normalized states of $A$ is denoted
by $\StN(A)$.
We assume that every system $A$ has at least one normalized state.
For any systems $A,B$, the nonemptiness of $\StN(B)$ allows us to choose $\rho \in \StN(B)$.
Since $\rho \c e_A \in \Trans(A,B)$, $\Trans(A,B)$ is nonempty.
\begin{ex}
 In quantum theory, $e_A = I_{n_A}$.
 A state $\rho$ is normalized if and only if $\tr \rho = 1$.
\end{ex}

For systems $A$ and $B$, equality of transformations $f$ and $g$ from $A$ to $B$ is defined operationally:
$f = g$ if and only if
$\craket{a|f|\rho} = \craket{a|g|\rho}$ for every system $E$ and every $\rho \in \StN(AE)$ and $a \in \Eff(BE)$, where
$\craket{a|f|\rho} \coloneqq a((f \ot \id_E)(\rho))$.
This notation is used several times in this subsection.
For clarity, the definition is depicted as follows:
\begin{alignat}{1}
 \adjustbox{valign=c}{\includegraphics[alt={f = g}]{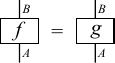}}
 &\qquad\Leftrightarrow\qquad
 \adjustbox{valign=c}{\includegraphics[alt={\craket{a|f|\rho} = \craket{a|g|\rho}}]{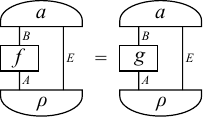}}
 \quad (\forall E \in \System, ~\rho \in \StN(AE), ~a \in \Eff(BE)).
 \label{sm:eq:OPT-process_equivalence}
\end{alignat}
In the special case $A = I$, states $\rho$ and $\sigma$ of $B$ are equal if and only if
$a(\rho) = a(\sigma)$ for every $a \in \Eff(B)$.%
\footnote{Proof: By the definition of equality of transformations, $\rho = \sigma$ if and only if
$b((\rho \ot \id_E)(\tau)) = b((\sigma \ot \id_E)(\tau))$
for all systems $E$, $\tau \in \StN(IE) = \StN(E)$, and $b \in \Eff(BE)$.
Put $a \coloneqq b \c (\id_B \ot \tau) \in \Eff(B)$.
Then the left-hand side is $a(\rho)$, and similarly the right-hand side is $a(\sigma)$.
Thus $\rho = \sigma$ if and only if $a(\rho) = a(\sigma)$ for every $a \in \Eff(B)$.}
Similarly, in the case $B = I$, effects $a$ and $b$ of $A$ are equal if and only if
$a(\rho) = b(\rho)$ for every $\rho \in \St(A)$.%
\footnote{Proof: By the definition of equality of transformations, $a = b$ if and only if
$c((a \ot \id_E)(\rho)) = c((b \ot \id_E)(\rho))$
for all systems $E$, $\rho \in \StN(AE)$, and $c \in \Eff(IE) = \Eff(E)$.
Put $\sigma \coloneqq (\id_A \ot c)(\rho) \in \St(A)$.
Then the left-hand side is $a(\sigma)$, and similarly the right-hand side is $b(\sigma)$.
Thus $a = b$ if and only if $a(\sigma) = b(\sigma)$ for every $\sigma \in \St(A)$.}

For transformations $f$ and $g$ from $A$ to $B$, if there exists a transformation $h$ from $A$ to $B$ such that
$\craket{a|h|\rho} = \craket{a|f|\rho} + \craket{a|g|\rho}$ for every $E$, $\rho \in \StN(AE)$, and $a \in \Eff(BE)$,
then we write $h = f + g$.
Such a transformation $f + g$ need not exist for arbitrary $f$ and $g$.
Finite sums $\sum_{i=1}^k f_i$ are defined similarly.
For a probability $p$ and a transformation $f$, the parallel composition $p \ot f$ is often written simply as $pf$,
and $pf$ is called the \termdef{probability multiple} of $f$ by $p$.
These conventions define sums of transformations and multiplication by probabilities.
The transformation $f \in \Trans(A,B)$ satisfying $\craket{a|f|\rho} = 0$ for every $E$, $\rho \in \StN(AE)$,
and $a \in \Eff(BE)$ is denoted by $\zero_{A,B}$, or simply $\zero$.
$0f = \zero_{A,B} \in \Trans(A,B)$ for any $f \in \Trans(A,B)$.

We formally extend the above definitions and consider finite real linear combinations $\sum_i c_i f_i$
with $f_i \in \Trans(A,B)$ and $c_i \in \Real$.
After identifying operationally equivalent linear combinations, i.e., identifying $\sum_i c_i f_i$ and $\sum_j d_j g_j$
whenever $\sum_i c_i \craket{a|f_i|\rho} = \sum_j d_j \craket{a|g_j|\rho}$ for all $E$,
$\rho \in \StN(AE)$, and $a \in \Eff(BE)$, the resulting real vector space is denoted by $\Trans_\Real(A,B)$.
It is the real vector space spanned by $\Trans(A,B)$.
If $f \in \Trans_\Real(A,B)$ is represented as $f = \sum_i c_i f_i$
$~(c_i \in \Real, ~f_i \in \Trans(A,B))$ and $g \in \Trans_\Real(B,C)$ is represented as
$g = \sum_j d_j g_j$ $~(d_j \in \Real, ~g_j \in \Trans(B,C))$, define
$g \c f \coloneqq \sum_i \sum_j c_i d_j g_j \c f_i \in \Trans_\Real(A,C)$.
Similarly, if $h \in \Trans_\Real(A',B')$ is represented as $h = \sum_k q_k h_k$
$~(q_k \in \Real, ~h_k \in \Trans(A',B'))$, define
$f \ot h \coloneqq \sum_i \sum_k c_i q_k f_i \ot h_k \in \Trans_\Real(AA',BB')$.
These definitions are independent of representatives because operational equivalence is defined by all statistics,
including all ancillary systems.
Indeed, if one representative is replaced by an operationally equivalent one, then the statistics obtained
after inserting it into the sequential or parallel contexts used above remain unchanged; the surrounding parts can be
absorbed into the ancillary system, state, and effect appearing in the definition of operational equivalence.
By linearity of these statistics, the same remains true after taking finite real linear combinations,
so the resulting element of the appropriate real vector space is unchanged.
The analogue of Eq.~\eqref{sm:eq:OPT-process_equivalence} holds for elements of $\Trans_\Real(A,B)$.
Thus, for $f,g \in \Trans_\Real(A,B)$, the equality $f = g$ holds if and only if
$\craket{a|f|\rho} = \craket{a|g|\rho}$ for every system $E$,
$\rho \in \StN(AE)$, and $a \in \Eff(BE)$.%
\footnote{Proof: Write $f = \sum_i c_i f_i$ and $g = \sum_j d_j g_j$,
with $c_i,d_j \in \Real$ and $f_i,g_j \in \Trans(A,B)$.
By the definition of operational identification of real linear combinations,
$f = g$ is equivalent to
$\sum_i c_i \craket{a|f_i|\rho} = \sum_j d_j \craket{a|g_j|\rho}$
for every system $E$, $\rho \in \StN(AE)$, and $a \in \Eff(BE)$.
By linearity, this is precisely the asserted equality of the statistics of $f$ and $g$.}
For $p \in [0,1]$ and $f \in \Trans(A,B)$, the probability multiple $pf$ agrees with the scalar
multiple of $f$ by $p$ in $\Trans_\Real(A,B)$.%
\footnote{Proof: Let $g$ denote the scalar multiple of $f$ by $p$ in $\Trans_\Real(A,B)$.
For every system $E$, $\omega \in \StN(AE)$, and $a \in \Eff(BE)$,
we have $\craket{a|pf|\omega} = p\craket{a|f|\omega} = \craket{a|g|\omega}$.
Consequently, $pf = g$.}
Set $\St_\Real(A) \coloneqq \Trans_\Real(I,A)$ and $\Eff_\Real(A) \coloneqq \Trans_\Real(A,I)$.
These are the real vector spaces spanned by $\St(A)$ and $\Eff(A)$, respectively.
Elements of $\St_\Real(A)$, $\Eff_\Real(A)$, and $\Trans_\Real(A,B)$ are drawn using the same kind of boxes
as in Eq.~\eqref{sm:eq:OPT-component}.

For each $A$, $\St_\Real(A)$ is assumed finite dimensional.
Its dimension is called the \termdef{dimension} of $A$ and denoted by $d_A$.
States and effects separate one another, so the dimension of either real vector space cannot exceed the dimension
of the dual of the other; hence $\Eff_\Real(A)$ also has dimension $d_A$.
We identify each $y \in \Eff_\Real(A)$ with the linear functional
$\St_\Real(A) \ni x \mapsto y(x) \in \Real$.
This identifies $\Eff_\Real(A)$ with the dual vector space $\St_\Real(A)^*$.
We assume from now on that $\StN(A)$ is closed in $\St_\Real(A)$.
This is naturally interpreted as the requirement that the limit of normalized states is again a normalized state.
Equivalently, whenever a sequence of normalized states converges in $\St_\Real(A)$, its limit is also a
normalized state.

For a test $\cF \coloneqq \{ f_i \}_{i=1}^k$ from $A$ to $B$, we assume that any two of its outcomes can be merged.
For example, if the chosen elements are $f_1$ and $f_2$, then the collection $\cG \coloneqq \{ g_i \}_{i=1}^{k - 1}$
with $g_1 \coloneqq f_1 + f_2$ and, when $k \ge 3$,
$g_i \coloneqq f_{i + 1}$ for $i=2,\dots,k-1$ is a test.
This is called \termdef{coarse-graining} of $f_1$ and $f_2$.
Intuitively, $\cG$ is obtained from $\cF$ by identifying the classical outcomes $1$ and $2$.
Thus, if there exists a test containing $f_1$ and $f_2$, then $f_1 + f_2$ is a transformation.
Other pairs can be coarse-grained in the same way.
We also assume that conditioned tests are allowed: if $\cF \coloneqq \{ f_i \}_{i=1}^k$ is a test from $A$ to $B$
and, for each $i$, $\cG^{(i)} \coloneqq \{ g^{(i)}_j \}_{j=1}^{l_i}$ is a test from $B$ to $C$, then
$\{ g^{(i)}_j \c f_i \mid i=1,\dots,k, ~j=1,\dots,l_i \}$ is a test from $A$ to $C$.
Intuitively, one first performs $\cF$ and then, depending on the classical outcome $i$, performs $\cG^{(i)}$.

Conditioning on probability tests and coarse-graining show that, for any systems $A,B$,
$\Trans(A,B)$ is convex.%
\footnote{Proof: Fix $f,g \in \Trans(A,B)$ and $p \in [0,1]$.
It suffices to show that $pf + (1 - p)g \in \Trans(A,B)$.
Choose tests $\cF$ and $\cG$ containing $f$ and $g$, respectively.
Parallel composition of the probability test $\{ p,1 - p \}$ with $\{ \id_A \}$ gives the test
$\{ p \ot \id_A, (1 - p) \ot \id_A \}$.
Conditionally composing this test with $\cF$ and $\cG$ yields a test containing both
$pf$ and $(1 - p)g$.
Coarse-graining these two transformations gives $pf + (1 - p)g \in \Trans(A,B)$.}
In particular, $\St(A)$ and $\Eff(A)$ are convex.
The set $\StN(A)$ is also convex, since it is the intersection of the convex set $\St(A)$
with the affine set $\{ x \in \St_\Real(A) \mid e_A(x) = 1 \}$.
A normalized state $\rho \in \StN(A)$ is \termdef{pure} if every decomposition $\rho = p \sigma + (1 - p) \tau$
with normalized states $\sigma$ and $\tau$ of $A$ and $0 < p < 1$ implies $\sigma = \tau = \rho$.
A normalized state that is not pure is called \termdef{mixed}.
The set of all normalized pure states of $A$ is denoted by $\StNP(A)$.

\begin{lemma} \label{sm:lemma:basic-deterministic_effect_uniqueness}
 For every system $A$, the only deterministic effect is $e_A$.
\end{lemma}
\begin{proof}
 Fix a deterministic effect $a$, a state $\rho \in \St(A)$, and a test
 $\cF \coloneqq \{ \rho,\rho_2,\dots,\rho_k \}$ containing $\rho$.
 Composing each outcome of $\cF$ with $\{ e_A \}$ yields the probability distribution
 $\{ e_A(\rho),e_A(\rho_2),\dots,e_A(\rho_k) \}$.
 Conditionally composing $\cF$ with $\{ a \}$ after its first outcome and with $\{ e_A \}$
 after each remaining outcome yields the probability distribution
 $\{ a(\rho),e_A(\rho_2),\dots,e_A(\rho_k) \}$.
 Both distributions sum to $1$, so $a(\rho) = e_A(\rho)$.
 As $\rho$ was arbitrary, equality of effects gives $a = e_A$.
\end{proof}

This lemma, together with the fact that the parallel composition of deterministic effects is deterministic,
implies $e_{AB} = e_A \ot e_B$.
Moreover, for any measurement $\{ a_i \}_{i=1}^k$ on $A$, coarse-graining all its outcomes gives
the deterministic effect $\sum_{i=1}^k a_i$; hence, by Lemma~\ref{sm:lemma:basic-deterministic_effect_uniqueness},
$\sum_{i=1}^k a_i = e_A$.

\begin{lemma} \label{sm:lemma:basic-zero_state}
 A state $\rho \in \St(A)$ satisfies $e_A(\rho) = 0$ if and only if $\rho = \zero$.
\end{lemma}
\begin{proof}
 If $\rho = \zero$, then $e_A(\rho) = 0$ by the definition of $\zero$.
 Conversely, suppose that $e_A(\rho) = 0$.
 Every $a \in \Eff(A)$ belongs to a measurement whose effects sum to $e_A$, so
 $0 \le a(\rho) \le e_A(\rho) = 0$.
 Thus $a(\rho) = a(\zero)$ for every $a \in \Eff(A)$, implying $\rho = \zero$.
\end{proof}

\begin{lemma} \label{sm:lemma:basic-normalized_states}
 A state is deterministic if and only if it is normalized.
\end{lemma}
\begin{proof}
 First, suppose that $\rho \in \St(A)$ is deterministic.
 Since the sequential composition of deterministic transformations is deterministic,
 $e_A(\rho)$ is a channel from $I$ to $I$.
 The only such channel is $1$, so $e_A(\rho) = 1$ and $\rho$ is normalized.
 Conversely, let $\rho \in \StN(A)$ and consider a test
 $\{ \rho,\rho_2,\dots,\rho_k \}$ containing $\rho$.
 If $k = 1$, then $\rho$ is deterministic.
 If $k \ge 2$, normalization of the resulting probability distribution and $e_A(\rho) = 1$
 give $e_A(\rho_i) = 0$ for $i \in \{ 2,\dots,k \}$.
 By Lemma~\ref{sm:lemma:basic-zero_state}, every such $\rho_i$ equals $\zero$.
 Coarse-graining the original test therefore gives the singleton test
 $\{ \rho + \sum_{i=2}^k \rho_i \} = \{ \rho \}$.
 Thus $\rho$ is deterministic.
\end{proof}

\begin{lemma} \label{sm:lemma:OPT-state_space}
 For every system $A$, the following equality holds:
 \begin{alignat}{1}
  \St(A) &= \{ p\sigma \mid p \in [0,1],~\sigma \in \StN(A) \}.
  \label{sm:eq:OPT-state_space}
 \end{alignat}
\end{lemma}
\begin{proof}
 The inclusion $\supseteq$ follows from $p\sigma = p \ot \sigma \in \St(A)$
 for $p \in [0, 1]$ and $\sigma \in \StN(A)$.
 Conversely, let $\rho \in \St(A)$ and set $p \coloneqq e_A(\rho)$.
 If $p = 0$, then $\rho = \zero = 0 \sigma$ for any $\sigma \in \StN(A)$
 by Lemma~\ref{sm:lemma:basic-zero_state}; if $p = 1$, then $\rho \in \StN(A)$.
 For $0 < p < 1$, the conditional-test rescaling construction in
 Sec.~4.1.4 of Ref.~\cite{Chi-Dar-Per-2016} yields normalized states converging to $\sigma \coloneqq \rho / p$.
 Since $\StN(A)$ is closed, $\sigma \in \StN(A)$, and hence $\rho = p\sigma$.
\end{proof}

From the definitions of sums and probability multiples, every $f \in \Trans(A,B)$ preserves convex combinations:
$f(\sum_{i=1}^k p_i \rho_i) = \sum_{i=1}^k p_i f(\rho_i)$ for every probability distribution $\{ p_i \}_{i=1}^k$
and states $\rho_i$ of $A$.%
\footnote{Proof: Take an arbitrary effect $a \in \Eff(B)$ and put $a_f \coloneqq a \c f \in \Eff(A)$.
Put $\sigma \coloneqq \sum_{i=1}^k p_i \rho_i$.
Then $a_f(\sigma) = \sum_{i=1}^k p_i a_f(\rho_i)$, equivalently
$a(f(\sigma)) = \sum_{i=1}^k p_i a(f(\rho_i))$.
By the definitions of probabilistic mixing and sums of output states, this equality for every $a$
means that $f(\sigma)$ is the weighted sum $\sum_{i=1}^k p_i f(\rho_i)$.
Thus $f(\sigma) = \sum_{i=1}^k p_i f(\rho_i)$.}
Thus $f$ can be regarded as an affine map from $\St(A)$ to $\St(B)$, and it has a unique linear extension
$f_\mathrm{lin} \colon \St_\Real(A) \to \St_\Real(B)$ satisfying $f(\rho) = f_\mathrm{lin}(\rho)$ for all
 $\rho \in \St(A)$.
At this stage, however, distinct transformations from $A$ to $B$ may still induce the same linear map
$\St_\Real(A) \to \St_\Real(B)$; this possible non-injectivity will be removed later by local equivalence.

Let $\St_+(A)$, $\Eff_+(A)$, and $\Trans_+(A,B)$ be the convex cones generated by $\St(A)$, $\Eff(A)$,
and $\Trans(A,B)$, respectively.
We call $\St_+(A)$ and $\Eff_+(A)$ the \termdef{state cone} and \termdef{effect cone} of $A$.
If $f$ and $g$ are elements of $\Trans_\Real(A,B)$ and satisfy $g - f \in \Trans_+(A,B)$, we write
$f \le g$ or $g \ge f$.
This relation is a partial order.
If $f \in \Trans_+(A,B)$ and $g \in \Trans_+(B,C)$, then
$g \c f \in \Trans_+(A,C)$.
Similarly, if $f \in \Trans_+(A,B)$ and $h \in \Trans_+(A',B')$, then
$f \ot h \in \Trans_+(AA',BB')$.
Each $a \in \Eff_+(A)$ maps each $x \in \St_+(A)$ to a nonnegative real number, and hence
$\Eff_+(A) \subseteq \St_+(A)^*$.
Note that the identification $\Eff_\Real(A) = \St_\Real(A)^*$ does not by itself imply
$\Eff_+(A) = \St_+(A)^*$; in an OPT, the latter equality need not hold.
\begin{ex}
 In quantum theory, $\Trans_\Real(A,B)$ and $\Trans_+(A,B)$ are the sets of Hermiticity-preserving maps and CP maps
 from $\Complex^{n_A \times n_A}$ to $\Complex^{n_B \times n_B}$, respectively.
 In particular, $\St_\Real(A)$ and $\Eff_\Real(A)$ are identified with $\Her_{n_A}$, while $\St_+(A)$ and $\Eff_+(A)$
 are identified with $\Her_{n_A}^+$.
 Since effects are identified with maps $X \mapsto \tr(EX)$, quantum theory satisfies $\Eff_+(A) = \St_+(A)^*$.
\end{ex}

By Eq.~\eqref{sm:eq:OPT-state_space} and the convexity of $\StN(A)$, every $x \in \St_+(A)$ can be written as
$x = a \sigma$ with $a \in \Realp$ and $\sigma \in \StN(A)$.
In particular, if $x \ne \zero$, then $a > 0$ and $e_A(x) = a$.
Thus every nonzero $x \in \St_+(A)$ can be written as
$x = e_A(x) \ol{x}$ with $\ol{x} \coloneqq x / e_A(x) \in \StN(A)$.
The state $\ol{x} = x / e_A(x)$ is called the \termdef{normalization} of $x$.
The following elementary consequence will be used later:
\begin{alignat}{3}
 x &\in \St_+(A), &~e_A(x) &\le 1 &\quad \Rightarrow \quad x &\in \St(A), \nonumber \\
 x &\in \St_+(A), &~e_A(x) & = 1 &\quad \Rightarrow \quad x &\in \StN(A).
 \label{sm:eq:OPT-state_cone_normalization}
\end{alignat}
Indeed, the case $x = \zero$ is immediate.
If $x \ne \zero$, then $x = e_A(x)\ol{x}$ with $\ol{x} \in \StN(A)$.
Hence $e_A(x) \le 1$ implies $x \in \St(A)$ by Eq.~\eqref{sm:eq:OPT-state_space}, and
$e_A(x) = 1$ implies $x = \ol{x} \in \StN(A)$.

A finite family of normalized states $\{ \rho_i \in \StN(A) \}_{i=1}^k$ is \termdef{perfectly distinguishable}
if there exists a measurement $\{ a_i \in \Eff(A) \}_{i=1}^k$ such that
$a_i(\rho_j) = \delta_{i,j}$ for all $i$ and $j$.
Here $\delta_{i,j}$ denotes the Kronecker delta.
The maximum number of perfectly distinguishable normalized states of $A$ is called the \termdef{informational dimension}
of $A$ and is denoted by $n_A$.
Thus there exists a perfectly distinguishable family of size $n_A$, but no such family of size $n_A + 1$.
A perfectly distinguishable family is linearly independent in $\St_\Real(A)$, so its size is at most $d_A$.
We assume that there exists a system with informational dimension at least $2$.
\begin{ex}
 In quantum theory, the informational dimension equals the level $n_A$ of $A$.
 A family $\{ \rho_i \in \StN(A) \}_{i=1}^k$ is perfectly distinguishable
 if and only if the states are mutually orthogonal,
 where $\sigma$ and $\tau$ are density matrices of order $n_A$ and are orthogonal if and only if $\tr(\sigma\tau) = 0$.
\end{ex}

A channel $U \in \Trans(A,B)$ is \termdef{reversible} if there exists a channel
$\tU \in \Trans(B,A)$ such that
$\tU \c U = \id_A$ and $U \c \tU = \id_B$.
Then $\tU$ is called the \termdef{inverse channel} of $U$ and is denoted by $U^{-1}$.
If there exists a reversible channel from $A$ to $B$, the systems $A$ and $B$ are called
\termdef{equivalent}, written $A \simeq B$.
In particular, since $\id_A$ is a reversible channel on $A$, every system $A$ is equivalent to itself.
For every pair of systems $A,B$, we assume that there exists a reversible swap channel
$\alpha_\AB \in \Trans(AB,BA)$ implementing their interchange.
We require these swaps to be compatible with parallel composition and associativity of composites.%
\footnote{Specifically, for all transformations $f \in \Trans(A,A')$ and $g \in \Trans(B,B')$,
we impose $\alpha_{B'A'} \c (g \ot f) \c \alpha_\AB = f \ot g$.
We also require $\alpha_{A(BC)} = (\id_B \ot \alpha_{AC}) \c (\alpha_\AB \ot \id_C)$
for all systems $A,B,C$, where $\alpha_{A(BC)}$ swaps $A$ with $BC$.
The first identity implies $\alpha_{BA} \c \alpha_\AB = \id_\AB$
by taking $f = \id_A$ and $g = \id_B$.}
Mathematically, the systems and transformations (or tests) form a strict symmetric monoidal category.
Thus properties of the composite system $AB$ that are invariant under equivalence also hold for the
reordered composite system $BA$.
If $A \simeq B$ and $B \simeq C$, then there exist reversible channels
$U \in \Trans(A,B)$ and $V \in \Trans(B,C)$, and
$V \c U \in \Trans(A,C)$ is reversible with inverse $U^{-1} \c V^{-1}$.
Hence $A \simeq C$.
Similarly, if $A \simeq B$ and $C \simeq D$, then there exist reversible channels
$U \in \Trans(A,B)$ and $W \in \Trans(C,D)$, and
$U \ot W \in \Trans(AC,BD)$ is reversible with inverse $U^{-1} \ot W^{-1}$.
Hence $AC \simeq BD$.
By repeated use of the order-interchange assumption, associativity of composites, and compatibility
of equivalence with parallel composition, any finite permutation of tensor factors changes a composite
system only up to equivalence.
In the proofs below, whenever composite systems differ only by such a permutation, we choose a
reversible channel implementing the corresponding equivalence and suppress it from the notation.
\begin{ex}
 In quantum theory, $A \simeq B$ if and only if $n_A = n_B$.
 A reversible channel from $A$ to $B$ has the form $\Ad_V$ for some unitary matrix $V$,
 and its inverse is $\Ad_{V^{-1}}$.
\end{ex}

A finite collection of states $\{ x_i \}_{i=1}^n$ satisfying $\rho = \sum_{i=1}^n x_i$ for a normalized state
$\rho \in \StN(A)$ is called a \termdef{refinement} of $\rho$.
For a normalized state $\rho \in \StN(A)$, a state $\sigma \in \St(A)$ with $\sigma \le \rho$ is called a
\termdef{refinement component} of $\rho$.
As defined earlier, $\sigma \le \rho$ means $\rho - \sigma \in \St_+(A)$.
If $\sigma$ is a refinement component of $\rho$, then $\{ \sigma, \rho - \sigma \}$ is a refinement of $\rho$.
The \termdef{refinement-component set} of $\rho$ is
\begin{alignat}{1}
 D_\rho &\coloneqq \{ \sigma \in \St(A) \mid \sigma \le \rho \}.
\end{alignat}
By Eq.~\eqref{sm:eq:OPT-state_cone_normalization},
$\sigma \in \St_+(A)$ and $\sigma \le \rho$ imply $\sigma \in D_\rho$;
likewise, any finite family in $\St_+(A)$ summing to $\rho$ is a refinement of $\rho$.
A normalized state $\rho$ is \termdef{internal} if the real vector space spanned by $D_\rho$ is $\St_\Real(A)$.
\begin{ex}
 In quantum theory, $D_\rho$ is the set of positive semidefinite matrices $\sigma$ satisfying $\sigma \le \rho$,
 and $\rho$ is internal if and only if it is full rank.
\end{ex}

\subsection{Assumptions of OPTs} \label{sm:subsec:OPT-assumptions}

For clarity, we summarize the assumptions made in the preceding subsection.
\begin{enumerate}
 \item \label{sm:ass:OPT-tests}
       There are systems, and for each pair of systems $A,B$, a collection of tests from $A$ to $B$
       is specified.
       Each test is a finite collection of transformations.
 \item \label{sm:ass:OPT-trivial_system}
       There is a trivial system $I$.
       Tests from $I$ to $I$ are exactly the finite probability distributions
       $\{ p_i \in \Realp \}_{i=1}^k$ with $\sum_{i=1}^k p_i = 1$.
 \item \label{sm:ass:OPT-composition}
       Tests can be composed sequentially and in parallel;
       systems can be composed in parallel.
       These operations satisfy the usual identity, associativity,
       unit, and interchange laws, including
       $\id_\AB = \id_A \ot \id_B$
       and $\id_I \ot f = f = f \ot \id_I$.
       For any tests for which the relevant sequential
       and parallel compositions are defined, the interchange law is
       $(\cG \c \cF) \ot (\cG' \c \cF')
       = (\cG \ot \cG') \c (\cF \ot \cF')$.
 \item \label{sm:ass:OPT-operational_identity}
       Equality of transformations is defined operationally by statistics:
       for $f,g \in \Trans(A,B)$, the equality $f = g$ holds if and only if
       $\craket{a|f|\rho} = \craket{a|g|\rho}$ for every system $E$,
       $\rho \in \StN(AE)$, and $a \in \Eff(BE)$.
 \item \label{sm:ass:OPT-deterministic_effect}
       Each system $A$ has a deterministic effect $e_A$ and at least one normalized state.
 \item \label{sm:ass:OPT-coarse_graining}
       Tests are closed under coarse-graining.
 \item \label{sm:ass:OPT-conditioned_tests}
       Conditioned tests are allowed.
       Specifically, if $\{ f_i \}_{i=1}^k$ is a test from $A$ to $B$ and, for each $i$,
       $\{ g_j^{(i)} \}_{j=1}^{l_i}$ is a test from $B$ to $C$, then
       $\{ g_j^{(i)} \c f_i \mid i=1,\dots,k, ~j=1,\dots,l_i \}$ is a test from $A$ to $C$.
 \item \label{sm:ass:OPT-finite_dimensionality}
       $\St_\Real(A)$ is finite dimensional, and $\StN(A)$ is closed in $\St_\Real(A)$.
 \item \label{sm:ass:OPT-nontriviality}
       There exists at least one system with informational dimension at least $2$.
 \item \label{sm:ass:OPT-symmetric_composition}
       For any systems $A,B$, there is a swap channel $\alpha_\AB \in \Trans(AB,BA)$.
       These swaps are compatible with parallel composition and associativity of composites.
\end{enumerate}

\subsection{Basic properties of OPTs} \label{sm:subsec:OPT-basic_properties}

We collect several elementary properties of OPTs.
We will often use these elementary lemmas without explicit reference.
Some of the lemmas below are closely related to lemmas in Ref.~\cite{Chi-Dar-Per-2010},
while others are used there without being stated as separate results.

\begin{lemma} \label{sm:lemma:basic-compact_closed}
 For every system $A$, $\StN(A)$ is a compact convex set and $\St_+(A)$ is a closed convex cone.
 Moreover, $A$ has a normalized pure state.
\end{lemma}
\begin{proof}
 We first show that $\StN(A)$ is compact and convex.
 It is closed and convex, and it remains to prove boundedness.
 Since $\St_\Real(A)$ is finite dimensional, all norms are equivalent.
 Use the operational norm of Ref.~\cite{Chi-Dar-Per-2010},
 $\|x\|_A \coloneqq \sup_{a \in \Eff(A)} a(x) - \inf_{b \in \Eff(A)} b(x)$.
 For $\rho \in \StN(A)$, $a(\rho) \le e_A(\rho) = 1$ and $b(\rho) \ge \zero(\rho) = 0$.
 The upper and lower bounds are attained by $e_A$ and $\zero$, respectively.
 Hence $\|\rho\|_A = 1$.
 Thus $\StN(A)$ is bounded, and therefore compact.

 Next, $\St_+(A)$ is the cone generated by the nonempty compact convex set $\StN(A)$, which does not contain $\zero$;
 such a cone is closed, e.g., by Proposition~1.4.7 of Ref.~\cite{Hir-Lem-2001}.

 Finally, any nonempty compact convex set has an extreme point, and every extreme point of $\StN(A)$ is a normalized
 pure state.
\end{proof}

It follows that every $\rho \in \StN(A)$ can be expressed as a finite convex combination
$\rho = \sum_{i=1}^k p_i \psi_i$ with $p_i \in [0,1]$ and $\psi_i \in \StNP(A)$.
Since the informational dimension of $A$ is $n_A$, there exists a perfectly distinguishable family
$\{ \rho_i \in \StN(A) \}_{i=1}^{n_A}$.
Consequently, there also exists a perfectly distinguishable family
$\{ \psi_i \in \StNP(A) \}_{i=1}^{n_A}$.%
\footnote{Proof: Let $\{ a_i \}_{i=1}^{n_A}$ be a measurement perfectly distinguishing
$\{ \rho_i \}_{i=1}^{n_A}$, and write each state as
$\rho_i = \sum_k p_{i,k} \psi_{i,k}$ with $0 < p_{i,k} \le 1$ and $\psi_{i,k} \in \StNP(A)$.
Choose arbitrary $i,j \in \{ 1,2,\dots,n_A \}$.
Since $\sum_k p_{i,k} = 1$ and $a_i(\psi_{i,k}) \le 1$, we have
$1 = a_i(\rho_i) = \sum_k p_{i,k} a_i(\psi_{i,k}) \le 1$,
which implies $a_i(\psi_{i,1}) = 1$.
If $i \ne j$, then $0 = a_i(\rho_j) \ge p_{j,1}a_i(\psi_{j,1}) \ge 0$, so
$a_i(\psi_{j,1}) = 0$.
Thus $\{ a_i \}_{i=1}^{n_A}$ perfectly distinguishes
$\{ \psi_{i,1} \in \StNP(A) \}_{i=1}^{n_A}$.}

\begin{lemma} \label{sm:lemma:basic-deterministic}
 A transformation $f \in \Trans(A,B)$ is a channel if and only if $e_B \c f = e_A$.
\end{lemma}
\begin{proof}
 This is shown in Lemma~5 of Ref.~\cite{Chi-Dar-Per-2010}.
 The lemma holds for every OPT in the present framework and hence applies to the theories considered here.
\end{proof}

\begin{lemma} \label{sm:lemma:basic-unitary_pure}
 Reversible channels map normalized pure states to normalized pure states.
\end{lemma}
\begin{proof}
 Let $U \in \Trans(A,B)$ be reversible and let $\psi \in \StN(A)$ be pure.
 By Lemma~\ref{sm:lemma:basic-deterministic},
 $e_B(U(\psi)) = e_A(\psi) = 1$, so $U(\psi)$ is normalized.
 If $U(\psi) = p\rho + (1 - p)\sigma$ with distinct normalized states
 $\rho,\sigma \in \StN(B)$ and $0 < p < 1$, then applying the inverse channel $U^{-1}$
 gives a nontrivial convex decomposition of $\psi$, a contradiction.
 Thus $U(\psi)$ is pure.
\end{proof}

\begin{lemma} \label{sm:lemma:basic-internal_exists}
 For every system $A$, there exists a normalized state that is an interior point
 of the state cone $\St_+(A)$ in $\St_\Real(A)$.
\end{lemma}
\begin{proof}
 Since $\St_\Real(A) = \Span \StN(A)$ is $d_A$-dimensional, choose a basis $\rho_1,\dots,\rho_{d_A}$ of $\St_\Real(A)$
 consisting of elements of $\StN(A)$.
 Put $\sigma \coloneqq \sum_{i=1}^{d_A} \rho_i / d_A$.
 Then $e_A(\sigma) = 1$, so $\sigma$ is normalized.
 Put $\cU \coloneqq \{ \sum_i \lambda_i \rho_i \mid |\lambda_i - 1/d_A| < 1/d_A ~(\forall i) \}$.
 Then $\cU$ is an open neighborhood of $\sigma$, and each element of $\cU$ has the form
 $\sum_i \lambda_i \rho_i$ $~(\lambda_i > 0)$, so $\cU \subseteq \St_+(A)$.
 Thus $\sigma$ is an interior point.
\end{proof}

\begin{lemma} \label{sm:lemma:basic-internal}
 A normalized state of $A$ is internal if and only if it is an interior point of $\St_+(A)$.
\end{lemma}
\begin{proof}
 Suppose first that $\rho$ is internal but not an interior point.
 By Lemma~\ref{sm:lemma:basic-compact_closed}, $\St_+(A)$ is a closed convex cone.
 Since $\rho \in \St_+(A)$ lies on its boundary, the supporting hyperplane theorem gives
 a nonzero functional $f \in \St_+(A)^*$ such that $f(\rho) = 0$.
 For every $\sigma \in D_\rho$, one has $0 \le f(\sigma) \le f(\rho) = 0$.
 Since $D_\rho$ spans $\St_\Real(A)$, this implies $f = \zero$, a contradiction.

 Conversely, suppose that $\rho$ is an interior point of $\St_+(A)$.
 For every $\tau \in \St_+(A)$, choose sufficiently small $\varepsilon > 0$
 with $\rho - \varepsilon\tau \in \St_+(A)$.
 Hence $\varepsilon\tau \in D_\rho$, so $\St_+(A) \subseteq \Span D_\rho$.
 Taking spans gives $\Span D_\rho = \St_\Real(A)$, and therefore $\rho$ is internal.
\end{proof}

\begin{lemma} \label{sm:lemma:basic-internal_refinement}
 For every normalized internal state $\rho \in \StN(A)$ and every $x \in \St_+(A)$, there exists $\varepsilon > 0$
 such that $\varepsilon x \in D_\rho$.
\end{lemma}
This is essentially a special case of Lemma~16 of Ref.~\cite{Chi-Dar-Per-2010}.
\begin{proof}
 If $x = \zero$, the claim is trivial.
 Otherwise, by Lemma~\ref{sm:lemma:basic-internal}, $\rho$ is an interior point of $\St_+(A)$.
 Hence $\rho - \varepsilon x \in \St_+(A)$ for sufficiently small $\varepsilon > 0$.
 Thus $\varepsilon x \in D_\rho$.
\end{proof}

\begin{lemma} \label{sm:lemma:basic-perfect_discrimination}
 For every positive integer $k$, there exists a system $E$ with a family of
 $k$ perfectly distinguishable normalized states $\{ \eta_i \in \StN(E) \}_{i=1}^k$.
\end{lemma}
\begin{proof}
 By Assumption~\ref{sm:ass:OPT-nontriviality}, choose a system $D$ with informational dimension at least $2$.
 The informational dimension $n_E$ of its $m$-fold composite $E \coloneqq D^{\ot m}$ is at least $2^m$,
 so taking $m$ sufficiently large ensures $n_E \ge k$.
 By the definition of informational dimension, we can choose a perfectly distinguishable family
 $\{ \eta_i \in \StN(E) \}_{i=1}^{n_E}$ and a measurement $\{ a_i \in \Eff(E) \}_{i=1}^{n_E}$
 that perfectly distinguishes it.
 Coarse-graining the outcomes $k,k+1,\dots,n_E$ of the measurement $\{ a_i \}_{i=1}^{n_E}$ yields
 a measurement $\{ b_i \}_{i=1}^k$, where $b_i \coloneqq a_i$ for $1 \le i < k$ and
 $b_k \coloneqq \sum_{i=k}^{n_E} a_i$.
 This measurement perfectly distinguishes $\{ \eta_i \}_{i=1}^k$.
\end{proof}

\section{Postulates} \label{sm:sec:postulates}

In this appendix, we formulate local equivalence and CDP/ES purification.
CDP purification and ES purification are alternative postulates, each combined with local equivalence.
We then define quantum theory and outline its reconstruction common to both routes.

\subsection{Local equivalence} \label{sm:subsec:postulates-LE}

Local equivalence states that channels from $A$ to $B$ can be identified by states of $A$ and effects of $B$.
\begin{postulate}[Local equivalence] \label{sm:postulate:LE}
 For arbitrary systems $A$ and $B$ and channels $\Lambda$ and $\cE$ from $A$ to $B$, if
 $a(\Lambda(\rho)) = a(\cE(\rho))$ for all $\rho \in \StN(A)$ and $a \in \Eff(B)$, then $\Lambda = \cE$.
 Equivalently,
 \begin{alignat}{1}
  \adjustbox{valign=c}{\includegraphics[alt={a(\Lambda(\rho)) = a(\cE(\rho))}]{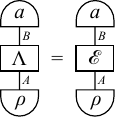}}
  \quad (\forall \rho \in \StN(A), ~a \in \Eff(B))
  &\qquad \Rightarrow \qquad
  \adjustbox{valign=c}{\includegraphics[alt={\Lambda = \cE}]{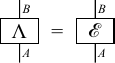}}.
  \label{sm:eq:LE-local_equivalence}
 \end{alignat}
\end{postulate}
We shall use the following equivalent formulation.
\begin{lemma} \label{sm:lemma:LE-transformation_equivalence}
 In an OPT, the following two properties are equivalent.
 \begin{enumerate}
  \item The theory satisfies local equivalence.
  \item For arbitrary systems $A$ and $B$ and transformations $f$ and $g$ from $A$ to $B$, if
        $a(f(\rho)) = a(g(\rho))$ for all $\rho \in \StN(A)$ and $a \in \Eff(B)$, then $f = g$.
 \end{enumerate}
\end{lemma}
\begin{proof}
 Since channels are special transformations, $(2) \Rightarrow (1)$ is immediate.
 We prove $(1) \Rightarrow (2)$.
 Complete $f$ to a binary test $\{ f,h \}$ as follows.
 If $f$ is a channel, set $h \coloneqq \zero$; otherwise, let $h$ be the coarse-graining
 of all other outcomes of a test containing $f$.
 Complete $g$ similarly to a binary test $\{ g,q \}$.
 Choose any $\tau \in \StN(B)$ and put $R \coloneqq \tau \c e_B \in \Trans(B,B)$, which is a channel.
 For $U = \id_B$ or $U = R$, define the following channels from $A$ to $B$:
 \begin{alignat}{1}
  \includegraphics[scale=1.0,alt={
  F \coloneqq U \c f + R \c h, \quad G \coloneqq U \c g + R \c q
  }]{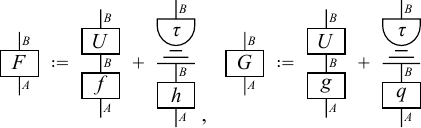}.
  \label{sm:eq:LE-auxiliary_operations}
 \end{alignat}
 The transformations $F$ and $G$ are channels because they are obtained by conditioning on the outcomes of the
 binary tests $\{ f,h \}$ and $\{ g,q \}$, respectively, and then coarse-graining all resulting outcomes.
 For every $\rho \in \StN(A)$ and $a \in \Eff(B)$, the assumption gives $e_B(f(\rho)) = e_B(g(\rho))$,
 and hence $e_B(h(\rho)) = e_B(q(\rho))$.
 Therefore
 \begin{alignat}{1}
  a(R(h(\rho))) & = a(\tau) e_B(h(\rho)) = a(\tau) e_B(q(\rho)) = a(R(q(\rho))).
 \end{alignat}
 Since $a \c U \in \Eff(B)$, the assumption also gives $a(U(f(\rho))) = a(U(g(\rho)))$.
 Hence
 \begin{alignat}{1}
  a(F(\rho)) & = a(U(f(\rho))) + a(R(h(\rho))) = a(U(g(\rho))) + a(R(q(\rho))) = a(G(\rho)).
 \end{alignat}
 Local equivalence applied to the channels $F$ and $G$ gives $F = G$.
 Taking $U = R$ gives $R \c f + R \c h = R \c g + R \c q$; since $R \c f = R \c g$, it follows that $R \c h = R \c q$.
 Taking $U = \id_B$ gives $f + R \c h = g + R \c q$, and hence $f = g$.
\end{proof}

In the standard OPT framework, equality of transformations is defined by statistics involving arbitrary ancillary
systems.
If local equivalence holds, Lemma~\ref{sm:lemma:LE-transformation_equivalence} shows that no ancillary system is needed
to identify a transformation.
For comparison, Ref.~\cite{Nak-2020} calls condition~(2) of
Lemma~\ref{sm:lemma:LE-transformation_equivalence} local equality.%
\footnote{Strictly speaking, the local equality postulate in Ref.~\cite{Nak-2020} is formulated for
arbitrary $f,g \in \Trans_+(A,B)$ rather than only for $f,g \in \Trans(A,B)$ as in
condition~(2) of Lemma~\ref{sm:lemma:LE-transformation_equivalence}.
Specifically, for any systems $A,B$ and $f,g \in \Trans_+(A,B)$, it states that
$a(f(\rho)) = a(g(\rho))$ for every $\rho \in \StN(A)$ and $a \in \Eff(B)$ implies $f = g$.
The same reference states that these two conditions are equivalent.}
By Lemma~\ref{sm:lemma:LE-transformation_equivalence}, this property is equivalent to the definition adopted here.
It is also known that local tomography, defined below, is equivalent to local equivalence~\cite[Lemma~34]{Nak-2020}.
\begin{define}[Local tomography] \label{sm:def:LE-local_tomography}
 A theory satisfies \termdef{local tomography} if, for every pair of systems $A$ and $B$,
 any two normalized states $\rho$ and $\sigma$ of $\AB$ that satisfy
 $(a \ot b)(\rho) = (a \ot b)(\sigma)$ for all $a \in \Eff(A)$ and $b \in \Eff(B)$
 must satisfy $\rho = \sigma$.
\end{define}
Local tomography can also be expressed as the dimensional identity~\cite{Chi-Dar-Per-2010}
\begin{alignat}{1}
 d_\AB & = d_A d_B.
 \label{sm:eq:LE-local_tomography_dimension}
\end{alignat}
The contrapositive formulation, namely that for any two distinct normalized states
$\rho,\sigma \in \StN(AB)$, there exist effects $a \in \Eff(A)$ and $b \in \Eff(B)$ such that
$(a \ot b)(\rho) \neq (a \ot b)(\sigma)$, is also
called \termdef{local discriminability}~\cite{Chi-Dar-Per-2010}.
Throughout the remainder of this subsection, we assume local equivalence, equivalently local tomography.

\begin{lemma} \label{sm:lemma:basic-product_span}
 For every $A$ and $B$, $\St_\Real(AB)$ is spanned by product states $\rho \ot \sigma$ with $\rho \in \StN(A)$
 and $\sigma \in \StN(B)$.
\end{lemma}
\begin{proof}
 Since $\St_\Real(A) = \Span \StN(A)$ and $\St_\Real(B) = \Span \StN(B)$,
 we can choose $\{ \rho_i \in \StN(A) \}_{i=1}^{d_A}$ and $\{ \sigma_j \in \StN(B) \}_{j=1}^{d_B}$
 as bases of $\St_\Real(A)$ and $\St_\Real(B)$, respectively.
 Let $\{ a_i \in \Eff_\Real(A) \}_{i=1}^{d_A}$ and $\{ b_j \in \Eff_\Real(B) \}_{j=1}^{d_B}$
 be their respective dual bases.
 If $\sum_{i,j} c_{i,j} \rho_i \ot \sigma_j = \zero$, then, for every $k$ and $l$,
 $c_{k,l} = (a_k \ot b_l) (\sum_{i,j} c_{i,j} \rho_i \ot \sigma_j) = 0$.
 Thus these product states are linearly independent.
 By local tomography, $d_{AB} = d_A d_B$, so they form a basis of $\St_\Real(AB)$.
\end{proof}

\begin{lemma} \label{sm:lemma:basic-transformation_linearization}
 For systems $A,B$ and $f \in \Trans_\Real(A,B)$, consider the linear map
 $\tf \colon \St_\Real(A) \ni x \mapsto f(x) \in \St_\Real(B)$.
 The correspondence $f \mapsto \tf$ is a real-linear isomorphism between $\Trans_\Real(A,B)$ and
 the real vector space of all linear maps from $\St_\Real(A)$ to $\St_\Real(B)$.
\end{lemma}
\begin{proof}
 For any $f,g \in \Trans_\Real(A,B)$ and $c,d \in \Real$, we have
 $\widetilde{cf + dg} = c\tf + d\widetilde g$, so the map $f \mapsto \tf$ is linear.
 It remains to show that $f \mapsto \tf$ is bijective.

 We first prove injectivity.
 Let $f, g \in \Trans_\Real(A,B)$ induce the same linear map
 from $\St_\Real(A)$ to $\St_\Real(B)$, and put $h \coloneqq f - g \in \Trans_\Real(A,B)$.
 Then $h(x) = \zero$ for every $x \in \St_\Real(A)$.
 It is enough to prove $h = \zero$, namely $a((h \ot \id_E)(\omega)) = 0$ for every system $E$,
 $\omega \in \StN(AE)$, and $a \in \Eff(BE)$.
 Fix such $E$, $\omega$, and $a$.
 By Lemma~\ref{sm:lemma:basic-product_span}, applied to $A$ and $E$, write
 $\omega = \sum_i c_i \rho_i \ot \eta_i$ with $c_i \in \Real$, $\rho_i \in \StN(A)$,
 and $\eta_i \in \StN(E)$.
 Then $(h \ot \id_E)(\omega) = \sum_i c_i h(\rho_i) \ot \eta_i = \zero$,
 because $h(\rho_i) = \zero$ for every $i$.
 Hence $a((h \ot \id_E)(\omega)) = 0$.
 Thus $h = \zero$, and injectivity follows.

 For surjectivity, every rank-one linear map from $\St_\Real(A)$ to $\St_\Real(B)$ has the form
 $x \mapsto a(x) \tau$ with $a \in \Eff_\Real(A) = \St_\Real(A)^*$ and $\tau \in \St_\Real(B)$,
 and is represented by $\tau \c a \in \Trans_\Real(A,B)$.
 Rank-one maps span the space of all linear maps, so the map is surjective.
\end{proof}

By this real-linear isomorphism $f \mapsto \tf$, we identify $f$ with $\tf$ below.
The map $\ot \colon \StN(A) \times \StN(B) \to \StN(AB)$ extends uniquely to a bilinear map
$\ot \colon \St_\Real(A) \times \St_\Real(B) \to \St_\Real(AB)$.
Similarly, the parallel composition $f \ot g$ of $f \in \Trans_\Real(A,B)$ and
$g \in \Trans_\Real(A',B')$ is the unique linear map satisfying
$(f \ot g)(\rho \ot \sigma) = f(\rho) \ot g(\sigma)$
$~(\forall \rho \in \StN(A),~\sigma \in \StN(A'))$.
Mathematically, the structure with systems as objects and transformations as morphisms forms a strict
symmetric monoidal category, which under local equivalence admits a faithful representation in
the symmetric monoidal category of finite-dimensional real vector spaces and linear maps.

\begin{lemma} \label{sm:lemma:basic-pure_internal_product}
 The parallel composition $\psi \ot \phi$ of normalized pure states $\psi \in \StNP(A)$ and $\phi \in \StNP(B)$ is pure.
 Also, the parallel composition $\rho \ot \sigma$ of normalized internal states $\rho \in \StN(A)$ and
 $\sigma \in \StN(B)$ is
 internal.
\end{lemma}
\begin{proof}
 The purity assertion is Lemma~18 of Ref.~\cite{Chi-Dar-Per-2010}; the internality assertion is Theorem~2 of the same
 reference.
 These results hold for any OPT satisfying local discriminability and hence apply here.
\end{proof}

\subsection{Purification} \label{sm:subsec:postulates-purification}

We first define purification.
\begin{define}[Purification] \label{sm:def:purification}
 For $\rho \in \StN(A)$ and a system $A'$, a normalized pure state $\Psi \in \StNP(AA')$ satisfying
 $(\id_A \ot e_{A'})(\Psi) = \rho$, i.e.,
 \begin{alignat}{1}
  \includegraphics[scale=1.0,alt={(\id_A \ot e_{A'})(\Psi) = \rho}]{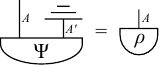}
  \label{sm:eq:purification-marginal}
 \end{alignat}
 is called a \termdef{purification} of $\rho$.
 Similarly, a normalized pure state $\Phi \in \StNP(A'A)$ satisfying
 $(e_{A'} \ot \id_A)(\Phi) = \rho$ is also called a purification of $\rho$.
 The system $A'$ is called a \termdef{purifying system}.
\end{define}
If $\Psi \in \StNP(AA')$ purifies $\rho \in \StN(A)$, define the
\termdef{complementary state} of $\rho$ with respect to $\Psi$ by
\begin{alignat}{1}
 \includegraphics[scale=1.0,alt={\trho \coloneqq (e_A \ot \id_{A'})(\Psi)}]{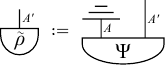}.
 \label{sm:eq:purification-marginal_complementary}
\end{alignat}
By definition, $\Psi$ is also a purification of the complementary state $\trho$.
We also consider purifications that reduce to the above definition by reversible swap channels.
For example, if $\Psi \in \StNP(BAC)$ satisfies $(e_B \ot \id_A \ot e_C)(\Psi) = \rho$,
we also call $\Psi$ a purification of $\rho \in \StN(A)$.

Chiribella--D'Ariano--Perinotti (CDP) purification is the purification postulate described in
Refs.~\cite{Chi-Dar-Per-2010,Chi-Dar-Per-2011,Dar-Chi-Per-2017}.
It requires that each normalized state $\rho \in \StN(A)$ admit a purification by some normalized pure state
$\Psi \in \StNP(AA')$, essentially uniquely.
\begin{postulate}[CDP purification] \label{sm:postulate:CDP_purification}
 For every system $A$ and $\rho \in \StN(A)$, there exists a system $A'$ such that
 $\rho$ admits a purification on $AA'$.
 Moreover, for any system $A'$, if $\Psi,\Psi' \in \StNP(AA')$ are both purifications of $\rho$,
 then there exists a reversible channel $U$ on $A'$ such that
 \begin{alignat}{1}
  \includegraphics[scale=1.0,alt={\Psi' = (\id_A \ot U)(\Psi)}]{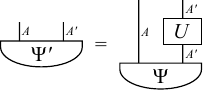}.
  \label{sm:eq:CDP-purification_uniqueness}
 \end{alignat}
\end{postulate}

ES purification requires that every normalized state $\rho \in \StN(A)$ admit a purification,
essentially uniquely, by a normalized pure state $\Psi \in \StNP(A\tA)$ with a purifying system $\tA$
equivalent to $A$.

\begin{postulate}[Equivalent-system (ES) purification] \label{sm:postulate:ES_purification}
 For every system $A$, there exists a system $\tA$ equivalent to $A$ such that every normalized state
 $\rho \in \StN(A)$ has a purification in $A\tA$.
 Moreover, if $\Psi,\Psi' \in \StNP(A\tA)$ are both purifications of $\rho$,
 then there exists a reversible channel $U$ on $\tA$ such that
 \begin{alignat}{1}
  \includegraphics[scale=1.0,alt={\Psi' = (\id_A \ot U)(\Psi)}]{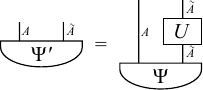}.
  \label{sm:eq:ES-purification_uniqueness}
 \end{alignat}
 Such a system $\tA$ is called a \termdef{conjugate system} of $A$.
\end{postulate}

By definition, a conjugate system $\tA$ of $A$ is equivalent to $A$.
Conversely, any system $C$ equivalent to $A$ is also a conjugate system of $A$.
Indeed, since $\tA \simeq C$, there exists a reversible channel $W \in \Trans(\tA,C)$.
If $\Psi \in \StNP(A\tA)$ purifies $\rho$, then Lemma~\ref{sm:lemma:basic-unitary_pure} gives
$\Phi \coloneqq (\id_A \ot W)(\Psi) \in \StNP(AC)$, and
$(\id_A \ot e_C)(\Phi) = (\id_A \ot e_\tA)(\Psi) = \rho$.
If $\Phi_1,\Phi_2$ both purify $\rho$ on $AC$, then
$(\id_A \ot W^{-1})(\Phi_1)$ and $(\id_A \ot W^{-1})(\Phi_2)$ are purifications
on $A\tA$ and are related by a reversible channel $U$ on $\tA$.
Thus $\Phi_1$ and $\Phi_2$ are related by the reversible channel
$W \c U \c W^{-1}$ on $C$.
Consequently, a conjugate system can be replaced by any equivalent system.
In particular, if $\tA$ is conjugate to $A$, then $A$ can also be regarded as conjugate to $\tA$.
For conjugate systems $\tA$ and $\tB$ of $A$ and $B$, the composite $\tA\tB$ is equivalent to $AB$ and hence is
a conjugate system of $AB$.
When ES purification is assumed, for each system $A$ fix one conjugate system
and denote it by $\tA$.
Whenever the purifying system is required to be a conjugate of a given system, we also allow any equivalent
system and suppress the implementing reversible equivalence in the notation.

ES purification requires, in contrast to CDP purification, that a system $\tA$ equivalent to
$A$ can serve as a purifying system for every normalized state of $A$.
On the other hand, CDP purification asserts that, for any purifying system $A'$,
if $\Psi,\Psi' \in \StNP(AA')$ both purify $\rho \in \StN(A)$, then there exists a reversible
channel $U$ on $A'$ satisfying $\Psi' = (\id_A \ot U)(\Psi)$.
ES purification requires this only for the conjugate system $\tA$, not for a general system $A'$.
Thus, ES purification imposes a stronger existence requirement, whereas CDP purification imposes
a broader uniqueness requirement.

\subsection{Quantum theory} \label{sm:subsec:postulates-quantum_theory}

In view of the examples above, we define quantum theory as the following special case of an OPT.
\begin{define}[Quantum theory] \label{sm:def:quantum_theory}
 An OPT is called \termdef{quantum theory} if one can choose, for every system $A$
 of informational dimension $n_A$, a linear isomorphism $\Repr_A \colon \St_\Real(A) \to \Her_{n_A}$
 from $\StN(A)$ to $\Den_{n_A}$ such that the following two conditions hold.
 Here $\operatorname{CPTNI}(m,n)$ denotes the set of trace-nonincreasing CP maps
 from $\Complex^{m \times m}$ to $\Complex^{n \times n}$.
 \begin{enumerate}
  \item[(1)] \phantomsection\label{sm:item:quantum_condition_1}
        For all systems $A,B$, $n_\AB = n_A n_B$.
  \item[(2)] \phantomsection\label{sm:item:quantum_condition_2}
        For all systems $A,B$, the following equality holds.%
        \footnote{More precisely, we identify a trace-nonincreasing CP map with its restriction
        from $\Complex^{m \times m}$ to the real vector space $\Her_m$.}
        \begin{alignat}{1}
         \left\{ \Repr_B \c f \c \Repr_A^{-1} \mid f \in \Trans(A,B) \right\}
         &= \operatorname{CPTNI}(n_A,n_B).
         \label{sm:eq:quantum_transformation}
        \end{alignat}
 \end{enumerate}
\end{define}
Note that the examples of quantum theory given above treated the particular case
$\StN(A) = \Den_{n_A}$ and $\Repr_A = \id_{\Her_{n_A}}$.
The quantum theories considered here need not include an $n$-level system for every positive integer $n$.
For instance, apart from the trivial level $1$, the levels of systems may be restricted to the even numbers
$2,4,6,\ldots$.
The set of allowed levels is closed under multiplication, as required for parallel composition of systems.

Choose any family of linear isomorphisms $\{ \Repr_C \colon \St_\Real(C) \to \Her_{n_C} \}_{C \in \System}$
satisfying conditions~\hyperref[sm:item:quantum_condition_1]{(1)} and
\hyperref[sm:item:quantum_condition_2]{(2)} of the above definition of quantum theory.
The following lemma describes how product states are represented under the chosen isomorphisms
(proved in Sec.~\ref{sm:subsec:quantum_lemmas-parallel_tensorator}).
\begin{lemma} \label{sm:lemma:quantum_parallel_tensorator}
 For every pair of systems $A,B$, there exists a unitary matrix $U_\AB$ of order $n_\AB$ such that
 for each $\rho \in \StN(A)$ and $\sigma \in \StN(B)$,
 \begin{alignat}{1}
  \Repr_\AB(\rho \ot \sigma) &= \Ad_{U_\AB}(\Repr_A(\rho) \ot \Repr_B(\sigma)).
  \label{sm:eq:quantum_parallel_states}
 \end{alignat}
\end{lemma}
The freedom in choosing the tensor-product identification is captured by
the unitary matrix $U_\AB$ in Lemma~\ref{sm:lemma:quantum_parallel_tensorator}.
The following lemma characterizes tests in quantum theory
(proved in Sec.~\ref{sm:subsec:quantum_lemmas-test_characterization}).
\begin{lemma} \label{sm:lemma:quantum_test_characterization}
 In quantum theory, for any systems $A,B$, a finite family $\{ f_i \in \Trans(A,B) \}_{i=1}^k$
 is a test if and only if $\sum_{i=1}^k \Repr_B \c f_i \c \Repr_A^{-1}$ is trace preserving.
\end{lemma}

Quantum theory satisfies local equivalence, CDP purification, and ES purification.
This can be readily verified from standard results in quantum theory.

\subsection{Outline of the reconstruction of quantum theory} \label{sm:subsec:postulates-roadmap}

As preparation, we define some basic terminology used in the proofs.
A convex cone $\cC$ in a real vector space $\V$ is \termdef{homogeneous} if it has an interior point and,
for any two interior points $x,y$, there exists $\Gamma_{x,y} \in \Aut(\cC)$ with $\Gamma_{x,y}(x) = y$.

\begin{define}[Quantum state-space property] \label{sm:def:quantum_state_space_property}
 An OPT satisfies the \termdef{quantum state-space property} if, for every system $A$,
 there exists a positive integer $r_A$ such that $\StN(A) \cong \Den_{r_A}$.
\end{define}

Under the standing assumption of local equivalence, Appendices~\ref{sm:sec:CDP_ES_basic}--\ref{sm:sec:simple_complex}
establish that an OPT satisfying either CDP or ES purification is quantum theory.
Here we outline the proof using Fig.~\ref{sm:fig:sm-proof-roadmap-accurate}.
Appendix~\ref{sm:sec:CDP_ES_basic} establishes basic properties derived from CDP/ES purification.
Appendix~\ref{sm:sec:E1_E2} proves Lemmas~\ref{sm:lemma:E1} and~\ref{sm:lemma:E2},
which follow readily from previous results and the basic facts of Appendix~\ref{sm:sec:CDP_ES_basic}.
The proof of Lemma~\ref{sm:lemma:E1} uses Lemma~\ref{sm:lemma:CDP_ES-pure_unitary},
which establishes pure-state transitivity, whereas that of Lemma~\ref{sm:lemma:E2} uses
Lemma~\ref{sm:lemma:CDP_ES-refinement_realization}, which establishes refinement realization.
Appendix~\ref{sm:sec:homogeneity} derives homogeneity of every state cone from CDP/ES purification
(Proposition~\ref{sm:pro:CDP_ES-cone_homogeneity}).
Appendix~\ref{sm:sec:simple_complex} uses Lemma~\ref{sm:lemma:E1} and \termdef{purifiability}
(the requirement, included in CDP/ES purification, that every normalized state admits a purification)
to derive the quantum state-space property
(Proposition~\ref{sm:pro:CDP_ES-density_matrix_state_space}).
This property, together with Lemma~\ref{sm:lemma:E2}, yields quantum theory.
The proof of homogeneity uses several basic consequences of CDP or ES purification.
Establishing homogeneity is arguably the most challenging part of our proof.
Once homogeneity has been established, the remainder of the proof uses CDP or ES purification only through
pure-state transitivity, purifiability, and refinement realization.

Local equivalence allows transformations to be identified with their corresponding linear maps on state spaces
and implies local tomography, which yields product-state spanning and dimension relations for composites.
These properties are used in the proofs of homogeneity, the quantum state-space property,
and Lemma~\ref{sm:lemma:E2}.
The symmetric-cone step does not invoke local equivalence directly.
The derivations of pure-state transitivity and refinement realization also use results established under
local tomography, including the purity of product states.

We briefly discuss the relation to earlier work.
Homogeneity of state cones has been derived using approaches involving, for example, reversible filters,
steering with systems of the same dimension, and symmetric purification
\cite{Wil-2012,Bar-Gae-Wil-2013,Bar-Udu-Wet-2023,Sel-Sca-Coe-2021}.
In contrast, we derive homogeneity directly from local equivalence and properties common to CDP and ES purification,
without imposing assumptions such as the existence of sufficiently many reversible filters
to guarantee homogeneity.
To establish the quantum state-space property, we first derive simplicity of state cones and properties of the ranks
of composite systems (these notions are discussed in Appendix~\ref{sm:sec:simple_complex}).
Related studies concerning the former include
Refs.~\cite{Bar-Hil-2020,Bar-Gra-Wil-2020,VanL-Sch-Sto-Wer-2024},
while those concerning the latter include Refs.~\cite{Bar-Gra-Wil-2020,Bar-Udu-Wet-2023}.
Here, however, we derive simplicity and the rank properties under assumptions different from those in these studies.
We then use these properties and the classification of Euclidean Jordan algebras to obtain
the quantum state-space property.
This step employs an approach also found, for example, in
Refs.~\cite{Ara-1980,Wet-2019,Nie-2020,Sel-Sca-Coe-2021}.

\begin{figure}[t]
 \centering
 \begin{tikzpicture}[>=Latex, every node/.style = {font = \small},
   block/.style = {draw, rounded corners = 1pt, align = center, inner sep = 2pt,
    minimum height = 9mm},
   arr/.style = {-{Latex[length = 2mm, width = 1.4mm]}, semithick}]
  \node[block, text width = 44mm] (root) at (0,0)
   {CDP/ES purification\\ {\footnotesize (Appendix~\ref{sm:sec:postulates},
   Postulates~\ref{sm:postulate:CDP_purification},~\ref{sm:postulate:ES_purification})}};
  \node[block, text width = 38mm] (hom) at (-5.9,-2.0)
   {Homogeneity\\ {\footnotesize (Appendix~\ref{sm:sec:homogeneity},
   Proposition~\ref{sm:pro:CDP_ES-cone_homogeneity})}};
  \node[block, text width = 34mm] (pst) at (-1.7,-2.0)
   {Pure-state transitivity\\ {\footnotesize (Appendix~\ref{sm:sec:CDP_ES_basic},
   Lemma~\ref{sm:lemma:CDP_ES-pure_unitary})}};
  \node[block, text width = 32mm] (pur) at (2.2,-2.0)
   {Purifiability};
  \node[block, text width = 34mm] (refine) at (6.1,-2.0)
   {Refinement realization\\ {\footnotesize (Appendix~\ref{sm:sec:CDP_ES_basic},
   Lemma~\ref{sm:lemma:CDP_ES-refinement_realization})}};
  \node[block, text width = 34mm] (sym) at (-3.65,-4.2)
   {Symmetric cones\\ {\footnotesize (Appendix~\ref{sm:sec:E1_E2},
   Lemma~\ref{sm:lemma:E1})}};
  \node[block, text width = 47mm] (den) at (0.7,-6.3)
   {Quantum state-space property\\ {\footnotesize (Appendix~\ref{sm:sec:simple_complex},
   Proposition~\ref{sm:pro:CDP_ES-density_matrix_state_space})}};
  \node[block, text width = 36mm] (qt) at (5.2,-8.3)
   {Quantum theory\\ {\footnotesize (Appendix~\ref{sm:sec:simple_complex},
   Lemma~\ref{sm:lemma:E2})}};
  \draw[arr] ($(root.south) + (-12mm,0)$) -- (hom.north);
  \draw[arr] ($(root.south) + (-4mm,0)$) -- (pst.north);
  \draw[arr] ($(root.south) + (4mm,0)$) -- (pur.north);
  \draw[arr] ($(root.south) + (12mm,0)$) -- (refine.north);
  \draw[arr] (hom.south) -- ($(sym.north) + (-8mm,0)$);
  \draw[arr] (pst.south) -- ($(sym.north) + (8mm,0)$);
  \draw[arr] (sym.south) -- ($(den.north) + (-12mm,0)$);
  \draw[arr] (pur.south) -- (pur.south |- den.north);
  \draw[arr] ($(den.south) + (9mm,0)$) -- ($(qt.north) + (-10mm,0)$);
  \draw[arr] (refine.south) -- (refine.south |- qt.north);
 \end{tikzpicture}
 \caption{Proof roadmap from CDP or ES purification to quantum theory, assuming local equivalence.
 Each arrow indicates a derivation using the property at its starting node.
 Where several arrows enter a node, the properties at their starting nodes are used jointly.
 Parenthetical references identify where the postulates are stated or the corresponding properties are established.}
 \label{sm:fig:sm-proof-roadmap-accurate}
\end{figure}

\section{Basic properties derived from CDP/ES purification} \label{sm:sec:CDP_ES_basic}

\subsection{Basic properties} \label{sm:subsec:CDP_ES_basic-lemmas}

Throughout this subsection, we consider an OPT satisfying local equivalence and either CDP or ES purification.
We establish basic properties common to these two settings.
Under CDP purification, the methods of Ref.~\cite{Chi-Dar-Per-2010} apply.
Under ES purification, several arguments resemble those used for CDP purification, but some proofs require modification.

\begin{lemma}[Pure-state transitivity] \label{sm:lemma:CDP_ES-pure_unitary}
 For any two normalized pure states $\psi,\phi \in \StNP(A)$,
 there exists a reversible channel $U$ on $A$ such that $\phi = U(\psi)$.
\end{lemma}
\begin{proof}
 Under CDP purification, this is Lemma~20 of Ref.~\cite{Chi-Dar-Per-2010}.
 We prove the claim under ES purification.
 Let $\tA$ be equivalent to $A$ and choose a normalized pure state
 $\alpha \in \StNP(\tA)$.
 By Lemma~\ref{sm:lemma:basic-pure_internal_product}, both
 $\alpha \ot \psi$ and $\alpha \ot \phi$ are pure.
 Moreover,
 \begin{alignat}{1}
  \includegraphics[scale=1.0,alt={(\id_\tA \ot e_A)(\alpha \ot \psi) = \alpha}]{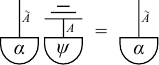}
  \label{sm:eq:CDP_ES-pure_unitary_purification}
 \end{alignat}
 shows that $\alpha \ot \psi$ purifies $\alpha$; similarly,
 $\alpha \ot \phi$ purifies $\alpha$.
 ES purification gives a reversible channel $U$ on $A$ such that
 \begin{alignat}{1}
  \includegraphics[scale=1.0,alt={\alpha \ot \phi = (\id_\tA \ot U)(\alpha \ot \psi)}]{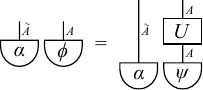}.
  \label{sm:eq:CDP_ES-pure_unitary_uniqueness}
 \end{alignat}
 Applying $e_\tA \ot \id_A$ yields $\phi = U(\psi)$.
\end{proof}

\begin{lemma} \label{sm:lemma:CDP_ES-channel}
 Let $A,B,C$ be systems, and suppose that $\Psi \in \StNP(AB)$ and
 $\Psi' \in \StNP(AC)$ are purifications of the same $\rho \in \StN(A)$.
 Then there exists a channel $\Lambda \in \Trans(B,C)$ such that
 $\Psi' = (\id_A \ot \Lambda)(\Psi)$.
\end{lemma}
\begin{proof}
 Under CDP purification, the statement is Lemma~21 of Ref.~\cite{Chi-Dar-Per-2010}.
 We establish it under ES purification.
 Choose normalized pure states
 $\alpha \in \StNP(\tA)$, $\beta \in \StNP(B)$, and $\gamma \in \StNP(C)$.
 Regard $\Phi \coloneqq \Psi \ot \alpha \ot \gamma$ and
 $\Phi' \coloneqq \Psi' \ot \alpha \ot \beta$ as states of $A\tA BC$
 after permuting tensor factors.
 Hereafter, using the assumption $AB \simeq BA$, we identify systems up to the necessary swaps,
 without displaying these swaps explicitly.
 By Lemma~\ref{sm:lemma:basic-pure_internal_product}, $\Phi$ and $\Phi'$ are pure,
 and both purify $\rho$ with purifying system $\tA BC$.
 Choose a normalized pure state $\varphi \in \StNP(\tC\tB)$.
 Again by Lemma~\ref{sm:lemma:basic-pure_internal_product},
 $\varphi \ot \Phi$ and $\varphi \ot \Phi'$ are pure.
 Furthermore,
 \begin{alignat}{1}
  \includegraphics[scale=1.0,alt={
  (\id_{\tC\tB A} \ot e_{\tA BC})(\varphi \ot \Phi) = \varphi \ot \rho
  }]{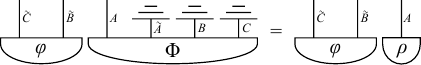}
  \label{sm:eq:ES-composite_unitary_purification}
 \end{alignat}
 shows that $\varphi \ot \Phi$ is a purification of $\varphi \ot \rho$,
 and likewise for $\varphi \ot \Phi'$.
 ES purification thus provides a reversible channel $U$ on $\tA BC$ such that
 \begin{alignat}{1}
  \includegraphics[scale=1.0,alt={
  \varphi \ot \Phi' = (\id_{\tC\tB A} \ot U)(\varphi \ot \Phi)
  }]{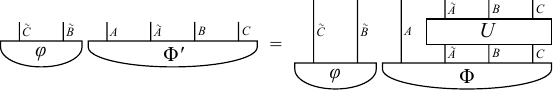}.
  \label{sm:eq:ES-composite_unitary_uniqueness}
 \end{alignat}
 Applying $e_{\tC\tB} \ot \id_{A\tA BC}$ to both sides gives
 \begin{alignat}{1}
  \includegraphics[scale=1.0,alt={\Phi' = (\id_A \ot U)(\Phi)}]{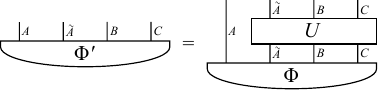}.
  \label{sm:eq:CDP_ES-channel_uniqueness}
 \end{alignat}
 Define $\Lambda \in \Trans(B,C)$ by
 \begin{alignat}{1}
  \includegraphics[scale=1.0,alt={
  \Lambda \coloneqq (e_\tA \ot e_B \ot \id_C) \c U \c (\alpha \ot \id_B \ot \gamma)
  }]{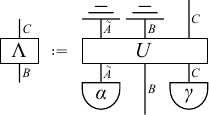}.
  \label{sm:eq:CDP_ES-channel_Lambda}
 \end{alignat}
 As a sequential composition of channels, $\Lambda$ is a channel, and
 \begin{alignat}{1}
  \footnoteinset{-1.75}{0.3}{\eqref{sm:eq:CDP_ES-channel_uniqueness}}{%
  \footnoteinset{5.25}{0.3}{\eqref{sm:eq:CDP_ES-channel_Lambda}}{%
  \includegraphics[scale=1.0,alt={\Psi' = (\id_A \ot \Lambda)(\Psi)}]{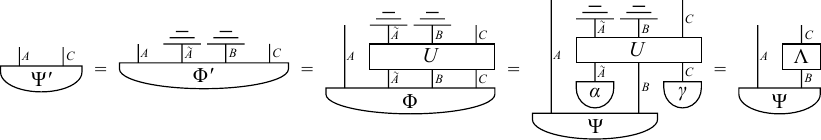}}}
  \label{sm:eq:CDP_ES-channel_Psi}
 \end{alignat}
 yields $\Psi' = (\id_A \ot \Lambda)(\Psi)$.
 The equation numbers written above equality signs, here
 \eqref{sm:eq:CDP_ES-channel_uniqueness} and~\eqref{sm:eq:CDP_ES-channel_Lambda},
 indicate identities used in the corresponding steps.
 The same convention applies to subsequent diagrams.
\end{proof}

\begin{lemma} \label{sm:lemma:CDP_ES-finite_steering}
 For systems $A,B$, let $\rho \in \StN(A)$ have a purification
 $\Psi \in \StNP(AB)$, and consider a refinement
 $\{ \sigma_i \in \St(A) \}_{i=1}^n$ of $\rho$.
 Then there exists a measurement $\{ b_i \}_{i=1}^n$ on $B$ such that
 $(\id_A \ot b_i)(\Psi) = \sigma_i$ for every $i$.
\end{lemma}
\begin{proof}
 Under CDP purification, this is Theorem~6 of Ref.~\cite{Chi-Dar-Per-2010}.
 We give a proof for ES purification.
 If $n = 1$, take $b_1 \coloneqq e_B$.
 Assume $n \ge 2$.
 By Lemma~\ref{sm:lemma:basic-perfect_discrimination}, there exists a system $C$ with
 a family of $n$ normalized states $\{ \varphi_i \}_{i=1}^n$ and a measurement
 $\{ c_i \}_{i=1}^n$ that perfectly distinguishes them.
 Put
 \begin{alignat}{1}
  \includegraphics[scale=1.0,alt={\Omega \coloneqq \sum_{i=1}^n \sigma_i \ot \varphi_i}]{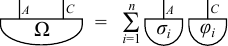}.
  \label{sm:eq:CDP_ES-finite_steering_Omega}
 \end{alignat}
 Then $e_{AC}(\Omega) = \sum_{i=1}^n e_A(\sigma_i)e_C(\varphi_i) = e_A(\rho) = 1$.
 Since $\Omega \in \St_+(AC)$, Eq.~\eqref{sm:eq:OPT-state_cone_normalization} gives $\Omega \in \StN(AC)$.
 Let $D$ be a conjugate system of $AC$, and choose a purification
 $\Phi \in \StNP(ACD)$ of $\Omega$.
 \begin{alignat}{1}
  \includegraphics[scale=1.0,alt={
  (\id_A \ot e_C \ot e_D)(\Phi) = (\id_A \ot e_C)(\Omega) = \rho
  }]{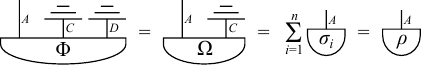}.
  \label{sm:eq:CDP_ES-finite_steering_Phi}
 \end{alignat}
 Hence $\Phi$ also purifies $\rho$.
 Lemma~\ref{sm:lemma:CDP_ES-channel} provides a channel
 $\Lambda \in \Trans(B,CD)$ satisfying
 $\Phi = (\id_A \ot \Lambda)(\Psi)$.
 Define
 \begin{alignat}{1}
  \includegraphics[scale=1.0,alt={b_i \coloneqq (c_i \ot e_D) \c \Lambda}]{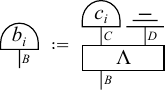}.
  \label{sm:eq:CDP_ES-finite_steering_b}
 \end{alignat}
 Since $\Lambda$ is a channel and $\{ c_i \ot e_D \}_{i=1}^n$ is a measurement on $CD$,
 $\{ b_i \}_{i=1}^n$ is a measurement on $B$.
 Moreover,
 \begin{alignat}{1}
  \footnoteinset{-3.40}{0.3}{\eqref{sm:eq:CDP_ES-finite_steering_b}}{%
  \footnoteinset{4.13}{0.3}{\eqref{sm:eq:CDP_ES-finite_steering_Omega}}{%
  \includegraphics[scale=1.0,alt={
  (\id_A \ot b_i)(\Psi) = (\id_A \ot ((c_i \ot e_D) \c \Lambda))(\Psi)
  = (\id_A \ot c_i \ot e_D)(\Phi) = (\id_A \ot c_i)(\Omega) = \sigma_i
  }]{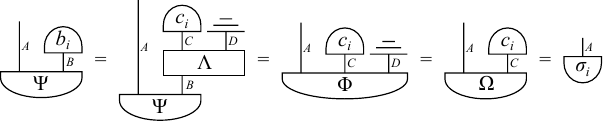}}}.
  \label{sm:eq:CDP_ES-finite_steering_sigma}
 \end{alignat}
\end{proof}

\begin{lemma} \label{sm:lemma:CDP_ES-steering}
 For systems $A,B$, let $\Psi \in \StNP(AB)$ purify
 $\rho \in \StN(A)$.
 For any refinement component $\sigma$ of $\rho$, there exists
 an effect $b \in \Eff(B)$ such that $\sigma = (\id_A \ot b)(\Psi)$.
\end{lemma}
\begin{proof}
 The pair $\{ \sigma, \rho - \sigma \}$ is a refinement of $\rho$.
 By Lemma~\ref{sm:lemma:CDP_ES-finite_steering}, there is a measurement
 $\{ b,c \}$ on $B$ with $\sigma = (\id_A \ot b)(\Psi)$.
 Since $b$ is an outcome of this measurement, it is an effect.
 Under CDP purification, the conclusion also follows from Corollary~9 of
 Ref.~\cite{Chi-Dar-Per-2010}.
\end{proof}

\begin{define}[Refinement realization] \label{sm:def:refinement_realization}
 An OPT satisfies \termdef{refinement realization} if, for all systems $A,B,C$,
 every normalized internal state $\omega \in \StN(A)$, purification
 $\Psi \in \StNP(AB)$ of $\omega$, and state $\sigma \in \St(CB)$ satisfying
 $(e_C \ot \id_B)(\sigma) \le (e_A \ot \id_B)(\Psi)$,
 there exists a transformation $K \in \Trans(A,C)$ such that
 $(K \ot \id_B)(\Psi) = \sigma$.
\end{define}

\begin{lemma}[Refinement realization] \label{sm:lemma:CDP_ES-refinement_realization}
 Under local equivalence, either CDP or ES purification implies refinement realization.
\end{lemma}
\begin{proof}
 Under CDP purification, Theorem~8 of Ref.~\cite{Chi-Dar-Per-2010} shows that
 every purification of a normalized internal state is dynamically faithful.
 Thus the claim follows from Theorem~17(3) of that reference.
 We establish it under ES purification.
 Take arbitrary systems $A,B$, a normalized internal state $\omega \in \StN(A)$,
 a purification $\Psi \in \StNP(AB)$ of $\omega$, and its complementary state
 $\tomega \coloneqq (e_A \ot \id_B)(\Psi) \in \StN(B)$.
 Choose any system $C$ and state $\sigma \in \St(CB)$ whose marginal
 $\tau \coloneqq (e_C \ot \id_B)(\sigma) \in \St(B)$ satisfies $\tau \le \tomega$.
 Choose any normalized state $\gamma \in \StN(C)$ and put
 $\xi \coloneqq \sigma + \gamma \ot (\tomega - \tau) \in \St_+(CB)$.
 Since $e_{CB}(\xi) = e_B(\tau) + e_B(\tomega - \tau) = 1$,
 Eq.~\eqref{sm:eq:OPT-state_cone_normalization} gives $\xi \in \StN(CB)$.
 Moreover, $(e_C \ot \id_B)(\xi) = \tau + (\tomega - \tau) = \tomega$.
 $\xi - \sigma = \gamma \ot (\tomega - \tau) \in \St_+(CB)$,
 so $\sigma \in D_\xi$.
 Choose a purifying system $D$ and a purification $\Phi \in \StNP(DCB)$ of $\xi$.
 By Lemma~\ref{sm:lemma:CDP_ES-steering}, there is an effect $d \in \Eff(D)$ such that
 \begin{alignat}{1}
  \includegraphics[scale=1.0,alt={(d \ot \id_{CB})(\Phi) = \sigma}]{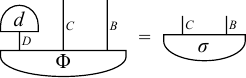}.
  \label{sm:eq:refinement_realization_d}
 \end{alignat}
 On the other hand, $\Phi$ and $\Psi$ both purify $\tomega$.
 By Lemma~\ref{sm:lemma:CDP_ES-channel}, there is a channel
 $\Lambda \in \Trans(A,DC)$ such that
 \begin{alignat}{1}
  \includegraphics[scale=1.0,alt={(\Lambda \ot \id_B)(\Psi) = \Phi}]{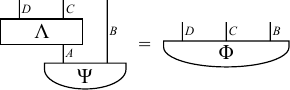}.
  \label{sm:eq:refinement_realization_Lambda}
 \end{alignat}
 Thus, setting $K \coloneqq (d \ot \id_C) \c \Lambda \in \Trans(A,C)$ gives
 \begin{alignat}{1}
  \footnoteinset{0.00}{0.3}{\eqref{sm:eq:refinement_realization_Lambda}}{%
  \footnoteinset{2.76}{0.3}{\eqref{sm:eq:refinement_realization_d}}{%
  \includegraphics[scale=1.0,alt={(K \ot \id_B)(\Psi) = (((d \ot \id_C) \c \Lambda) \ot \id_B)(\Psi) = (d \ot \id_{CB})(\Phi) = \sigma
  }]{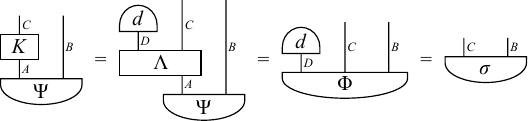}}}.
 \end{alignat}
\end{proof}

\begin{lemma} \label{sm:lemma:CDP_ES-purification_transformation_equality}
 For systems $A,B,C$, a normalized internal state $\rho \in \StN(A)$,
 and its purification $\Psi \in \StNP(AB)$, the following implication holds for all
 $f,g \in \Trans_\Real(A,C)$:
 \begin{alignat}{1}
  \adjustbox{valign=c}{\includegraphics[alt={(f \ot \id_B)(\Psi) = (g \ot \id_B)(\Psi)}]{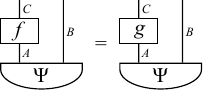}}
  &\quad \Rightarrow \quad
  \adjustbox{valign=c}{\includegraphics[alt={f = g}]{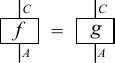}}.
  \label{sm:eq:CDP_ES-purification_transformation_equality_fg}
 \end{alignat}
\end{lemma}
The converse is immediate.
This is the linear-map version of Theorem~8 of Ref.~\cite{Chi-Dar-Per-2010}.
As explained in Sec.~\ref{sm:subsec:postulates-LE},
$f \ot \id_B \in \Trans_\Real(AB,CB)$ is the uniquely determined real-linear map from
$\St_\Real(AB)$ to $\St_\Real(CB)$ satisfying
$(f \ot \id_B)(\rho \ot \sigma) = f(\rho) \ot \sigma$
$~(\forall \rho \in \StN(A), ~\sigma \in \StN(B))$; likewise for $g \ot \id_B$.
\begin{proof}
 For every $\sigma \in D_\rho$, Lemma~\ref{sm:lemma:CDP_ES-steering} gives
 $b \in \Eff(B)$ with $\sigma = (\id_A \ot b)(\Psi)$.
 Hence $f(\sigma) = (f \ot b)(\Psi) = (g \ot b)(\Psi) = g(\sigma)$.
 Since $\rho$ is internal, $D_\rho$ spans $\St_\Real(A)$, and hence $f = g$.
\end{proof}

Among the lemmas in this appendix, only the following one assumes ES purification alone.
\begin{lemma} \label{sm:lemma:ES-complementary_internal}
 In an OPT satisfying local equivalence and ES purification,
 let $A$ be a system, $\rho \in \StN(A)$ a normalized internal state,
 and $\Psi \in \StNP(A\tA)$ a purification of $\rho$.
 Then the complementary state $\trho \coloneqq (e_A \ot \id_\tA)(\Psi)$ is internal.
\end{lemma}
This statement is essentially obtained by combining Theorems~8 and~9 and Lemma~30 of
Ref.~\cite{Chi-Dar-Per-2010}.
For comparison, under CDP purification, the complementary state $\trho \in \StN(B)$ of a normalized internal state
$\rho \in \StN(A)$ is guaranteed to be internal when $B \simeq A$.
\begin{proof}
 To prove that $\trho$ is internal, it suffices to show that
 $\dim \Span(D_{\trho}) = d_\tA = d_A$.
 Indeed, a reversible channel between $A$ and $\tA$ induces a linear isomorphism
 $\St_\Real(A) \cong \St_\Real(\tA)$, so $d_\tA = d_A$.
 Consider the linear map $\widehat{\Psi}$:
 \begin{alignat}{1}
  \widehat{\Psi} &\colon \Eff_\Real(A) \ni \adjustbox{valign=c}{\includegraphics[alt={x}]{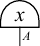}}
  \mapsto \adjustbox{valign=c}{\includegraphics[alt={(x \ot \id_\tA)(\Psi)}]{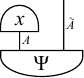}} \in \St_\Real(\tA).
  \label{sm:eq:basic_hat_Psi}
 \end{alignat}
 Take any $a \in \Eff_\Real(A)$ satisfying $\widehat{\Psi}(a) = \zero$.
 For every $\sigma \in D_\rho$, Lemma~\ref{sm:lemma:CDP_ES-steering} gives
 $b \in \Eff(\tA)$ with $\sigma = (\id_A \ot b)(\Psi)$.
 Consequently,
 \begin{alignat}{1}
  \includegraphics[scale=1.0,alt={a(\sigma) = (a \ot b)(\Psi)}]{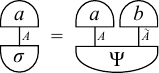}.
  \label{sm:eq:ES-complementary_internal_a_sigma}
 \end{alignat}
 The right-hand side equals $b(\widehat{\Psi}(a)) = 0$.
 Since $D_\rho$ spans $\St_\Real(A)$, we obtain $a = \zero$.
 Hence $\widehat{\Psi}$ is injective and
 $\dim \widehat{\Psi}(\Eff_\Real(A)) = d_A$.
 For each $a \in \Eff(A)$, $a \le e_A$ gives
 $\widehat{\Psi}(a) \le \trho$, so $\widehat{\Psi}(a) \in D_{\trho}$.
 Thus $\widehat{\Psi}(\Eff_\Real(A)) \subseteq \Span D_{\trho}$.
 Therefore,
 $d_A = \dim \widehat{\Psi}(\Eff_\Real(A))
 \le \dim \Span D_{\trho} \le d_\tA = d_A$,
 which proves $\dim \Span D_{\trho} = d_\tA$.
\end{proof}

\begin{define}[Reversible dilation] \label{sm:def:reversible_dilation}
 For a channel $N \in \Trans(A,B)$, suppose that there exist systems $E,F$,
 a reversible channel $W \in \Trans(EA,FB)$, and a normalized pure state
 $\alpha \in \StNP(E)$ such that
 \begin{alignat}{1}
  \includegraphics[scale=1.0,alt={N = (e_F \ot \id_B) \c W \c (\alpha \ot \id_A)}]{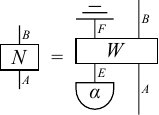}.
  \label{sm:eq:CDP_ES-reversible_dilation}
 \end{alignat}
 Then $W \c (\alpha \ot \id_A) \in \Trans(A,FB)$ is called a
 \termdef{reversible dilation} of $N$.
\end{define}

\begin{lemma} \label{sm:lemma:CDP_ES-reversible_dilation}
 Every channel $N \in \Trans(A,B)$ admits a reversible dilation.
 In particular, under ES purification it admits one from $A$ to $\tB AB$.
\end{lemma}
\begin{proof}
 Under CDP purification, this is Theorem~15 of Ref.~\cite{Chi-Dar-Per-2010}.
 We prove the statement under ES purification.
 Choose an arbitrary normalized internal state $\omega \in \StN(\tA)$ and
 its purification $\Psi \in \StNP(A\tA)$.
 Since $N$ is a channel, $(N \ot \id_\tA)(\Psi) \in \StN(B\tA)$.
 Let $F \coloneqq \tB A$ and let $\Phi \in \StNP(FB\tA)$ be a purification of this state, so that
 \begin{alignat}{1}
  \includegraphics[scale=1.0,alt={(e_F \ot \id_{B\tA})(\Phi) = (N \ot \id_\tA)(\Psi)}]{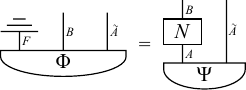}.
  \label{sm:eq:CDP_ES-reversible_dilation_Phi}
 \end{alignat}
 Set $E \coloneqq \tB B$, and choose a normalized pure state $\alpha \in \StNP(E)$.
 By Lemma~\ref{sm:lemma:basic-pure_internal_product}, both
 $\alpha \ot \Psi \ot \alpha \in \StNP(EA\tA E)$ and
 $\Phi \ot \alpha \in \StNP(FB\tA E)$ are pure.
 They both purify $\omega \ot \alpha \in \StN(\tA E)$.
 Their purifying systems $EA$ and $FB$ are both equivalent to $\tA E$.
 Hence ES purification supplies a reversible channel
 $W \in \Trans(EA,FB)$ satisfying
 \begin{alignat}{1}
  \includegraphics[scale=1.0,alt={
  (W \ot \id_{\tA E})(\alpha \ot \Psi \ot \alpha) = \Phi \ot \alpha
  }]{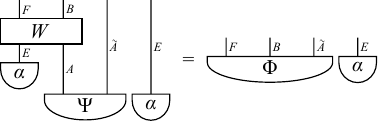}.
  \label{sm:eq:CDP_ES-reversible_dilation_W}
 \end{alignat}
 Applying $e_F \ot \id_{B\tA} \ot e_E$ to both sides yields
 \begin{alignat}{1}
  \footnoteinset{-1.01}{0.3}{\eqref{sm:eq:CDP_ES-reversible_dilation_W}}{%
  \footnoteinset{1.75}{0.3}{\eqref{sm:eq:CDP_ES-reversible_dilation_Phi}}{%
  \includegraphics[scale=1.0,alt={(((e_F \ot \id_B) \c W \c (\alpha \ot \id_A)) \ot \id_\tA)(\Psi) = (N \ot \id_\tA)(\Psi)
  }]{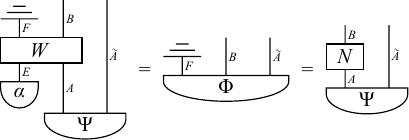}}}.
  \label{sm:eq:CDP_ES-reversible_dilation_proof}
 \end{alignat}
 By Lemma~\ref{sm:lemma:ES-complementary_internal},
 the complementary state $(\id_A \ot e_\tA)(\Psi)$ of $\omega$ is internal.
 Lemma~\ref{sm:lemma:CDP_ES-purification_transformation_equality}
 gives Eq.~\eqref{sm:eq:CDP_ES-reversible_dilation}.
\end{proof}

\subsection{Properties equivalent to purification} \label{sm:subsec:CDP_ES_basic-equivalence}

\begin{cor} \label{sm:cor:CDP_ES-reversible_dilation}
 In an OPT satisfying local equivalence, the following two conditions are equivalent.
 \begin{enumerate}
  \item CDP purification holds.
  \item For each pair of systems $A,B$, every channel $N \in \Trans(A,B)$ admits a reversible dilation.
        Moreover, for any system $F$ and any two reversible dilations $f,f' \in \Trans(A,FB)$ of $N$,
        there exists a reversible channel $U$ on $F$ such that
        \begin{alignat}{1}
         \includegraphics[scale=1.0,alt={f' = (U \ot \id_B) \c f}]{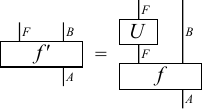}.
         \label{sm:eq:CDP_ES-reversible_dilation_uniq}
        \end{alignat}
 \end{enumerate}
 Likewise, the following two conditions are equivalent.
 \begin{enumerate}
  \item[(3)] ES purification holds.
  \item[(4)] For each pair of systems $A,B$, let $\tB$ be a system equivalent to $B$
        and put $F \coloneqq \tB A$.
        Every channel $N \in \Trans(A,B)$ admits a reversible dilation from $A$ to $FB$.
        Moreover, for any two reversible dilations $f,f' \in \Trans(A,FB)$ of $N$,
        there exists a reversible channel $U$ on $F$ satisfying
        Eq.~\eqref{sm:eq:CDP_ES-reversible_dilation_uniq}.
 \end{enumerate}
\end{cor}
\begin{proof}
 \noindent\textbf{(1) $\Leftrightarrow$ (2):}
 This is shown in Corollary~23 of Ref.~\cite{Chi-Dar-Per-2010}.

 \noindent\textbf{(3) $\Rightarrow$ (4):}
 Existence follows from Lemma~\ref{sm:lemma:CDP_ES-reversible_dilation}.
 To prove uniqueness, take any two reversible dilations $f,f' \in \Trans(A,FB)$ of $N$
 and a purification $\Psi \in \StNP(A\tA)$ of a normalized internal state $\omega \in \StN(A)$.
 Since $f$ and $f'$ are reversible dilations, $(f \ot \id_\tA)(\Psi)$ and
 $(f' \ot \id_\tA)(\Psi)$ both purify $(N \ot \id_\tA)(\Psi)$.
 The uniqueness clause of ES purification gives a reversible channel $U$ on $F$ satisfying
 \begin{alignat}{1}
  \footnoteinset{0.00}{0.3}{\eqref{sm:eq:ES-purification_uniqueness}}{%
  \includegraphics[scale=1.0,alt={
  (f' \ot \id_{\tA})(\Psi) = (((U \ot \id_B) \c f) \ot \id_{\tA})(\Psi)
  }]{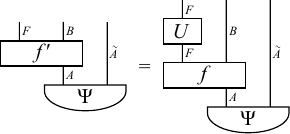}}.
  \label{sm:eq:CDP_ES-reversible_dilation_uniq_Psi}
 \end{alignat}
 As $\omega$ is internal, Lemma~\ref{sm:lemma:CDP_ES-purification_transformation_equality}
 yields Eq.~\eqref{sm:eq:CDP_ES-reversible_dilation_uniq}.

 \noindent\textbf{(4) $\Rightarrow$ (3):}
 Regard any normalized state $\rho \in \StN(B)$ as a channel from $I$ to $B$.
 Applying (4) with $A = I$, we obtain a reversible dilation
 $\Phi \in \StNP(\tB B)$ of $\rho$ satisfying $(e_{\tB} \ot \id_B)(\Phi) = \rho$.
 Indeed, a purification $\Phi$ is such a dilation by choosing
 $E \coloneqq \tB B$, $F \coloneqq \tB$, $\alpha \coloneqq \Phi$, and $W \coloneqq \id_{\tB B}$.
 Any two purifications $\Phi,\Phi' \in \StNP(\tB B)$ of $\rho$ are also reversible dilations of $\rho$.
 Thus (4) gives a reversible channel $U$ on $\tB$ with $\Phi' = (U \ot \id_B)(\Phi)$.
 Hence ES purification holds.
\end{proof}

Conditions (2) and (4) of Corollary~\ref{sm:cor:CDP_ES-reversible_dilation} allow us to interpret
the irreversibility of a channel as an apparent loss of information caused by discarding an environment,
rather than as the destruction of information itself (see Ref.~\cite{Chi-Dar-Per-2010}).
Indeed, a channel $N \in \Trans(A,B)$ can be implemented by adjoining an ancillary system in a
normalized pure state, applying a reversible channel $W$, and discarding the output environment $F$.
If $F$ is retained along with the output system $B$, applying $W^{-1}$ recovers the input state
together with the initial ancillary state, so information is preserved in the composite system $FB$.
This corresponds to Stinespring's theorem in quantum theory~\cite{Sti-1955}.
In particular, condition (4) allows the same output environment $F \coloneqq \tB A$
to be used for every channel from $A$ to $B$.
Since $F \simeq AB$, an output environment equivalent to the composite system $AB$
of the input and output systems is sufficient.
There is also a difference between conditions (2) and (4) in their uniqueness requirements.
Condition (2) requires essential uniqueness for any output environment $F$ whenever two
reversible dilations with that environment exist, whereas condition (4) requires it only for
$F \coloneqq \tB A$.
This reflects the difference between the uniqueness requirements of CDP and ES purification.
Our main theorems, together with the fact that quantum theory satisfies both purification postulates,
show that, within the present OPT framework satisfying local equivalence,
each of conditions (1)--(4) is equivalent to the theory being quantum theory.

\section{Lemmas E1 and E2} \label{sm:sec:E1_E2}

In this appendix, we state and prove Lemmas E1 and E2, which are used to reconstruct quantum theory.
Their proofs follow directly from existing results and the basic facts established in
Appendix~\ref{sm:sec:CDP_ES_basic}.

We first consider Lemma~\ref{sm:lemma:E1}.
A convex cone $\cC \subseteq \V$ in a real vector space $\V$ is called \termdef{self-dual} if
there exists an inner product $\braket{\cdot,\cdot}$ on $\V$ such that
$\{ y \in \V \mid \braket{y,x} \ge 0~(\forall x \in \cC) \} = \cC$.
A convex cone is called a \termdef{symmetric cone} if it is homogeneous and self-dual.
\begingroup
\renewcommand{\thethm}{E1}
\begin{lemma} \label{sm:lemma:E1}
 In an OPT satisfying local equivalence, if every state cone $\St_+(A)$ is homogeneous
 and pure-state transitivity (Lemma~\ref{sm:lemma:CDP_ES-pure_unitary}) holds,
 then every $\St_+(A)$ is a symmetric cone.
\end{lemma}
\endgroup
\begin{proof}
 Define an order on $\St_\Real(A)$ by $x \le y$ if and only if $y - x \in \St_+(A)$.
 A reversible channel $U$ is a normalized order isomorphism in the sense of Ref.~\cite{Bar-Udu-Wet-2023}.
 Indeed, $U$ is a linear bijection satisfying $U(\St_+(A)) = \St_+(A)$, and Lemma~\ref{sm:lemma:basic-deterministic}
 gives $e_A \c U = e_A$.
 Thus the theory satisfies pure transitivity in the sense of Ref.~\cite{Bar-Udu-Wet-2023}.
 Since $\St_+(A)$ is homogeneous by assumption, Theorem~25 of Ref.~\cite{Bar-Udu-Wet-2023} implies that $\St_+(A)$ is
 self-dual.
 Hence it is homogeneous and self-dual, and therefore symmetric.
\end{proof}

We next turn to Lemma~\ref{sm:lemma:E2}.
Assume the quantum state-space property of Definition~\ref{sm:def:quantum_state_space_property}.
For every system $A$, there exists a linear isomorphism
$\Repr_A \colon \St_\Real(A) \to \Her_{r_A}$ from $\StN(A)$ to $\Den_{r_A}$.
We call such an isomorphism a \termdef{representation} of $A$.
We first state a preparatory lemma (proved in Appendix~\ref{sm:sec:tensorator}).

\begin{lemma} \label{sm:lemma:complex_monoidal_tensorator}
 Consider an OPT satisfying local equivalence and the quantum state-space property.
 There exist a family of representations
 $\{ \Repr_C \colon \St_\Real(C) \to \Her_{r_C} \}_{C \in \System}$
 and, for every pair of systems $A,B$, a unitary matrix $U_\AB$ of order $r_A r_B$
 such that
 \begin{alignat}{1}
  \Repr_\AB(\rho \ot \sigma) &= \Ad_{U_\AB}(\Repr_A(\rho) \ot \Repr_B(\sigma)), \qquad
  \forall \rho \in \StN(A), ~\sigma \in \StN(B).
  \label{sm:eq:complex_monoidal_tensorator}
 \end{alignat}
\end{lemma}

The next lemma shows that the OPT satisfies the definition of quantum theory
(Definition~\ref{sm:def:quantum_theory}).%
\footnote{
Under CDP purification, Theorem~19 of Ref.~\cite{Chi-Dar-Per-2010} shows that
the normalized state spaces of all systems determine all allowed tests,
as used in Ref.~\cite{Chi-Dar-Per-2011}.
It provides an alternative route here once the state spaces of all systems,
including composites, are identified with those of quantum theory
in a way compatible with product states, as in Lemma~\ref{sm:lemma:complex_monoidal_tensorator}.
}

\begingroup
\renewcommand{\thethm}{E2}
\begin{lemma} \label{sm:lemma:E2}
 Every OPT satisfying local equivalence, the quantum state-space property,
 and refinement realization (Definition~\ref{sm:def:refinement_realization}) is quantum theory.
\end{lemma}
\endgroup
\begin{proof}

 Choose the representations $\{ \Repr_C \}_{C \in \System}$ and unitary matrices $U_\AB$ provided by
 Lemma~\ref{sm:lemma:complex_monoidal_tensorator}.
 Local tomography gives $d_\AB = d_A d_B$, and the quantum state-space property gives $d_C = r_C^2$;
 hence $r_\AB = r_A r_B$.
 We first establish the equality obtained from Eq.~\eqref{sm:eq:quantum_transformation}
 by replacing $n_A$ and $n_B$ with $r_A$ and $r_B$, respectively:

 \begin{alignat}{1}
  \left\{ \Repr_B \c f \c \Repr_A^{-1} \mid f \in \Trans(A,B) \right\}
  &= \operatorname{CPTNI}(r_A,r_B).
  \label{sm:eq:quantum_transformation_r}
 \end{alignat}

 We then verify the two conditions of Definition~\ref{sm:def:quantum_theory}.

 \noindent\textbf{Equation~\eqref{sm:eq:quantum_transformation_r}:}
 For each system $A$, set
 $\tPsi_A \coloneqq \left(\sum_{i=1}^{r_A} \ket{i} \ot \ket{i}\right)
 \left(\sum_{j=1}^{r_A} \bra{j} \ot \bra{j}\right) / r_A \in \Den_{r_{AA}}$,
 $\Psi_A \coloneqq \Repr_{AA}^{-1}(\Ad_{U_{AA}}(\tPsi_A)) \in \StN(AA)$, and
 $\beta_A \coloneqq (e_A \ot \id_A)(\Psi_A) \in \StN(A)$,
 where $\{ \ket{i} \}_{i=1}^{r_A}$ is the standard basis of $\Complex^{r_A}$.
 The state $\Psi_A$ is normalized and pure, and both of its marginals are the normalized internal state
 $\Repr_A^{-1}(I_{r_A}/r_A)$.
 For any system $B$ and real-linear map $F \colon \Her_{r_A} \to \Her_{r_B}$, set
 $J(F) \coloneqq (F \ot \id_{\Her_{r_A}})(\tPsi_A) \in \Her_{r_{BA}}$ and
 $\sigma_F \coloneqq \Repr_{BA}^{-1}(\Ad_{U_{BA}}(J(F))) \in \St_\Real(BA)$.
 Then the following equivalences hold:

 \begin{alignat}{1}
  F \in \operatorname{CPTNI}(r_A,r_B)
  &\quad\Leftrightarrow\quad \sigma_F \in \St(BA),~(e_B \ot \id_A)(\sigma_F) \le \beta_A \notag \\
  &\quad\Leftrightarrow\quad \exists f \in \Trans(A,B), ~\sigma_F = (f \ot \id_A)(\Psi_A) \notag \\
  &\quad\Leftrightarrow\quad \exists f \in \Trans(A,B), ~F = \Repr_B \c f \c \Repr_A^{-1}.
  \label{sm:eq:E2_CPTNI}
 \end{alignat}

 For the first equivalence, we used the Choi criterion together with
 $e_B = \tr_{r_B} \c \Repr_B$,
 $\Repr_A((e_B \ot \id_A)(\sigma_F)) = (\tr_{r_B} \ot \id_{\Her_{r_A}})(J(F))$, and
 $\Repr_A(\beta_A) = I_{r_A}/r_A$.
 In particular, the forward implication gives $\sigma_F \in \St_+(BA)$, while
 $(e_B \ot \id_A)(\sigma_F) \le \beta_A$ implies
 $e_{BA}(\sigma_F) \le e_A(\beta_A) = 1$.
 Thus Eq.~\eqref{sm:eq:OPT-state_cone_normalization} gives $\sigma_F \in \St(BA)$.
 In the second equivalence, the forward implication uses refinement realization.
 For the reverse implication, choose a test $\{ f,f_2,\ldots,f_k \}$ containing the given
 $f \in \Trans(A,B)$; then
 $\beta_A - (e_B \ot \id_A)((f \ot \id_A)(\Psi_A))
 = \sum_{i=2}^k (e_B \ot \id_A)((f_i \ot \id_A)(\Psi_A)) \in \St_+(A)$.
 For the third equivalence, Lemmas~\ref{sm:lemma:basic-product_span}
 and~\ref{sm:lemma:complex_monoidal_tensorator} give
 $\sigma_{\Repr_B \c f \c \Repr_A^{-1}} = (f \ot \id_A)(\Psi_A)$;
 we also use injectivity of the Choi correspondence $F \mapsto J(F)$.
 Consequently, Eq.~\eqref{sm:eq:quantum_transformation_r} holds.

 \noindent\textbf{Condition~\hyperref[sm:item:quantum_condition_1]{(1)}:}
 Fix any system $A$.
 Equation~\eqref{sm:eq:E2_CPTNI} implies that every measurement on $A$ is represented by a
 positive-operator-valued measure (POVM) on $\Complex^{r_A}$.
 Since such a POVM cannot perfectly distinguish more than $r_A$ normalized states, we have $n_A \le r_A$.
 By Lemma~\ref{sm:lemma:basic-perfect_discrimination}, there exists a system $E$ with
 a family of $r_A$ normalized states $\{ \eta_i \in \StN(E) \}_{i=1}^{r_A}$
 and a measurement $\{ a_i \in \Eff(E) \}_{i=1}^{r_A}$ that perfectly distinguishes them.
 Define

 \begin{alignat}{1}
  K &\colon \Her_{r_A} \ni X \mapsto \sum_{i=1}^{r_A} \braket{i|X|i} \Repr_E(\eta_i) \in \Her_{r_E}.
 \end{alignat}

 This is a trace-preserving CP map, so Eq.~\eqref{sm:eq:E2_CPTNI} yields a transformation
 $h \in \Trans(A,E)$ such that $\Repr_E \c h \c \Repr_A^{-1} = K$.
 Since $K$ is trace preserving, $e_E \c h = e_A$;
 thus $h$ is a channel by Lemma~\ref{sm:lemma:basic-deterministic}.
 Thus $\{ a_i \c h \}_{i=1}^{r_A}$ is a measurement on $A$ that perfectly distinguishes
 the normalized states $\{ \Repr_A^{-1}(\ket{j}\bra{j}) \}_{j=1}^{r_A}$, since
 $(a_i \c h)(\Repr_A^{-1}(\ket{j}\bra{j})) = \delta_{i,j}$.
 Therefore $n_A \ge r_A$, and hence $n_A = r_A$.
 Since this holds for every $A$, the equality $r_\AB = r_A r_B$ gives $n_\AB = n_A n_B$.

 \noindent\textbf{Condition~\hyperref[sm:item:quantum_condition_2]{(2)}:}
 It follows directly from Eq.~\eqref{sm:eq:quantum_transformation_r} and
 $n_A = r_A$, $n_B = r_B$.

\end{proof}

\section{Derivation of homogeneity} \label{sm:sec:homogeneity}

\subsection{Derivation common to CDP/ES purification} \label{sm:subsec:homogeneity-common}

In this subsection, we assume local equivalence and either CDP purification or ES purification,
and derive homogeneity of state cones.

For each system $A$, let $\cL_A$ be the real vector space of all
real-linear endomorphisms of $\St_\Real(A)$.
Define the cone of state-cone-preserving linear maps and the automorphism group by

\begin{alignat}{1}
 \cP_A &\coloneqq \{ P \in \cL_A \mid P(\St_+(A)) \subseteq \St_+(A) \}, \notag \\
 \cA_A &\coloneqq \Aut(\St_+(A)).
\end{alignat}

Under local equivalence, the identification given by
Lemma~\ref{sm:lemma:basic-transformation_linearization} yields
$\cL_A = \Trans_\Real(A,A)$.
We use the above definition because $\cL_A$ will also appear in the next subsection
without assuming local equivalence.
We have $\cA_A \subseteq \cP_A$ and $\Trans_+(A,A) \subseteq \cP_A$.
However, $\cA_A$ need not be contained in $\Trans_+(A,A)$.
Fix a norm on each finite-dimensional real vector space, and use the induced operator norm for linear maps.
Define the matrix exponential by $e^Z \coloneqq \sum_{k=0}^\infty Z^k/k!$.

\begin{lemma} \label{sm:lemma:homogeneity-positive_map_bound}
 For each system $A$, the following properties hold.
 \begin{enumerate}
   \item[(a)] $\cP_A$ is closed.
   \item[(b)] For every $x \in \Int \St_+(A)$, there exists $\kappa_x > 0$ such that
         $\|P\| \le \kappa_x e_A(P(x))$ $~(\forall P \in \cP_A)$.
 \end{enumerate}
\end{lemma}

\begin{proof}

 \noindent\textbf{(a):}
 If $P_n \in \cP_A$ converges to $P$, then for every $v \in \St_+(A)$,
 $P_n(v) \xrightarrow{n \to \infty} P(v)$.
 Since the state cone is closed, $P(v) \in \St_+(A)$ and hence $P \in \cP_A$.

 \noindent\textbf{(b):}
 Since $\StN(A)$ is compact, there exists $c > 0$ such that
 $\|w\| \le c e_A(w)$ $~(\forall w \in \St_+(A))$.
 Because $x$ is an interior point of the state cone, there exists $r > 0$ such that
 $x \pm r v \in \St_+(A)$ $~(\forall v \in \St_\Real(A),~\|v\| \le 1)$.
 For any $P \in \cP_A$, we have $P(x) \pm rP(v) \in \St_+(A)$.
 The triangle inequality and linearity of $e_A$ therefore give
 $2r\|P(v)\| \le \|P(x) + rP(v)\| + \|P(x) - rP(v)\| \le 2c e_A(P(x))$.
 Taking the supremum over $\|v\| \le 1$ proves (b) with $\kappa_x = c/r$.

\end{proof}

\begin{lemma} \label{sm:lemma:homogeneity-positive_tangent_exponentiation}
 For a system $A$ and $X \in \cL_A$, the following two conditions are equivalent.
 \begin{enumerate}
  \item $e^{tX} \in \cA_A$ for every $t \in \Real$.
  \item There exist a sequence of positive numbers $\{ s_n \}_{n \in \Natural}$ converging to $0$
        and, for each $\varepsilon \in \{ 1,-1 \}$, a sequence
        $\{ P_{n,\varepsilon} \in \cP_A \}_{n \in \Natural}$ such that
        $(P_{n,\varepsilon} - \id_A)/s_n \xrightarrow{n \to \infty} \varepsilon X$.
 \end{enumerate}
\end{lemma}

\begin{proof}

 \noindent\textbf{(1) $\Rightarrow$ (2):}
 Take the positive sequence $\{ s_n \coloneqq 1/n \}_{n \in \Natural}$ and,
 for each $\varepsilon \in \{ 1,-1 \}$, the sequence
 $\{ P_{n,\varepsilon} \coloneqq e^{\varepsilon s_nX} \in \cP_A \}_{n \in \Natural}$.
 The exponential series gives
 $(P_{n,\varepsilon} - \id_A)/s_n = \sum_{k=1}^\infty s_n^{k-1}(\varepsilon X)^k/k!
 \xrightarrow{n \to \infty} \varepsilon X$.

 \noindent\textbf{(2) $\Rightarrow$ (1):}
 For arbitrary $t \ge 0$, set $h_n \coloneqq t/s_n$.

 \begin{alignat}{1}
  e^{h_n(P_{n,\varepsilon} - \id_A)} &= e^{-h_n}\sum_{i=0}^\infty \frac{h_n^i}{i!}P_{n,\varepsilon}^i
  \xrightarrow{n \to \infty} e^{\varepsilon tX}.
  \label{sm:eq:homogeneity-positive_tangent_exponential_limit}
 \end{alignat}

 For every $n$, since $P_{n,\varepsilon} \in \cP_A$, every term of the series preserves the state cone.
 The sum preserves the cone, and its limit does so by Lemma~\ref{sm:lemma:homogeneity-positive_map_bound}(a).
 Thus $e^{tX},e^{-tX} \in \cP_A$.
 Since these maps are mutual inverses, $e^{tX} \in \cA_A$ $~(\forall t \in \Real)$.

\end{proof}

\begin{lemma} \label{sm:lemma:CDP_ES-infinite_pure_states}
 Every system $D$ with informational dimension at least $2$ has infinitely many normalized pure states.
\end{lemma}

\begin{proof}
 Suppose, for a contradiction, that $D$ has only finitely many normalized pure states.
 Then $\StN(D)$ is a polytope.
 Since the informational dimension of $D$ is at least $2$, there are two perfectly distinguishable
 normalized pure states $v_0,w_0$ and an effect $a \in \Eff(D)$ such that $a(v_0) = 0$ and $a(w_0) = 1$.
 The graph consisting of the vertices and edges of a polytope is connected
 (see, e.g., Ref.~\cite[Chap.~3]{Zie-1995}).
 Hence there is an edge path from $v_0$ to $w_0$.
 Let $u$ be the first vertex along this path where $a$ takes a positive value, and let $v$ be the preceding vertex.
 Then $a(u) > 0$ and $a(v) = 0$.
 Since the segment joining $u,v$ is an edge of $\StN(D)$,
 $\cone \{ u,v \}$ is a face of $\St_+(D)$.
 Choose any reversible channel $U$ satisfying $U(u) = v$.
 Such a channel exists by Lemma~\ref{sm:lemma:CDP_ES-pure_unitary}.
 Since $U$ permutes the finitely many normalized pure states of $D$
 by Lemma~\ref{sm:lemma:basic-unitary_pure},
 the sequence $\{ a(U^j(u)) \}_{j \in \Natural}$ is periodic.
 Because $a(u) > 0$ while $a(U(u)) = a(v) = 0$,
 a positive entry must also immediately follow a zero entry in one period.
 Consequently, some $j$ satisfies $a(U^j(u)) = 0$ and $a(U^{j + 1}(u)) > 0$.
 Put $b \coloneqq a \c U^j \in \Eff(D)$; then $b(u) = 0$ and $b(v) > 0$.

 Choose any purification $\Omega \in \StNP(DB)$ of $\rho \coloneqq (u + v)/2 \in \StN(D)$.
 For any effect $c$ on $B$, the inequality $c \le e_B$ implies
 $(\id_D \ot c)(\Omega) \le \rho$.
 Since $\cone \{ u,v \}$ is a face of $\St_+(D)$,
 $(\id_D \ot c)(\Omega) \in \cone \{ u,v \}$.
 It follows that $(\id_D \ot f)(\Omega) \in \Span \{ u,v \}$ for every $f \in \Eff_\Real(B)$.
 Since local equivalence implies local tomography, we obtain
 $\Omega \in \Span \{ u,v \} \ot \St_\Real(B)$.
 Thus $\Omega$ can be expressed as
 $\Omega = u \ot \eta + v \ot \theta$ $~(\eta,\theta \in \St_\Real(B))$.
 Moreover, $(a \ot \id_B)(\Omega) = a(u)\eta \in \St_+(B)$ and
 $(b \ot \id_B)(\Omega) = b(v)\theta \in \St_+(B)$,
 so $\eta,\theta \in \St_+(B)$.
 Comparing the $D$-marginal gives $e_B(\eta) = e_B(\theta) = 1/2$.
 Hence
 $\Omega = \left(u \ot (2\eta) + v \ot (2\theta)\right)/2$
 is a nontrivial convex decomposition into two distinct normalized states,
 contradicting the purity of $\Omega$.
 Therefore $D$ has infinitely many normalized pure states.

\end{proof}

For each system $A$, let $\mu_A$ denote the $(d_A-1)$-dimensional Lebesgue measure on the hyperplane
$\{ x \in \St_\Real(A) \mid e_A(x) = 1 \}$, normalized so that $\mu_A(\StN(A)) = 1$.
Define the centroid of $\StN(A)$ by
$\omega_A \coloneqq \int_{\StN(A)}x\,d\mu_A(x)$.
For every system $A$, the state $\omega_A$ is normalized and internal.%
\footnote{Proof: By Lemma~\ref{sm:lemma:basic-internal_exists}, $\StN(A)$ has nonempty relative interior
in this hyperplane, and that relative interior has $\mu_A$-measure $1$.
For every nonzero $b \in \St_+(A)^*$, we have $b(\rho) > 0$ for every normalized internal state $\rho$ of $A$.
Hence $b(\omega_A) = \int_{\StN(A)} b(\rho)\,d\mu_A(\rho) > 0$.
Since $\St_+(A)$ is a closed convex cone, this implies
$\omega_A \in \Int \St_+(A)$, so $\omega_A$ is internal by Lemma~\ref{sm:lemma:basic-internal}.}

\begin{lemma} \label{sm:lemma:homogeneity-nonconstant_two_sided_filter}
 There exist a system $C$, an element $Y \in \Trans_\Real(C,C)$,
 and a sequence of positive numbers $\{ r_n \}_{n \in \Natural}$ converging to $0$,
 together with, for each $\varepsilon \in \{ 1,-1 \}$,
 a sequence $\{ T_{n,\varepsilon} \in \Trans_+(C,C) \}_{n \in \Natural}$ such that
 $(T_{n,\varepsilon} - \id_C)/r_n \xrightarrow{n \to \infty} \varepsilon Y$.
 Moreover, $e_C \c Y$ is not constant on $\StN(C)$.
\end{lemma}

\begin{proof}
 By Assumption~\ref{sm:ass:OPT-nontriviality}, choose a system $D$ with informational dimension at least $2$.
 By Lemmas~\ref{sm:lemma:CDP_ES-pure_unitary} and~\ref{sm:lemma:CDP_ES-infinite_pure_states},
 there are infinitely many reversible channels on $D$.
 Compactness of $\StN(D)$, together with a basis for $\St_\Real(D)$ chosen from this set,
 gives a uniform bound on the norms of reversible channels and their inverses.
 That is, there exists $M > 0$ such that $\|U\| \le M$ and $\|U^{-1}\| \le M$
 for every reversible channel $U$.
 Choose a convergent sequence $\{ V_n \}_{n \in \Natural}$ of distinct reversible channels,
 and define $\{ U_n \coloneqq V_n^{-1} \c V_{n + 1} \}_{n \in \Natural}$.
 Since $U_n - \id_D = V_n^{-1} \c (V_{n + 1} - V_n)$,
 we have $U_n \ne \id_D$ and $U_n \xrightarrow{n \to \infty} \id_D$.
 Put $r_n \coloneqq \|U_n - \id_D\| > 0$; then $r_n \xrightarrow{n \to \infty} 0$.
 By compactness of the unit sphere, passage to a subsequence yields
 some $\Delta \in \Trans_\Real(D,D)$ such that

 \begin{alignat}{1}
  \frac{U_n - \id_D}{r_n} &\xrightarrow{n \to \infty} \Delta, \qquad
  \|\Delta\| = 1, \qquad e_D \c \Delta = \zero, \qquad
  \frac{U_n^{-1} - \id_D}{r_n} = -U_n^{-1} \c \frac{U_n - \id_D}{r_n}
  \xrightarrow{n \to \infty} -\Delta.
 \end{alignat}

 Here we used $e_D \c U_n = e_D$ and $U_n^{-1} \xrightarrow{n \to \infty} \id_D$.
 Since $\Eff(D)$ spans $\Eff_\Real(D)$ and $\Delta \neq \zero$,
 we can choose $a \in \Eff(D)$ for which $b \coloneqq a \c \Delta \neq \zero$.
 Choose $\rho_+,\rho_- \in \StN(D)$ satisfying
 $a(\rho_-) \le a(\rho) \le a(\rho_+)$ $~(\forall \rho \in \StN(D))$.
 If $a(\rho_+) = a(\rho_-)$, let $c$ denote their common value.
 Then $a(\rho) = c$ for every $\rho \in \StN(D)$.
 Since $\StN(D)$ spans $\St_\Real(D)$, this implies $a = c e_D$.
 But then $b = c e_D \c \Delta = \zero$, a contradiction.
 Thus $a(\rho_+) > a(\rho_-)$.
 Since $U_n^\varepsilon(\rho_\pm) \in \StN(D)$ for every $\varepsilon \in \{ 1,-1 \}$,

 \begin{alignat}{1}
  \varepsilon b(\rho_+)
  &= \lim_{n \to \infty} \frac{a(U_n^\varepsilon(\rho_+)) - a(\rho_+)}{r_n} \le 0, \qquad
  \varepsilon b(\rho_-)
  = \lim_{n \to \infty} \frac{a(U_n^\varepsilon(\rho_-)) - a(\rho_-)}{r_n} \ge 0.
 \end{alignat}

 Applying the left inequality with both $\varepsilon = 1$ and $\varepsilon = -1$ gives
 $b(\rho_+) \le 0$ and $-b(\rho_+) \le 0$, hence $b(\rho_+) = 0$.
 Similarly, the right inequality gives $b(\rho_-) = 0$.
 Suppose $b$ can be written as $b = c e_D + d a$ $~(c,d \in \Real)$.
 Subtracting the values on $\rho_+$ and $\rho_-$ gives
 $d(a(\rho_+) - a(\rho_-)) = 0$, hence $d = 0$.
 Moreover, $c = ce_D(\rho_+) = b(\rho_+) = 0$,
 so $b = \zero$, a contradiction.
 Therefore $b \notin \Span \{ e_D,a \}$.

 Choose $u,v \in \StN(D)$ and $f \in \Eff(D)$ such that $f(u) = 0$ and $f(v) = 1$.
 By the conditioned-test assumption, perform $\{ u \}$ after outcome $e_D - a$
 and $\{ v \}$ after outcome $a$ of the measurement $\{ e_D - a,a \}$.
 Coarse-graining the resulting test gives the channel
 $N \coloneqq u \c (e_D - a) + v \c a \in \Trans(D,D)$.
 The effect $g \coloneqq (f + e_D)/2 \in \Eff(D)$, obtained by mixing measurements,
 satisfies $g \c N = (e_D + a)/2$ and
 $g(\xi) \ge 1/2$ $~(\forall \xi \in \StN(D))$.
 By Lemma~\ref{sm:lemma:CDP_ES-reversible_dilation},
 there exist systems $E,C$, a normalized pure state $\alpha \in \StNP(E)$,
 and a reversible channel $V \in \Trans(DE,DC)$ such that $N$ can be written as

 \begin{alignat}{1}
  \includegraphics[scale=1.0,alt={N = (\id_D \ot e_C) \c V \c (\id_D \ot \alpha)}]{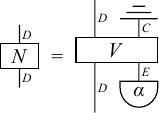}.
  \label{sm:eq:homogeneity-filter_pullback}
 \end{alignat}

 For each $\varepsilon \in \{ 1,-1 \}$, define a sequence of reversible channels
 $\{ R_{n,\varepsilon} \}_{n \in \Natural}$ and a map $G$ by

 \begin{alignat}{1}
  \includegraphics[scale=1.0,alt={
  R_{n,\varepsilon} \coloneqq V \c (U_n^\varepsilon \ot \id_E) \c V^{-1},
  G \coloneqq V \c (\Delta \ot \id_E) \c V^{-1}
  }]{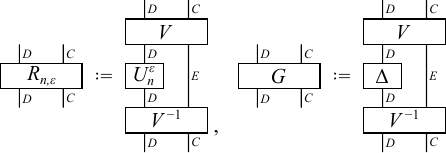}.
  \label{sm:eq:homogeneity-filter_G}
 \end{alignat}

 Then $(R_{n,\varepsilon} - \id_{DC})/r_n \xrightarrow{n \to \infty} \varepsilon G$.
 Let

 \begin{alignat}{1}
  \includegraphics[scale=1.0,alt={\gamma \coloneqq g \c (\id_D \ot e_C) \c G}]{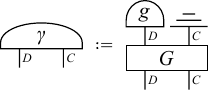}.
  \label{sm:eq:homogeneity-filter_gamma}
 \end{alignat}

 We show that there exists a normalized pure state $\xi \in \StNP(D)$
 for which $\gamma \c (\xi \ot \id_C)$ is not constant on $\StN(C)$.
 Suppose instead that it is constant for every $\xi \in \StNP(D)$.
 Fix $\tau \in \StN(C)$ and put
 $\gamma' \coloneqq \gamma \c (\id_D \ot \tau) \in \Eff_\Real(D)$.
 Then
 $\gamma(\xi \ot \sigma) = \gamma'(\xi)e_C(\sigma) = (\gamma' \ot e_C)(\xi \ot \sigma)$
 $~(\forall \xi \in \StNP(D),~\sigma \in \StN(C))$.
 Since $\StNP(D)$ spans $\St_\Real(D)$ and $\StN(C)$ spans $\St_\Real(C)$,
 local equivalence implies $\gamma = \gamma' \ot e_C$.
 Consequently,

 \begin{alignat}{1}
  \footnoteinset{-1.16}{0.3}{\eqref{sm:eq:homogeneity-filter_pullback}}{
  \footnoteinset{0.81}{0.3}{\eqref{sm:eq:homogeneity-filter_G}}{
  \footnoteinset{2.78}{0.3}{\eqref{sm:eq:homogeneity-filter_gamma}}{
  \footnoteinset{4.76}{0.3}{\eqref{sm:eq:homogeneity-filter_pullback}}{
  \includegraphics[scale=1.0,alt={b/2 = g \c N \c \Delta = \gamma' \c N}]{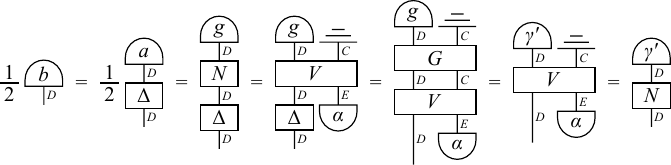}}}}}.
  \label{sm:eq:homogeneity-filter_b}
 \end{alignat}

 Here the first equality uses $b \coloneqq a \c \Delta$,
 and the second uses $g \c N = (e_D + a)/2$ and $e_D \c \Delta = \zero$.
 The penultimate equality uses both Eq.~\eqref{sm:eq:homogeneity-filter_gamma}
 and $\gamma = \gamma' \ot e_C$.
 On the other hand, the definition
 $N \coloneqq u \c (e_D - a) + v \c a$ gives
 $\gamma' \c N = \gamma'(u)(e_D - a) + \gamma'(v)a \in \Span \{ e_D,a \}$.
 Therefore $b \in \Span \{ e_D,a \}$, a contradiction.
 Thus a normalized pure state $\xi$ with the required nonconstancy exists.

 Put $p \coloneqq g(\xi) \ge 1/2 > 0$ and define
 $\Upsilon \colon \Trans_\Real(DC,DC) \to \Trans_\Real(C,C)$ by assigning to each
 $T \in \Trans_\Real(DC,DC)$ the map

 \begin{alignat}{1}
  \includegraphics[scale=1.0,alt={\Upsilon(T) \coloneqq p^{-1}(g \ot \id_C) \c T \c (\xi \ot \id_C)}]{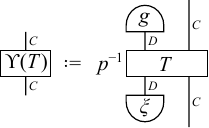}.
  \label{sm:eq:homogeneity-filter_Upsilon}
 \end{alignat}

 For each $\varepsilon \in \{ 1,-1 \}$, consider the sequence
 $\{ T_{n,\varepsilon} \coloneqq \Upsilon(R_{n,\varepsilon}) \in \Trans_+(C,C) \}_{n \in \Natural}$
 and $Y \coloneqq \Upsilon(G) \in \Trans_\Real(C,C)$.
 Since $g(\xi) = p$, we have $\Upsilon(\id_{DC}) = \id_C$.
 Thus $(R_{n,\varepsilon} - \id_{DC})/r_n \xrightarrow{n \to \infty} \varepsilon G$ implies
 $(T_{n,\varepsilon} - \id_C)/r_n
 = \Upsilon((R_{n,\varepsilon} - \id_{DC})/r_n)
 \xrightarrow{n \to \infty} \varepsilon Y$.
 Moreover, $e_C \c Y = p^{-1}\gamma \c (\xi \ot \id_C)$ is not constant on $\StN(C)$.

\end{proof}

\begin{lemma} \label{sm:lemma:homogeneity-effect_direction_transfer}
 Let $A,C$ be arbitrary systems and $Y \in \Trans_\Real(C,C)$.
 If $e_C \c Y$ is not constant on $\StN(C)$, then for every $a \in \Eff_\Real(A)$,
 there exist a real-linear map
 $\Gamma_a \colon \Trans_\Real(C,C) \to \Trans_\Real(A,A)$,
 a real number $c_e$, and a positive real number $c_a$ such that
 \begin{alignat}{1}
  \Gamma_a(\id_C) = \id_A, \qquad \Gamma_a(\Trans_+(C,C)) \subseteq \Trans_+(A,A), \qquad
  e_A \c \Gamma_a(Y) = c_e e_A + c_a a.
  \label{sm:eq:homogeneity-effect_direction_transfer}
 \end{alignat}
\end{lemma}

\begin{proof}

 Since $e_C \c Y$ is not constant on $\StN(C)$, choose
 $\zeta_\pm \in \StN(C)$ such that $q_+ > q_-$, where
 $q_\pm \coloneqq e_C(Y(\zeta_\pm))$.
 Put $c \coloneqq (q_++q_-)/2$ and $\lambda \coloneqq (q_+-q_-)/2 > 0$.
 Choose a measurement $\{ h_j \}$ whose effects span $\Eff_\Real(A)$.%
 \footnote{To construct such a measurement,
 choose a basis $\{ a_i \in \Eff(A) \}_{i=1}^k$ of $\Eff_\Real(A)$ and,
 for each $i$, a measurement $\{ a_{i,j} \}_j$ containing $a_i$
 $~(a_{i,1} \coloneqq a_i)$.
 Mixing these measurements with probability $1/k$ gives the measurement
 $\{ a_{i,j}/k \}_{i,j}$, whose effects span $\Eff_\Real(A)$.}
 Write $a = \sum_j a_jh_j$ $~(a_j \in \Real)$,
 and choose $\delta > 0$ such that $|\delta a_j| \le 1$ $~(\forall j)$.
 For each $j$, put $p_j \coloneqq (1 + \delta a_j)/2$.
 By the conditioned-test assumption, after outcome $j$ of $\{ h_j \}$,
 perform the probability test $\{ p_j,1 - p_j \}$ on the trivial system $I$.
 Coarse-graining over $j$ gives the measurement $\{ f_A,e_A - f_A \}$ on $A$,
 where $f_A \coloneqq \sum_j p_jh_j = (e_A + \delta a)/2$.
 By the conditioned-test assumption, prepare $\zeta_+$ after outcome $f_A$
 and $\zeta_-$ after outcome $e_A - f_A$.
 Coarse-graining the resulting test gives the channel
 $N_a \coloneqq \zeta_+ \c f_A + \zeta_- \c (e_A - f_A) \in \Trans(A,C)$,
 which satisfies $e_C \c Y \c N_a = c_e e_A + c_a a$,
 where $c_e \coloneqq c$ and $c_a \coloneqq \lambda\delta > 0$.

 Since $N_a$ is a channel, Lemma~\ref{sm:lemma:CDP_ES-reversible_dilation} gives
 $\alpha \in \StNP(B)$ and a reversible channel $W \in \Trans(AB,CF)$ such that

 \begin{alignat}{1}
  \includegraphics[scale=1.0,alt={
  N_a = (\id_C \ot e_F) \c W \c (\id_A \ot \alpha)
  }]{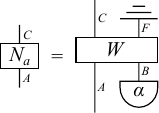}.
  \label{sm:eq:homogeneity-center-Na}
 \end{alignat}

 Define the real-linear map
 $\Gamma_a \colon \Trans_\Real(C,C) \to \Trans_\Real(A,A)$ by sending
 each $T \in \Trans_\Real(C,C)$ to

 \begin{alignat}{1}
  \includegraphics[scale=1.0,alt={
  \Gamma_a(T) = (\id_A \ot e_B) \c W^{-1} \c (T \ot \id_F) \c W \c (\id_A \ot \alpha)
  }]{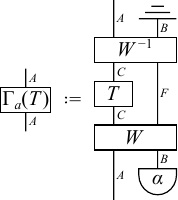}.
  \label{sm:eq:homogeneity-center-filter_pullback}
 \end{alignat}

 Since $W^{-1} \c W = \id_{AB}$ and $e_B(\alpha) = 1$,
 we have $\Gamma_a(\id_C) = \id_A$.
 Moreover, $e_B,W^{-1},W,\alpha$ are all transformations, so
 $\Gamma_a(\Trans_+(C,C)) \subseteq \Trans_+(A,A)$.
 Furthermore,

 \begin{alignat}{1}
  \footnoteinset{-1.91}{0.3}{\eqref{sm:eq:homogeneity-center-filter_pullback}}{%
  \footnoteinset{2.13}{0.3}{\eqref{sm:eq:homogeneity-center-Na}}{%
  \includegraphics[scale=1.0,alt={e_A \c \Gamma_a(Y) = e_C \c Y \c N_a}]{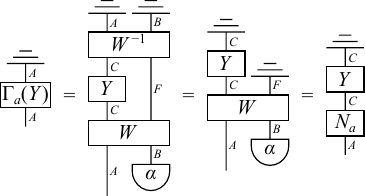}}}.
  \label{sm:eq:homogeneity-center-e_Gamma}
 \end{alignat}

 Hence $e_A \c \Gamma_a(Y) = c_e e_A + c_a a$, as claimed.

\end{proof}

\begin{lemma} \label{sm:lemma:CDP_ES-linear_tangent_lifting}
 For every system $A$ and $a \in \Eff_\Real(A)$, there exists
 $X_a \in \Trans_\Real(A,A)$ such that
 $e_A \c X_a = a$ and $e^{tX_a} \in \cA_A$ $~(\forall t \in \Real)$.
\end{lemma}

\begin{proof}
 Choose $C,Y,r_n,T_{n,\varepsilon}$ from
 Lemma~\ref{sm:lemma:homogeneity-nonconstant_two_sided_filter},
 and fix $a \in \Eff_\Real(A)$.
 Since $e_C \c Y$ is not constant on $\StN(C)$,
 Lemma~\ref{sm:lemma:homogeneity-effect_direction_transfer} gives
 a real-linear map $\Gamma_a \colon \Trans_\Real(C,C) \to \Trans_\Real(A,A)$,
 a real number $c_e$, and a positive real number $c_a$
 satisfying Eq.~\eqref{sm:eq:homogeneity-effect_direction_transfer}.
 Define $X_a \in \Trans_\Real(A,A)$, a positive sequence
 $\{ s_n \}_{n \in \Natural}$, and, for each $\varepsilon \in \{ 1,-1 \}$,
 a sequence $\{ P_{n,\varepsilon} \}_{n \in \Natural}$ by

 \begin{alignat}{1}
  P_{n,\varepsilon} &\coloneqq (1 - \varepsilon c_e r_n)\Gamma_a(T_{n,\varepsilon}), \qquad
  X_a \coloneqq \frac{\Gamma_a(Y) - c_e\id_A}{c_a}, \qquad s_n \coloneqq c_a r_n > 0.
 \end{alignat}

 Since $r_n \to 0$, the factor $1 - \varepsilon c_e r_n$ is positive for sufficiently large $n$.
 Discarding the finitely many terms that do not meet this condition and reindexing,
 we have $P_{n,\varepsilon} \in \Trans_+(A,A)$ and
 $s_n \xrightarrow{n \to \infty} 0$.
 By linearity of $\Gamma_a$ and
 $T_{n,\varepsilon} \xrightarrow{n \to \infty} \id_C$,

 \begin{alignat}{1}
  e_A \c X_a &= \frac{e_A \c \Gamma_a(Y) - c_e e_A}{c_a} = a, \qquad
  \frac{P_{n,\varepsilon} - \id_A}{s_n} = \frac{1}{c_a}
  \left[\Gamma_a\left(\frac{T_{n,\varepsilon} - \id_C}{r_n}\right)
  - \varepsilon c_e\Gamma_a(T_{n,\varepsilon})\right]
  \xrightarrow{n \to \infty} \varepsilon X_a.
  \label{sm:eq:homogeneity-lifted_two_sided_tangent}
 \end{alignat}

 Moreover, Lemma~\ref{sm:lemma:homogeneity-positive_tangent_exponentiation} gives
 $e^{tX_a} \in \cA_A$ $~(\forall t \in \Real)$.

\end{proof}

\begin{proposition} \label{sm:pro:CDP_ES-cone_homogeneity}
 In an OPT satisfying local equivalence and either CDP purification or ES purification,
 the state cone of every system is homogeneous.
\end{proposition}
\begin{proof}
 Fix a system $A$ and $z \in \Int \St_+(A)$, and consider
 \begin{alignat}{1}
  \argmax_{g \in \cD_z} |\det g|, \quad \cD_z \coloneqq \{ g \in \cA_A \mid e_A(g(z)) = 1 \}.
  \label{sm:eq:homogeneity-determinant_maximization}
 \end{alignat}
 Here $\det g$ denotes the determinant of the linear map $g$.
 The set $\cD_z$ is nonempty because $e_A(z)^{-1}\id_A \in \cD_z$.
 By Lemma~\ref{sm:lemma:homogeneity-positive_map_bound}(b),
 $\|g\| \le \kappa_z$ for every $g \in \cD_z$.
 Thus the supremum of $|\det g|$ is finite and positive.
 Let $g_z$ be the limit of a convergent subsequence of a maximizing sequence.
 By continuity, $\det g_z \ne 0$ and $e_A(g_z(z)) = 1$.
 The map $g_z$ is therefore invertible, and continuity of inversion shows
 that $g_z^{-1}$ is the limit of the corresponding sequence of inverses.
 By closedness of the state cone,
 $g_z,g_z^{-1} \in \cP_A$.
 Hence $g_z \in \cD_z$ attains the maximum.

 For any $a \in \Eff_\Real(A)$, Lemma~\ref{sm:lemma:CDP_ES-linear_tangent_lifting}
 provides $X_a \in \Trans_\Real(A,A)$ satisfying
 $e_A \c X_a = a$ and $e^{tX_a} \in \cA_A$ $~(\forall t \in \Real)$.
 This choice of $X_a$ does not depend on $z$.
 Set $\omega_z \coloneqq g_z(z)$ and
 $\lambda_z(t) \coloneqq e_A(e^{tX_a}(\omega_z)) > 0$.
 Since $\lambda_z(0) = 1$,
 $h_{z,t} \coloneqq \lambda_z(t)^{-1}e^{tX_a} \c g_z \in \cD_z$.
 The identity
 $|\det h_{z,t}| = \lambda_z(t)^{-d_A}\det(e^{tX_a})|\det g_z|$
 and the maximality inequality $|\det h_{z,t}| \le |\det g_z|$ imply
 $\lambda_z(t)^{-d_A}\det(e^{tX_a}) \le 1$.
 Substituting $\det(e^{tX_a}) = e^{t\tr X_a}$ and taking logarithms of both sides gives
 \begin{alignat}{1}
  t \tr X_a - d_A\log\lambda_z(t) \le 0, \quad \forall t \in \Real.
 \end{alignat}
 Since $e^{tX_a}$ is differentiable in $t$ and $\lambda_z(t) > 0$,
 the left-hand side is differentiable.
 It attains its maximum value $0$ at $t = 0$,
 so its derivative $\tr X_a - d_A\lambda_z'(0)$ vanishes there.
 As $\lambda_z'(0) = e_A(X_a(\omega_z)) = a(\omega_z)$,
 we obtain $a(\omega_z) = d_A^{-1}\tr X_a$.
 The right-hand side is independent of $z$.
 Thus $a(\omega_x) = a(\omega_y)$ for all
 $x,y \in \Int \St_+(A)$.
 Since $a \in \Eff_\Real(A)$ is arbitrary,
 $\omega_x = \omega_y$.
 Consequently, $g_y^{-1} \c g_x \in \cA_A$ maps $x$ to $y$,
 and the state cone is homogeneous.
\end{proof}

\begin{cor} \label{sm:cor:CDP_ES-symmetric_cones}
 Let an OPT satisfy local equivalence and either CDP purification or ES purification.
 Then every state cone is symmetric.
\end{cor}
\begin{proof}

 By Proposition~\ref{sm:pro:CDP_ES-cone_homogeneity}, every state cone is homogeneous.
 Lemmas~\ref{sm:lemma:CDP_ES-pure_unitary} and~\ref{sm:lemma:E1}
 then imply that every state cone is symmetric.

\end{proof}

\subsection{Further characterizations of homogeneity} \label{sm:subsec:homogeneity-criterion}

The following lemma supplements the proof of homogeneity given above.

\begin{lemma} \label{sm:lemma:homogeneity_criterion}
 For any OPT and any system $A$, let $\cD_z$ be the set defined in
 Eq.~\eqref{sm:eq:homogeneity-determinant_maximization} for each $z \in \Int\St_+(A)$.
 Then the following four conditions are equivalent.
 \begin{enumerate}
  \item The state cone $\St_+(A)$ is homogeneous.
  \item For every $z \in \Int\St_+(A)$, the set $G_z \coloneqq \argmax_{g \in \cD_z}|\det g|$ is nonempty,
        and the image $g_z(z)$ is independent of $z$ and of the choice of $g_z \in G_z$.
  \item For every $a \in \Eff_\Real(A)$, there exists $X_a \in \cL_A$ satisfying
        $e_A \c X_a = a$ and $e^{tX_a} \in \cA_A$ $~(\forall t \in \Real)$.
  \item For every $a \in \Eff_\Real(A)$, there exist $X_a \in \cL_A$,
        a sequence of positive numbers $\{ s_n \}_{n \in \Natural}$ converging to $0$,
        and, for each $\varepsilon \in \{ 1,-1 \}$,
        a sequence $\{ P_{n,\varepsilon} \in \cP_A \}_{n \in \Natural}$ such that
        $e_A \c X_a = a$ and
        $(P_{n,\varepsilon} - \id_A)/s_n \xrightarrow{n \to \infty} \varepsilon X_a$
        $~(\forall \varepsilon \in \{ 1,-1 \})$.
 \end{enumerate}
 Moreover, if condition~(2) holds, every maximizer $g_z$ satisfies $g_z(z) = \omega_A$.
\end{lemma}

Condition~(3) is the conclusion of Lemma~\ref{sm:lemma:CDP_ES-linear_tangent_lifting}.
Proposition~\ref{sm:pro:CDP_ES-cone_homogeneity} establishes homogeneity by using
(3) $\Rightarrow$ (1).

\begin{proof}

 We prove (1) $\Rightarrow$ (3), (3) $\Rightarrow$ (2), (2) $\Rightarrow$ (1),
 and (3) $\Leftrightarrow$ (4), and then show $g_z(z) = \omega_A$.

 \noindent\textbf{(1) $\Rightarrow$ (3):}
 The group $\cA_A$ is a closed subgroup of $\GL(\St_\Real(A))$.
 Indeed, if $g_n \in \cA_A$ converges to an invertible $g$,
 then $g_n^{-1}$ converges to $g^{-1}$, and both limits preserve the closed cone $\St_+(A)$.
 Hence $\cA_A$ is a Lie group by the closed subgroup theorem~\cite[Theorem~20.12]{Lee-2012}.
 Let its Lie algebra be
 $\mathfrak g_A \coloneqq \{ X \in \cL_A \mid e^{tX} \in \cA_A\ (\forall t \in \Real) \}$,
 and consider the map $H \colon \mathfrak g_A \ni X \mapsto e_A \c X \in \Eff_\Real(A)$.
 It suffices to show that $H$ is surjective.
 Consider the map $T \colon \mathfrak g_A \ni X \mapsto X(\omega_A) \in \St_\Real(A)$.
 If $H(X) = \zero$, then $e_A \c X^k = \zero$ $~(\forall k \in \Natural)$, so
 $e_A \c e^{tX} = e_A \c \sum_{k=0}^\infty (tX)^k/k! = e_A$ for every $t \in \Real$.
 Thus $e^{tX}$ is an affine automorphism of $\StN(A)$ and fixes its centroid $\omega_A$;
 that is, $e^{tX}(\omega_A) = \omega_A$.
 Differentiating at $t = 0$ gives $X(\omega_A) = \zero$, hence $T(X) = \zero$,
 so $\ker H \subseteq \ker T$.
 By condition~(1), the $\cA_A$-orbit of $\omega_A \in \Int\St_+(A)$ equals $\Int\St_+(A)$.
 Since the differential at the identity of the orbit map $g \mapsto g(\omega_A)$ is $T$,
 the homogeneous-space theorems~\cite[Theorems~21.17 and~21.18]{Lee-2012}
 imply $\rank T = d_A$.
 As $\dim\Eff_\Real(A) = d_A$, the inclusion of kernels implies
 $\rank H \ge \rank T = d_A$; thus $H$ is surjective and condition~(3) follows.

 \noindent\textbf{(3) $\Rightarrow$ (2):}
 In the proof of Proposition~\ref{sm:pro:CDP_ES-cone_homogeneity},
 replace the application of Lemma~\ref{sm:lemma:CDP_ES-linear_tangent_lifting}
 with condition~(3).
 The same argument then proves condition~(2).

 \noindent\textbf{(2) $\Rightarrow$ (1):}
 For any $x,y \in \Int\St_+(A)$, choose maximizers $g_x,g_y$.
 By condition~(2), $g_x(x) = g_y(y)$, so $g_y^{-1} \c g_x \in \cA_A$ maps $x$ to $y$.
 Thus $\St_+(A)$ is homogeneous.

 \noindent\textbf{(3) $\Leftrightarrow$ (4):}
 This follows immediately from Lemma~\ref{sm:lemma:homogeneity-positive_tangent_exponentiation}.

 \noindent\textbf{$g_z(z) = \omega_A$:}
 Assume condition~(2), and put $K \coloneqq \{v \in \St_+(A) \mid e_A(v) \le 1\}$.
 Choose any $g \in \cD_{\omega_A}$ and put $f \coloneqq e_A \c g$.
 Then $f(\rho) > 0$ for any $\rho \in \StN(A)$.
 Let $\operatorname{vol}(W)$ denote the Lebesgue volume of a set $W$ in $\St_\Real(A)$,
 and consider the volumes of $K$ and $g^{-1}(K)$.
 For fixed $t$, the map $\trho \mapsto t\trho$ scales
 $(d_A - 1)$-dimensional volume along the section by $t^{d_A - 1}$.
 Hence, with a proportionality constant $c$ independent of $\trho$,
 the volume element on $\St_\Real(A)$ is
 $ct^{d_A - 1}dt\,d\mu_A(\trho)$.
 Consequently,

 \begin{alignat}{1}
  \operatorname{vol}(K) &= c \int_{\StN(A)}\int_0^1 t^{d_A - 1}\,dt\,d\mu_A(\trho) = \frac{c}{d_A}, \notag \\
  \operatorname{vol}(g^{-1}(K)) &= c \int_{\StN(A)}\int_0^{f(\trho)^{-1}} t^{d_A - 1}\,dt\,d\mu_A(\trho)
  = \frac{c}{d_A} \int_{\StN(A)} f(\trho)^{-d_A}\,d\mu_A(\trho)
 \end{alignat}

 and therefore

 \begin{alignat}{1}
  |\det g|^{-1} &= \frac{\operatorname{vol}(g^{-1}(K))}{\operatorname{vol}(K)}
  = \int_{\StN(A)} f(\trho)^{-d_A}\,d\mu_A(\trho)
  \ge \left(\int_{\StN(A)} f(\trho)\,d\mu_A(\trho)\right)^{-d_A}
  = f(\omega_A)^{-d_A} = 1.
 \end{alignat}

 Here Jensen's inequality applies because $x \mapsto x^{-d_A}$ is convex for $x > 0$,
 and the penultimate equality follows from the definition of $\omega_A$
 and the linearity of $f$.
 Hence $|\det g| \le 1$, and $\id_A \in \cD_{\omega_A}$ is a maximizer.
 By condition~(2), for any $g_z \in G_z$ and $g_{\omega_A} \in G_{\omega_A}$,
 we have $g_z(z) = g_{\omega_A}(\omega_A)$;
 taking $g_{\omega_A} = \id_A$ gives $g_z(z) = \id_A(\omega_A) = \omega_A$.

\end{proof}

\subsection{Alternative derivation of homogeneity under ES purification} \label{sm:subsec:homogeneity-ES}

In this subsection, we assume local equivalence and ES purification and present
an alternative derivation of homogeneity that is simpler than the one in
Sec.~\ref{sm:subsec:homogeneity-common}.

\begin{lemma} \label{sm:lemma:ES-zigzag}
 Let $\Psi_\rho \in \StNP(A\tA)$ be a purification of a normalized internal state
 $\rho \in \StN(A)$ of system $A$.
 Then there exist $E_\rho \in \Eff(\tA A)$ and a positive real number $p_\rho$ such that
 \begin{alignat}{1}
  \includegraphics[scale=1.0,alt={
  (\id_A \ot E_\rho) \c (\Psi_\rho \ot \id_A) = p_\rho \id_A,
  (E_\rho \ot \id_\tA) \c (\id_\tA \ot \Psi_\rho) = p_\rho \id_\tA
  }]{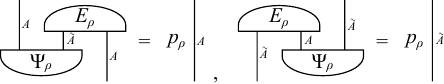}.
  \label{sm:eq:ES-zigzag}
 \end{alignat}
 The left and right equalities in the diagram mean, respectively,
 $(\id_A \ot E_\rho) \c (\Psi_\rho \ot \id_A) = p_\rho\id_A$ and
 $(E_\rho \ot \id_\tA) \c (\id_\tA \ot \Psi_\rho) = p_\rho\id_\tA$.
\end{lemma}
This is essentially a special case of Corollary~19 of Ref.~\cite{Chi-Dar-Per-2010}.
\begin{proof}
 By Lemma~\ref{sm:lemma:ES-complementary_internal},
 the complementary state $\trho$ of $\rho$ with respect to $\Psi_\rho$ is internal.
 By Lemma~\ref{sm:lemma:basic-pure_internal_product}, the product $\rho \ot \trho$ is internal.
 Lemma~\ref{sm:lemma:basic-internal_refinement} therefore gives $p_\rho > 0$ such that
 $p_\rho\Psi_\rho \le \rho \ot \trho$, i.e.,
 $p_\rho\Psi_\rho \in D_{\rho \ot \trho}$.
 Since $\Psi_\rho \ot \Psi_\rho$ is the parallel composition of normalized pure states,
 it is pure by Lemma~\ref{sm:lemma:basic-pure_internal_product}.
 It purifies $\rho \ot \trho$, since
 $(\id_A \ot e_{\tA A} \ot \id_\tA)(\Psi_\rho \ot \Psi_\rho) = \rho \ot \trho$.
 Applying Lemma~\ref{sm:lemma:CDP_ES-steering} to
 $p_\rho\Psi_\rho \in D_{\rho \ot \trho}$ gives
 an effect $E_\rho \in \Eff(\tA A)$ such that
 \begin{alignat}{1}
  \includegraphics[scale=1.0,alt={
  (\id_A \ot E_\rho \ot \id_\tA)(\Psi_\rho \ot \Psi_\rho) = p_\rho \Psi_\rho
  }]{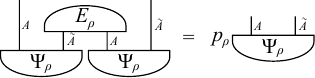}.
  \label{sm:eq:ES-zigzag_steering}
 \end{alignat}
 This equation can be rewritten as
 \begin{alignat}{1}
  \includegraphics[scale=1.0,alt={
  (((\id_A \ot E_\rho) \c (\Psi_\rho \ot \id_A)) \ot \id_\tA)(\Psi_\rho)
  = (p_\rho \id_A \ot \id_\tA)(\Psi_\rho),
  (\id_A \ot ((E_\rho \ot \id_\tA) \c (\id_\tA \ot \Psi_\rho)))(\Psi_\rho)
  = (\id_A \ot p_\rho \id_\tA)(\Psi_\rho)
  }]{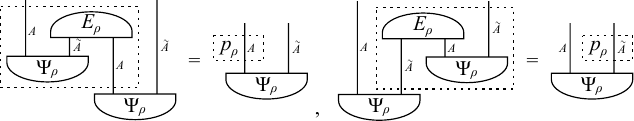}.
  \label{sm:eq:ES-zigzag_steering_rewrite}
 \end{alignat}
 Since $\Psi_\rho$ purifies the internal state $\rho$,
 the left equality of Eq.~\eqref{sm:eq:ES-zigzag_steering_rewrite}
 and Lemma~\ref{sm:lemma:CDP_ES-purification_transformation_equality}
 imply the left equality of Eq.~\eqref{sm:eq:ES-zigzag}.
 Similarly, $\Psi_\rho$ purifies the internal state $\trho$.
 Interchanging the roles of the two systems in
 Lemma~\ref{sm:lemma:CDP_ES-purification_transformation_equality}
 therefore gives the right equality of Eq.~\eqref{sm:eq:ES-zigzag}
 from the right equality of Eq.~\eqref{sm:eq:ES-zigzag_steering_rewrite}.
\end{proof}

\begin{lemma} \label{sm:lemma:ES-homogeneity_isomorphism}
 For normalized internal states $\tau,\rho$ of a system $A$,
 there exists $\Gamma_{\tau,\rho} \in \cA_A$ satisfying
 $\Gamma_{\tau,\rho}(\tau) = \rho$ and
 $\Gamma_{\tau,\rho},\Gamma_{\tau,\rho}^{-1} \in \Trans_+(A,A)$.
\end{lemma}
\begin{proof}
 For each $\sigma \in \{ \tau,\rho \}$, choose a purification
 $\Psi_\sigma \in \StNP(A\tA)$ of $\sigma$.
 Let $E_\sigma,p_\sigma$ be as in Lemma~\ref{sm:lemma:ES-zigzag}.
 Define
 \begin{alignat}{1}
  \includegraphics[scale=1.0,alt={
   \Gamma_{\tau,\rho} \coloneqq p_\tau^{-1} (\id_A \ot E_\tau) \c (\Psi_\rho \ot \id_A),
   \Gamma_{\rho,\tau} \coloneqq p_\rho^{-1} (\id_A \ot E_\rho) \c (\Psi_\tau \ot \id_A)
  }]{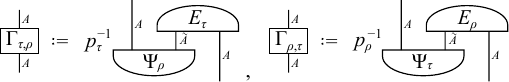}.
  \label{sm:eq:ES-homogeneity_Gamma}
 \end{alignat}
 Then Lemma~\ref{sm:lemma:ES-zigzag} yields
 \begin{alignat}{1}
  \footnoteinset{3.74}{0.3}{\eqref{sm:eq:ES-zigzag}}{%
  \footnoteinset{7.02}{0.3}{\eqref{sm:eq:ES-zigzag}}{%
  \includegraphics[scale=1.0,alt={
  \Gamma_{\rho,\tau} \c \Gamma_{\tau,\rho} = \id_A
  }]{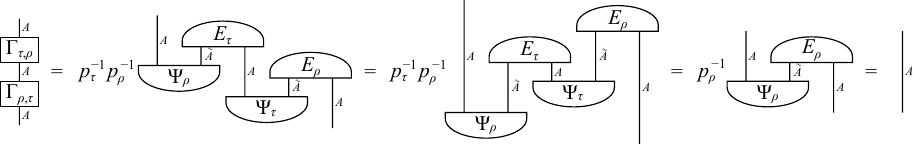}}}
  \label{sm:eq:ES-homogeneity_Gamma_id1}
 \end{alignat}
 and
 \begin{alignat}{1}
  \footnoteinset{3.74}{0.3}{\eqref{sm:eq:ES-zigzag}}{%
  \footnoteinset{7.02}{0.3}{\eqref{sm:eq:ES-zigzag}}{%
  \includegraphics[scale=1.0,alt={
  \Gamma_{\tau,\rho} \c \Gamma_{\rho,\tau} = \id_A
  }]{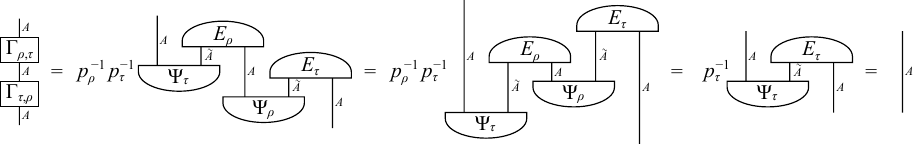}}}.
  \label{sm:eq:ES-homogeneity_Gamma_id2}
 \end{alignat}
 Therefore $\Gamma_{\rho,\tau}$ is the inverse of $\Gamma_{\tau,\rho}$.
 Furthermore,
 \begin{alignat}{1}
  \footnoteinset{4.60}{0.3}{\eqref{sm:eq:ES-zigzag}}{%
  \includegraphics[scale=1.0,alt={\Gamma_{\tau,\rho}(\tau) = \rho}]{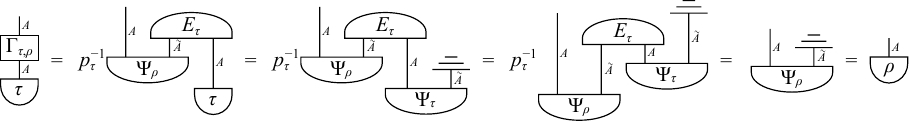}}.
  \label{sm:eq:ES-homogeneity_Gamma_tau}
 \end{alignat}
 Here the second equality uses the fact that $\Psi_\tau$ purifies $\tau$,
 and the last uses that $\Psi_\rho$ purifies $\rho$.
 Equation~\eqref{sm:eq:ES-homogeneity_Gamma} shows that
 $\Gamma_{\tau,\rho} \in \Trans_+(A,A)$, and likewise for its inverse.
 Hence $\Gamma_{\tau,\rho} \in \cA_A$.
\end{proof}

\begin{proposition} \label{sm:pro:ES-homogeneity}
 In an OPT satisfying local equivalence and ES purification,
 every state cone is homogeneous.
\end{proposition}
\begin{proof}
 By Lemma~\ref{sm:lemma:basic-internal_exists}, $\St_+(A)$ has an interior point.
 Given any two interior points $x,y \in \St_+(A)$,
 it suffices to find $g \in \cA_A$ satisfying $g(x) = y$.
 Let $\ol{x},\ol{y} \in \StN(A)$ be the normalizations of $x,y$, respectively,
 and put $\lambda_x \coloneqq e_A(x) > 0$ and $\lambda_y \coloneqq e_A(y) > 0$.
 Then $x = \lambda_x\ol{x}$ and $y = \lambda_y\ol{y}$.
 Since $\ol{x},\ol{y}$ are also interior points of $\St_+(A)$,
 they are internal states by Lemma~\ref{sm:lemma:basic-internal}.
 By Lemma~\ref{sm:lemma:ES-homogeneity_isomorphism},
 there exists $\Gamma_{\ol{x},\ol{y}} \in \cA_A$
 mapping $\ol{x}$ to $\ol{y}$.
 Define $\Upsilon_{\ol{x},\ol{y}} \coloneqq \lambda_x^{-1}\lambda_y\Gamma_{\ol{x},\ol{y}}$.
 By Lemma~\ref{sm:lemma:ES-homogeneity_isomorphism} and $\lambda_x,\lambda_y > 0$,
 $\Upsilon_{\ol{x},\ol{y}} \in \cA_A$ and
 $\Upsilon_{\ol{x},\ol{y}}(x)
 = \lambda_x\Upsilon_{\ol{x},\ol{y}}(\ol{x})
 = \lambda_y\ol{y} = y$.
\end{proof}

\section{Derivation of the quantum state-space property} \label{sm:sec:simple_complex}

In this appendix, we first present the necessary background on symmetric cones and Euclidean Jordan algebras in
Sec.~\ref{sm:subsec:simple_complex-EJA} and then derive the quantum state-space property
(Definition~\ref{sm:def:quantum_state_space_property}) in Sec.~\ref{sm:subsec:simple_complex-derivation}.

\subsection{Symmetric cones and Euclidean Jordan algebras} \label{sm:subsec:simple_complex-EJA}

A state cone is \termdef{simple} if it cannot be written as a nontrivial direct sum of two nonzero convex cones.
It is well known that the real vector space spanned by a symmetric cone can be regarded as a Euclidean Jordan algebra
(EJA).
This subsection recalls how the real vector space spanned by a symmetric cone is naturally associated with an EJA,
and records basic facts about symmetric cones and EJAs.
For proofs, see, for example, Ref.~\cite{Far-1994}.

A \termdef{Euclidean Jordan algebra} (EJA) is a finite-dimensional real inner-product space $\EJA$ equipped
with a bilinear product $\b$ and a unit element $u$ such that, for all $x$ and $y$ in $\EJA$,
$x \b y = y \b x$, $x \b (x^2 \b y) = x^2 \b (x \b y)$, where $x^2 \coloneqq x \b x$,
$u \b x = x$, and
\begin{alignat}{1}
 \braket{x \b y, z} &= \braket{y, x \b z}, \quad \forall x,y,z \in \EJA.
 \label{sm:eq:EJA-associative_inner_product}
\end{alignat}
The \termdef{positive cone} of $\EJA$ is $\EJA^+ \coloneqq \{ x^2 \mid x \in \EJA \}$.
This cone is symmetric.
Conversely, as recalled below, for every symmetric cone $\cC$, there exists an EJA $\EJA$ with $\cC = \EJA^+$.
 If $x$ and $y$ are in $\EJA$ and satisfy $y - x \in \EJA^+$, we write $x \le y$ or $y \ge x$.

An element $p \in \EJA$ with $p \b p = p$ is called an \termdef{idempotent}.
A nonzero idempotent $p$ is \termdef{primitive} if there do not exist nonzero idempotents $p_1$ and $p_2$ with
$p_1 + p_2 = p$ and $p_1 \b p_2 = \zero$.
For every primitive idempotent $p$, $\Realp p$ is an extreme ray of $\EJA^+$, and every extreme ray
of $\EJA^+$ is of this form.
Elements $x$ and $y$ in $\EJA$ are called \termdef{orthogonal} if $x \b y = \zero$.
A family of primitive idempotents $\{ p_i \}_{i=1}^r$ satisfying $\sum_{i=1}^r p_i = u$ is called a
\termdef{Jordan frame} of $\EJA$.
The positive integer $r$ is uniquely determined and is called the \termdef{rank} of $\EJA$.
Elements of a Jordan frame are mutually orthogonal.
Every primitive idempotent belongs to some Jordan frame.
For all $x$ and $y$ in $\EJA^+$,
\begin{alignat}{1}
 \braket{x,y} = 0 &\quad \Leftrightarrow \quad x \b y = \zero.
 \label{sm:eq:EJA-positive_orthogonality}
\end{alignat}

Every $x \in \EJA$ admits a \termdef{spectral decomposition} $x = \sum_{i=1}^r \lambda_i p_i$, where
$\lambda_i \in \Real$ and $\{ p_i \}_{i=1}^r$ is a Jordan frame.
The numbers $\lambda_1,\dots,\lambda_r$ are the \termdef{eigenvalues} of $x$.
In this decomposition, $x \in \EJA^+$ if and only if all $\lambda_i \ge 0$.
In particular, all eigenvalues of an idempotent are either $0$ or $1$.
An element $x$ is in the interior of $\EJA^+$ if and only if all its eigenvalues are positive; hence the only
idempotent in the interior is the unit $u$.
Since $u$ is an interior point of $\EJA^+$, for every $x \in \EJA^+$ there exists $t > 0$ such that $x \le t u$.
The map $x \mapsto \sum_{i=1}^r \lambda_i$ is called the \termdef{trace} and is denoted by $\tr$.
Throughout the remainder of this appendix, we equip every EJA with its standard trace inner product,
namely $\braket{x,y} \coloneqq \tr(x \b y)$.
Then $\braket{u,x} = \tr x$ for the unit $u$, and $\EJA^+$ is self-dual with respect to this inner product.
\begin{thm}[Koecher--Vinberg theorem~\cite{Koe-1957,Vin-1960}] \label{sm:thm:EJA-Koecher_Vinberg}
 For any symmetric cone $\cC$ and any interior point $u$ of $\cC$, there exists an EJA $\EJA$ whose underlying real
 vector space is the span of $\cC$, whose unit is $u$, and whose positive cone satisfies $\EJA^+ = \cC$.
\end{thm}

Consider a system $A$ whose state cone $\St_+(A)$ is symmetric.
Since the only $x \in \St_+(A)$ with $e_A(x) = 0$ is $\zero$, $e_A$ is an interior point of the dual cone $\St_+(A)^*$.
For each such system $A$, we fix an EJA structure $\EJA_A$ whose positive cone is $\St_+(A)$ and whose trace
$\tr_A$ agrees with the deterministic effect: $\tr_A x = e_A(x)$ for all $x \in \EJA_A$.
Such a choice is always possible: for a symmetric cone $\cC$ and an interior point $e$ of its dual cone,
one can put an EJA structure on $\Span \cC$ whose positive cone is $\cC$ and whose trace functional is $e$;
see Chap.~III of Ref.~\cite{Far-1994}.
Let $u_A$ be the unit of $\EJA_A$, and let the inner product of $\EJA_A$ be the trace inner product.
Then $e_A(x) = \tr_A x = \braket{u_A,x}$ for every $x \in \EJA_A$.
As a real vector space, $\EJA_A = \St_\Real(A)$, and $\St_+(A) = \EJA_A^+$.
We use this EJA $\EJA_A$ throughout the following discussion.
We write $r_A$ for the rank of $\EJA_A$.
A state $\rho \in \St(A)$ is normalized if and only if $\tr_A \rho = 1$.
In particular, $\rho$ is a normalized pure state if and only if it is a
primitive idempotent, since
primitive idempotents are precisely the extreme points of the convex set
$\StN(A) = \{ x \in \EJA_A^+ \mid \tr_A x = 1 \}$.

A nonzero linear subspace $\EJA'$ of an EJA $\EJA$ is a \termdef{subalgebra} if $x \b y \in \EJA'$ for all
$x$ and $y$ in $\EJA'$.
For every nonzero idempotent $p$ of $\EJA$, the set $\{ x \in \EJA \mid p \b x = x \}$ is a subalgebra whose unit
is $p$.
In particular, if $p$ is primitive, this subalgebra equals $\Real p$.
Two subalgebras $\EJA_1$ and $\EJA_2$ are \termdef{orthogonal} if $x \b y = \zero$
for every $x \in \EJA_1$ and $y \in \EJA_2$.

A \termdef{direct-sum decomposition} $\EJA = \EJA_1 \oplus \EJA_2$ means that
$\EJA = \EJA_1 \dotplus \EJA_2$ as real vector spaces and that the two subalgebras are orthogonal.
Here $\oplus$ denotes direct sums as EJAs, whereas $\dotplus$ denotes direct sums as real vector spaces.
If $\EJA = \EJA_1 \oplus \EJA_2$, then every $x \in \EJA$ can be written uniquely as
$x = x_1 + x_2$ $~(x_1 \in \EJA_1, ~x_2 \in \EJA_2)$.
The linear map $P_1 \colon \EJA \to \EJA$ defined by $P_1(x) \coloneqq x_1$ is called the
\termdef{projection} onto $\EJA_1$.
The projection onto $\EJA_2$ is defined similarly.
For $x \in \EJA^+$, the components satisfy $x_i \in \EJA_i^+$.%
\footnote{Proof: If $x = y^2$ and $y = y_1 + y_2$ with $y_i \in \EJA_i$, then
$x = y^2 = y_1^2 + y_2^2$ by orthogonality, so $x_i = y_i^2 \in \EJA_i^+$.}
Thus the projection $P_1$ satisfies $\zero \le P_1(x) \le x$.

Two EJAs $\EJA_1$ and $\EJA_2$ are called \termdef{EJA-isomorphic}, written
$\EJA_1 \congEJA \EJA_2$, if there
exists a linear isomorphism $\Phi$ from $\EJA_1$ to $\EJA_2$ preserving both the Jordan product and the inner
product.
Such a map is called an \termdef{EJA isomorphism}.
Then $\Phi(\EJA_1^+) = \EJA_2^+$ and $\tr_2 \c \Phi = \tr_1$, where $\tr_1$ and $\tr_2$ are the respective traces.

A linear subspace $\W$ of an EJA $\EJA$ is an \termdef{ideal} if $x \b y \in \W$ for every $x \in \W$ and $y \in \EJA$.
A nonzero EJA is \termdef{simple} if its only ideals are $\{ \zero \}$ and itself.
If an EJA is not simple, it decomposes as $\EJA = \EJA_1 \oplus \EJA_2$ for nonzero EJAs $\EJA_1$ and $\EJA_2$.
Its positive cone $\EJA^+$ accordingly decomposes as the direct sum of the two positive cones $\EJA_1^+$ and
$\EJA_2^+$.
The positive cone $\EJA^+$ is simple if and only if the EJA $\EJA$ is simple.

Let $\Quaternion$ and $\Octonion$ denote the quaternions and octonions.
For $\Field \in \{ \Real, \Complex, \Quaternion, \Octonion \}$, let $\Her_n(\Field)$ ($n$ a positive integer) be the set
of Hermitian matrices of order $n$ over $\Field$.
Simple EJAs are classified into five types: real $\Her_r(\Real)$, complex $\Her_r = \Her_r(\Complex)$, quaternionic
$\Her_r(\Quaternion)$, spin factors $\Spin_d$ with $d \ge 5$ and $d \neq 6$, and the exceptional algebra
$\Her_3(\Octonion)$ ($r$ a positive integer)~\cite{Jor-Neu-Wig-1934}.
Their ranks and dimensions are summarized in Table~\ref{sm:tab:EJA-simple_classification}.
Except for the exceptional EJA isomorphisms
$\Her_1(\Real) \congEJA \Her_1 \congEJA \Her_1(\Quaternion) ~(\congEJA \Real)$,
the simple EJAs in this classification are mutually non-EJA-isomorphic.%
\footnote{The algebras $\Spin_3$, $\Spin_4$, $\Spin_6$, $\Her_1(\Octonion)$, and
$\Her_2(\Octonion)$ are also simple,
but they are EJA-isomorphic to $\Her_2(\Real)$, $\Her_2$, $\Her_2(\Quaternion)$, $\Real$, and $\Spin_{10}$,
respectively, so excluding them entails no essential loss.}
\begin{table}[h]
 \centering
 \caption{Ranks and dimensions of simple EJAs}
 \label{sm:tab:EJA-simple_classification}
 \begin{tabular}{ccc}
  \hline
  Type & Rank & Dimension \\ \hline
  $\Her_r(\Real)$ & $r$ & $r(r + 1) / 2$ \\
  $\Her_r$ & $r$ & $r^2$ \\
  $\Her_r(\Quaternion)$ & $r$ & $r(2r - 1)$ \\
  $\Spin_d$ & $2$ & $d$ \\
  $\Her_3(\Octonion)$ & $3$ & $27$ \\ \hline
 \end{tabular}
\end{table}

\subsection{Derivation of the quantum state-space property} \label{sm:subsec:simple_complex-derivation}

Throughout this subsection, the EJA representations of all systems are chosen as in the previous subsection, so
that $e_A(x) = \tr_A x = \braket{u_A,x}$ $~(x \in \EJA_A)$.

\begin{lemma} \label{sm:lemma:EJA-effect_separation_orthogonality}
 Suppose that the state cone $\St_+(A)$ is symmetric.
 Let $f \in \St_+(A)^*$ satisfy $e_A - f \in \St_+(A)^*$ and be represented as
 $f(z) = \braket{h,z}$ $~(z \in \EJA_A)$.
 Then $\zero \le h \le u_A$.
 If $x,y \in \St_+(A)$ satisfy $f(x) = e_A(x)$ and $f(y) = 0$, then
 $h \b x = x$, $h \b y = \zero$, and $x \b y = \zero$.
\end{lemma}
\begin{proof}
 The inclusions $f, e_A - f \in \St_+(A)^*$, together with self-duality of $\St_+(A)$,
 imply $\zero \le h \le u_A$.
 Since $e_A(z) = \braket{u_A,z}$, we have
 $\braket{u_A - h,x} = 0$ and $\braket{h,y} = 0$.
 As $u_A - h,h,x,y \in \EJA_A^+$, Eq.~\eqref{sm:eq:EJA-positive_orthogonality} gives
 $(u_A - h) \b x = \zero$ and $h \b y = \zero$.
 From $u_A \b x = x$, it follows that $h \b x = x$.
 Equation~\eqref{sm:eq:EJA-associative_inner_product} therefore yields
 $\braket{x,y} = \braket{h \b x,y} = \braket{x,h \b y} = 0$.
 Applying Eq.~\eqref{sm:eq:EJA-positive_orthogonality} again gives $x \b y = \zero$.
\end{proof}

\begin{lemma} \label{sm:lemma:CDP_ES-simplicity}
 Let an OPT satisfy local equivalence and have symmetric state cones for all systems.
 If every normalized state admits a purification, then every state cone is simple.
\end{lemma}
\begin{proof}
 Suppose, for a contradiction, that there exists a system $A$ for which $\St_+(A)$ is not simple.
 By assumption, $\EJA_A$ admits a nontrivial direct-sum decomposition $\EJA_A = \EJA_1 \oplus \EJA_2$.
 Let $\pi_1$ be the projection onto the first summand and define $f \coloneqq e_A \c \pi_1 \in \St_+(A)^*$.
 Choose a normalized internal state $\chi$ of $A$ and a purification $\Psi \in \StNP(AB)$ of $\chi$.
 Since both direct-sum components are nonzero, $0 < f(\chi) < 1$.
 Define $h \coloneqq f \ot e_B \in \St_\Real(AB)^*$.
 Since $f, e_A - f \in \St_+(A)^*$ and $(\id_A \ot e_B)(\St_+(AB)) \subseteq \St_+(A)$,
 we have $h, e_\AB - h \in \St_+(AB)^*$.
 By self-duality of $\St_+(AB)$, there exists $\zero \le q \le u_\AB$ such that $h(z) = \braket{q,z}$.

 Take $x_i \in \EJA_i^+$ and $y_i \in \EJA_B^+$.
 Since $h(x_1 \ot y_1) = e_\AB(x_1 \ot y_1)$ and $h(x_2 \ot y_2) = 0$,
 applying Lemma~\ref{sm:lemma:EJA-effect_separation_orthogonality} to system $AB$ gives
 $q \b (x_1 \ot y_1) = x_1 \ot y_1$, $q \b (x_2 \ot y_2) = \zero$, and
 $(x_1 \ot y_1) \b (x_2 \ot y_2) = \zero$.
 By Lemma~\ref{sm:lemma:basic-product_span} and $\EJA_i = \Span \EJA_i^+$,
 we can write $\Psi = z_1 + z_2$, where $z_i$ is a linear combination
 of elements $x \ot y$ with $x \in \EJA_i^+$ and $y \in \EJA_B^+$.
 Bilinearity therefore gives $z_1 \b z_2 = \zero$, $q \b z_1 = z_1$, and $q \b z_2 = \zero$.
 Since $\Psi$ is a primitive idempotent, $\Psi = \Psi^2 = z_1^2 + z_2^2$, with $z_i^2 \in \EJA_\AB^+$.
 Put $t \coloneqq h(\Psi) = f(\chi)$, so $0 < t < 1$.
 By Eq.~\eqref{sm:eq:EJA-associative_inner_product},
 $h(z_1^2) = \braket{q \b z_1,z_1} = \braket{z_1,z_1} = e_\AB(z_1^2)$ and $h(z_2^2) = 0$.
 Hence $e_\AB(z_1^2) = h(z_1^2) + h(z_2^2) = h(\Psi) = t$.
 By purity of $\Psi$ and $\zero \le z_1^2 \le \Psi$, we have $z_1^2 = t\Psi$.
 It follows that $t = h(z_1^2) = h(t\Psi) = t^2$, contradicting $0 < t < 1$.
 Thus every state cone is simple.
\end{proof}

\begin{lemma} \label{sm:lemma:jordan_frame_product}
 Let an OPT satisfy local equivalence and have symmetric state cones for all systems.
 For any systems $A,B$, if $\{ p_i \}_{i=1}^{r_A}$ and $\{ q_j \}_{j=1}^{r_B}$ are Jordan frames of $\EJA_A$ and
 $\EJA_B$, respectively, then $\{ p_i \ot q_j \}_{(i,j)=(1,1)}^{(r_A,r_B)}$ is a Jordan frame of $\EJA_\AB$.
\end{lemma}
\begin{proof}
 Set $x_{ij} \coloneqq p_i \ot q_j$.
 By Lemma~\ref{sm:lemma:basic-pure_internal_product}, $x_{ij}$ is a normalized pure state and hence a primitive
 idempotent of $\EJA_\AB$.
 Choose distinct index pairs $(i,j)$ and $(k,l)$.
 If $i \neq k$, put $C \coloneqq A$, $M \coloneqq \id_A \ot e_B$,
 $v \coloneqq p_i$, and $w \coloneqq p_k$.
 If $i = k$, in which case $j \neq l$, put $C \coloneqq B$, $M \coloneqq e_A \ot \id_B$,
 $v \coloneqq q_j$, and $w \coloneqq q_l$.
 In either case, $M$ is a positive linear map satisfying $e_C \c M = e_\AB$,
 $M(x_{ij}) = v$, and $M(x_{kl}) = w$.
 Moreover, $v$ and $w$ are distinct primitive idempotents in the same Jordan frame, so
 $\zero \le v \le u_C$, $v^2 = v$, $\tr_C v = 1$, and $v \b w = \zero$.
 Define $f(z) \coloneqq \braket{v,M(z)}$.
 For every $z \in \St_+(AB)$, we have
 $0 \le f(z) \le \braket{u_C,M(z)} = e_\AB(z)$.
 Furthermore,
 $f(x_{ij}) = \braket{v,v} = \tr_C(v^2) = 1$ and
 $f(x_{kl}) = \braket{v,w} = \tr_C(v \b w) = 0$.
 Lemma~\ref{sm:lemma:EJA-effect_separation_orthogonality} therefore gives
 $x_{ij} \b x_{kl} = \zero$, so the elements $\{x_{ij}\}_{i,j}$ are pairwise orthogonal.
 Set $s \coloneqq \sum_{i,j} x_{ij} = u_A \ot u_B$.
 Since the $x_{ij}$ are pairwise orthogonal idempotents,
 $s^2 = \sum_{i,j} x_{ij}^2 = s$, and hence $s$ is an idempotent.
 On the other hand, $u_A$ and $u_B$ lie in the interiors of their respective positive cones, and
 $\tr_A u_A = r_A$ and $\tr_B u_B = r_B$.
 Thus $u_A/r_A$ and $u_B/r_B$ are normalized internal states.
 By Lemma~\ref{sm:lemma:basic-pure_internal_product},
 $(u_A/r_A) \ot (u_B/r_B)$ lies in the interior of $\St_+(AB)$.
 Therefore $s = r_A r_B((u_A/r_A) \ot (u_B/r_B))$ also lies in the interior.
 An idempotent in the interior of an EJA positive cone must be its unit, so $s = u_\AB$.
 Consequently, $\{x_{ij}\}_{i,j}$ is a Jordan frame of $\EJA_\AB$.
\end{proof}

\begin{cor} \label{sm:cor:jordan_rank_multiplicativity}
 If an OPT satisfies local equivalence and all its state cones are symmetric, then $r_\AB = r_A r_B$
 for all systems $A,B$.
\end{cor}
\begin{proof}
 Immediate from Lemma~\ref{sm:lemma:jordan_frame_product}.
\end{proof}

\begin{lemma} \label{sm:lemma:complex_state_space_reduction}
 If an OPT satisfies local equivalence and all its state cones are simple and symmetric,
 then it has the quantum state-space property.
\end{lemma}

\begin{proof}
 Choose any system $A$.
 If $r_A = 1$, then the simple EJA is one-dimensional and $\StN(A) \cong \Den_1$.
 Assume henceforth that $r_A \ge 2$, and put $r \coloneqq r_A^2$ and $d \coloneqq d_A^2$.
 By Corollary~\ref{sm:cor:jordan_rank_multiplicativity} and local equivalence, the rank of $\EJA_{AA}$ is $r \ge 4$
 and its dimension is $d$.
 Table~\ref{sm:tab:EJA-simple_classification} gives $r(r + 1)/2 \le d \le r(2r - 1)$.
 These inequalities exclude every type other than $\Her_{r_A}$ for $\EJA_A$ as follows.
 \begin{itemize}
  \item If $\EJA_A$ is of real type, then $d = r_A^2(r_A + 1)^2/4 < r(r + 1)/2$, a contradiction.
        Here the strict inequality uses $2(r_A^2 + 1) - (r_A + 1)^2 = (r_A - 1)^2 > 0$.
  \item If $\EJA_A$ is of quaternionic or exceptional type, then $d_A \ge r_A(2r_A - 1)$.
        For the exceptional type, $r_A = 3$ and $d_A = 27 \ge 15$.
        Hence $d \ge r_A^2(2r_A - 1)^2 > r(2r - 1)$, a contradiction.
        The second strict inequality uses $(2r_A - 1)^2 - (2r_A^2 - 1) = 2(r_A - 1)^2 > 0$.
  \item If $\EJA_A$ is a spin factor, then $r = 4$ and $d = s^2$ $~(s \ge 5)$.
        Among the dimensions $10,16,28$ of simple EJAs of rank $4$, only $16$ is a perfect square,
        contradicting $s \ge 5$.
 \end{itemize}
 Thus $\EJA_A \congEJA \Her_{r_A}$.
 Since an EJA isomorphism preserves the positive cone and trace, $\StN(A) \cong \Den_{r_A}$ follows.
\end{proof}

\begin{proposition} \label{sm:pro:CDP_ES-density_matrix_state_space}
 Let an OPT satisfy local equivalence and have symmetric state cones for all systems.
 If every normalized state admits a purification, then the OPT has the quantum state-space property.
\end{proposition}
\begin{proof}
 By Lemma~\ref{sm:lemma:CDP_ES-simplicity}, the state cone of every system is simple.
 Thus Lemma~\ref{sm:lemma:complex_state_space_reduction} yields the quantum state-space property.
\end{proof}

\begin{thm} \label{sm:thm:CDP_ES-quantum_theory}
 Every OPT satisfying local equivalence and either CDP purification or ES purification is quantum theory.
\end{thm}
\begin{proof}
 CDP/ES purification guarantees the existence of a purification for every normalized state, and
 Corollary~\ref{sm:cor:CDP_ES-symmetric_cones} shows that all state cones are symmetric.
 Thus Proposition~\ref{sm:pro:CDP_ES-density_matrix_state_space},
 Lemma~\ref{sm:lemma:CDP_ES-refinement_realization}, and Lemma~\ref{sm:lemma:E2} establish quantum theory.
\end{proof}

\subsection{Additional result: Derivation of measurement no-restriction} \label{sm:subsec:simple_complex-no_restriction}

The \termdef{measurement no-restriction} property states that every finite family
$\{ a_i \}_{i=1}^n$ of elements of $\St_+(A)^*$ satisfying $\sum_{i=1}^n a_i = e_A$ is a measurement.
We show here that, under local equivalence, self-duality of the symmetric cones, together with purification,
implies measurement no-restriction.
This result is not used in our derivation of quantum theory.
Hereafter, we denote the inner product on $\EJA_C$ by $\braket{~,~}_C$ for each system $C$.

\begin{lemma} \label{sm:lemma:CDP_ES-symmetric_inner_product}
 In an OPT satisfying local equivalence, if the state cone of every system is symmetric, then
 \begin{alignat}{1}
  \braket{a \ot b,a' \ot b'}_\AB &= \braket{a,a'}_A \braket{b,b'}_B, \quad
  \forall a,a' \in \EJA_A, ~b,b' \in \EJA_B.
 \end{alignat}
\end{lemma}
\begin{proof}
 Choose arbitrary primitive idempotents $p \in \EJA_A^+$ and $q \in \EJA_B^+$, and define
 $f(a,b) \coloneqq \braket{a \ot b,p \ot q}_\AB$.
 It suffices to show that $f(a,b) = \braket{a,p}_A\braket{b,q}_B$ for every $a \in \EJA_A^+$ and $b \in \EJA_B^+$.
 Indeed, because $\EJA_A^+$ and $\EJA_B^+$ span $\EJA_A$ and $\EJA_B$, respectively, bilinearity extends this
 identity to arbitrary $a \in \EJA_A$ and $b \in \EJA_B$.
 For arbitrary $a' \in \EJA_A$ and $b' \in \EJA_B$, take spectral decompositions $a' = \sum_i \lambda_i p_i$ and
 $b' = \sum_j \mu_j q_j$.
 We then obtain
 $\braket{a \ot b,a' \ot b'}_\AB = \sum_{i,j} \lambda_i\mu_j\braket{a \ot b,p_i \ot q_j}_\AB
 = \braket{a,a'}_A\braket{b,b'}_B$.

 Take arbitrary $a \in \EJA_A^+$ and $b \in \EJA_B^+$.
 Since $a \ot b,p \ot q \in \EJA_\AB^+$, we have $f(a,b) \ge 0$.
 Moreover, $f(u_A,u_B) = 1$ and $f(u_A - p,u_B) = f(u_A,u_B - q) = 0$ (choose Jordan frames containing $p$ and $q$
 and apply Lemma~\ref{sm:lemma:jordan_frame_product}).
 By self-duality of $\EJA_A^+$, there exists $h_b \in \EJA_A^+$ with $f(a,b) = \braket{h_b,a}_A$.
 Since $b \le t u_B$ for some $t > 0$, we have $0 \le \braket{h_b,u_A - p}_A = f(u_A - p,b) \le f(u_A - p,t u_B) = 0$.
 Equation~\eqref{sm:eq:EJA-positive_orthogonality} therefore gives $h_b \b (u_A - p) = \zero$, or $h_b = p \b h_b$.
 Thus $h_b$ belongs to the subalgebra $E_p \coloneqq \{ x \in \EJA_A \mid p \b x = x \}$.
 Its unit is the primitive idempotent $p$, so $\{ p \}$ is a Jordan frame of $E_p$, and consequently $h_b = sp$
 $~(s \in \Real)$.
 Taking the trace of both sides and using $\tr_A p = 1$ gives $s = \tr_A h_b = \braket{h_b,u_A}_A = f(u_A,b)$.
 Therefore $h_b = f(u_A,b)p$, and hence $f(a,b) = \braket{h_b,a}_A = f(u_A,b)\braket{p,a}_A$.
 Interchanging the roles of systems $A$ and $B$ gives $f(a,b) = f(a,u_B)\braket{q,b}_B$.
 Since $f(a,u_B) = f(u_A,u_B)\braket{p,a}_A = \braket{p,a}_A$, it follows that $f(a,b) = \braket{p,a}_A\braket{q,b}_B$.
\end{proof}

\begin{proposition} \label{sm:pro:CDP_ES-no_restriction}
 Every OPT satisfying local equivalence and either CDP purification or ES purification satisfies
 measurement no-restriction.
\end{proposition}
\begin{proof}

 Choose any finite family $\{ a_i \in \St_+(A)^* \}_{i=1}^m$ satisfying $\sum_{i=1}^m a_i = e_A$.
 We must show that $\{ a_i \}_{i=1}^m$ is a measurement.
 Choose any purification $\Psi \in \StNP(AB)$ of a normalized internal state $\rho \in \StN(A)$, and consider the
 complementary state $\trho \coloneqq (e_A \ot \id_B)(\Psi)$.
 Set $\sigma_i \coloneqq (a_i \ot \id_B)(\Psi)$.
 By self-duality of $\EJA_A^+$, for each $i$ there exists $h_i \in \EJA_A^+$ such that
 $a_i(x) = \braket{h_i,x}_A$.
 For arbitrary $y \in \EJA_B^+$, consider the functional
 $P_y \colon \EJA_B \ni x \mapsto \braket{y,x}_B \in \Real$.
 This functional can be identified with an element $P_y \in \Eff_\Real(B)$.
 Lemma~\ref{sm:lemma:CDP_ES-symmetric_inner_product} and local equivalence then give
 \begin{alignat}{1}
  \braket{y,\sigma_i}_B &= (a_i \ot P_y)(\Psi)
  = \braket{h_i \ot y,\Psi}_\AB \ge 0,
 \end{alignat}
 where the inequality uses $h_i \ot y,\Psi \in \St_+(AB)$.
 Hence $\sigma_i \in \St_+(B)$, and since $\sum_{i=1}^m \sigma_i = \trho$, the collection $\{ \sigma_i \}_{i=1}^m$
 is a refinement of $\trho$.
 Apply Lemma~\ref{sm:lemma:CDP_ES-finite_steering} with the roles of the systems interchanged.
 This gives a measurement $\{ b_i \}_{i=1}^m$ on $A$ satisfying $(b_i \ot \id_B)(\Psi) = \sigma_i$.
 Since $(a_i \ot \id_B)(\Psi) = \sigma_i$ as well and $\rho$ is a normalized internal state,
 Lemma~\ref{sm:lemma:CDP_ES-purification_transformation_equality} implies $a_i = b_i$.
 Thus $\{ a_i \}_{i=1}^m$ is a measurement.

\end{proof}

\section{FPS purification} \label{sm:sec:FPS}

In this appendix, we consider FPS purification, a postulate weaker than ES purification.
We give a counterexample showing that an OPT satisfying local equivalence and FPS purification need not be
quantum theory,
and discuss its differences from CDP/ES purification.

\subsection{The FPS purification postulate} \label{sm:subsec:FPS-postulate}

Consider an OPT satisfying local equivalence, and impose the following postulate, which is similar to
CDP and ES purification
(Postulates~\ref{sm:postulate:CDP_purification} and~\ref{sm:postulate:ES_purification}, respectively).

\begin{postulate}[Fixed-purifying-system (FPS) purification] \label{sm:postulate:FPS_purification}
 For every system $A$, there exists a system $A'$ such that every normalized state $\rho \in \StN(A)$ admits a
 purification on $AA'$.
 Moreover, if $\Psi,\Psi' \in \StNP(AA')$ are both purifications of $\rho$, then there exists a reversible
 channel $U$ on $A'$ satisfying
 \begin{alignat}{1}
  \includegraphics[scale=1.0,alt={\Psi' = (\id_A \ot U)(\Psi)}]{figures/purify_uniq_CDP.pdf}.
  \label{sm:eq:FPS-purification_uniqueness}
 \end{alignat}
\end{postulate}

If one additionally requires that, for each system $A$, the corresponding system $A'$ be equivalent to $A$,
FPS purification becomes ES purification.
In this sense, FPS purification relaxes the conditions of ES purification.
By contrast, neither FPS purification nor CDP purification is obtained from the other simply by relaxing its conditions.
Indeed, unlike FPS purification, CDP purification does not require, for each system $A$, the existence of a
single system $A'$ serving as a purifying system for every $\rho \in \StN(A)$.
On the other hand, FPS purification requires the existence of reversible channels only for
purifications on $AA'$.

From one perspective, the FPS purification postulate might appear more natural than CDP/ES purification.
For example, one might regard the requirement $\tA \simeq A$ in ES purification or the condition on the
existence of reversible channels in CDP purification as overly restrictive.
These requirements are relaxed in FPS purification.
As we shall show, however, local equivalence together with FPS purification does not suffice to derive quantum
theory.
Furthermore, even if FPS purification is supplemented by the conditions $n_{A'} = n_A$ and
$d_{A'} = d_A$, quantum theory still cannot be derived.
Indeed, the following counterexample satisfies these conditions.
Table~\ref{sm:tab:purification} compares the three purification postulates in terms of the existence
requirements for a purifying system, the systems to which the uniqueness requirement applies, and whether
quantum theory follows when combined with local equivalence.
\begin{table}[h]
 \caption{Comparison of purification postulates under local equivalence, including whether each
 implies quantum theory.
 $\checkmark$ indicates that quantum theory can be derived; $\times$ indicates that a counterexample exists.}
 \label{sm:tab:purification}
 \small
 \begin{ruledtabular}
  \begin{tabular}{lccc}
   Postulate & Purifying system & Uniqueness applies to & Quantum theory \\ \hline
   CDP purification (Postulate~\ref{sm:postulate:CDP_purification}) & Chosen per state & Any purifying system & $\checkmark$ \\
   ES purification (Postulate~\ref{sm:postulate:ES_purification}) & Fixed equivalent system $\tA$ & $\tA$ & $\checkmark$ \\
   FPS purification (Postulate~\ref{sm:postulate:FPS_purification}) & Fixed system $A'$ & $A'$ & $\times$ \\
  \end{tabular}
 \end{ruledtabular}
\end{table}

\subsection{A nonquantum OPT satisfying local equivalence and FPS purification}
\label{sm:subsec:FPS-counterexample}

\subsubsection{Construction of the counterexample}

We construct the counterexample explicitly.
For each positive integer $l$, introduce a \termdef{primitive system} $C_l$, with $C_1 \coloneqq I$.
A nontrivial system is a finite sequence of primitive systems
$A = C_{l_1} \cdots C_{l_k}$ $~(k,l_j \in \Natural, ~l_j \ge 2)$, and each $C_{l_j}$ is called
a \termdef{primitive factor} of $A$.
Parallel composition is defined by concatenating the sequences of primitive factors, with $C_1 = I$ as the
unit system.
For a nontrivial system $A = C_{l_1} \cdots C_{l_k}$, put $m_A \coloneqq \prod_{j=1}^k l_j$, and set
$m_I \coloneqq 1$.
For each system $A$, define the set of normalized states by
\begin{alignat}{1}
 \StN(A) &\coloneqq \Den_{m_A}.
\end{alignat}
As explained below, for each system $A$ we can choose the primitive system $C_{m_A}$ as its purifying system.

For a nontrivial system $A = C_{l_1} \cdots C_{l_k}$, define the cone of separable positive semidefinite
matrices with respect to the primitive factors by
\begin{alignat}{1}
 \cS_A &\coloneqq
 \cone \left\{ F_1 \ot \cdots \ot F_k \mid F_j \in \Her_{l_j}^+~(\forall j \in \{ 1,\dots,k \}) \right \}.
 \label{sm:eq:FPS-separable_cone}
\end{alignat}
We also set $\cS_I \coloneqq \Her_1^+ = \Realp$.
For systems $A$ and $B$, define the tests from $A$ to $B$ to be all finite families $\{ T_i \}_{i=1}^r$ of CP
maps from $\Complex^{m_A \times m_A}$ to $\Complex^{m_B \times m_B}$ satisfying
\begin{alignat}{1}
 T_i^*(\cS_B) &\subseteq \cS_A \quad (\forall i \in \{ 1,\dots,r \}), \qquad
 \sum_{i=1}^r T_i^*(I_{m_B}) = I_{m_A}.
 \label{sm:eq:FPS-allowed_tests}
\end{alignat}
Here $T_i^*$ is the adjoint with respect to the trace pairing, determined by
$\tr(F T_i(X)) = \tr(T_i^*(F)X)$ for all
$X \in \Complex^{m_A \times m_A}$ and $F \in \Complex^{m_B \times m_B}$.
For each $p \in \Realp$, identify the CP map
$\Complex^{1 \times 1} \ni z \mapsto p z \in \Complex^{1 \times 1}$ with $p$.
Likewise, for each positive semidefinite matrix $\rho \in \Her_{m_A}^+$, identify the CP map
$\Complex^{1 \times 1} \ni z \mapsto z \rho \in \Complex^{m_A \times m_A}$ with $\rho$.
(The same identification was made in Example~\ref{sm:ex:quant_basic}.)
Sequential and parallel compositions are defined by composition and tensor product of maps, respectively.
In this case, $C_{ln}$ and $C_lC_n$ are inequivalent for all $l,n \ge 2$.%
\footnote{Proof: Let $A \coloneqq C_lC_n$ and $B \coloneqq C_{ln}$.
Then $\cS_B = \Her_{ln}^+$, whereas $\cS_A$ is a proper subcone of $\Her_{ln}^+$.
Suppose that there exists a reversible channel $U$ from $A$ to $B$.
Since $U^*$ and $(U^{-1})^*$ are mutually inverse CP maps,
we have $U^*(\Her_{ln}^+) = \Her_{ln}^+$.
However, Eq.~\eqref{sm:eq:FPS-allowed_tests} implies $U^*(\cS_B) \subseteq \cS_A$, and hence
$\Her_{ln}^+ = U^*(\cS_B) \subseteq \cS_A \subsetneq \Her_{ln}^+$, a contradiction.}

This construction is not quantum theory.
Indeed, take $A \coloneqq C_2C_2$ and set
\begin{alignat}{1}
 \ket\phi &\coloneqq \frac{\ket{00} + \ket{11}}{\sqrt2},
 \qquad \Pi \coloneqq \ket\phi\bra\phi \in \Den_4.
 \label{sm:eq:FPS-Bell_state}
\end{alignat}
The operator $\Pi$ is not separable with respect to the two primitive factors, so $\Pi \notin \cS_A$.
Thus $a_\Pi \colon X \mapsto \tr(\Pi X)$ is not an effect, although $a_\Pi,e_A-a_\Pi \in \St_+(A)^*$.
Since $\Repr_A^{-1}$ preserves positivity and trace for every representation $\Repr_A$,
the functional $a_\Pi \c \Repr_A^{-1}$ is positive and trace-nonincreasing, hence CP.
Condition~\hyperref[sm:item:quantum_condition_2]{(2)} of Definition~\ref{sm:def:quantum_theory}
with $B = I$ would imply $a_\Pi \in \Eff(A)$, a contradiction.

\subsubsection{Verification that the construction is an OPT satisfying local equivalence and FPS purification}

We next verify that the construction defines an OPT satisfying local equivalence and FPS purification.

In what follows, for each $F \in \Her_n$, we often denote the map
$\Her_n \ni X \mapsto \tr(FX) \in \Real$ simply by $F$.
In particular, every element of $\Eff_\Real(B)$ can be represented by an element of $\Her_{m_B}$.
As a preliminary step, we show that $\Eff(B)$ spans $\Her_{m_B}$ for every system $B$.
\begin{alignat}{1}
 \Eff(B) &= \{ F \in \cS_B \mid I_{m_B} - F \in \cS_B \}.
 \label{sm:eq:FPS-effect_set}
\end{alignat}
This implies $\Eff_+(B) = \cS_B$.%
\footnote{Proof: For every rank-one projector $P \in \cS_B$,
$I_{m_B}-P \in \cS_B$, as is seen by extending its local factors to orthonormal bases.
For any nonzero $F \in \cS_B$, write
$F = \sum_i \lambda_i P_i$ with $\lambda_i > 0$ and rank-one projectors $P_i \in \cS_B$.
Setting $t \coloneqq (\sum_i \lambda_i)^{-1}$ gives
$I_{m_B}-tF = t\sum_i\lambda_i(I_{m_B}-P_i) \in \cS_B$.
The case $F = \zero$ is immediate.
Hence $\Eff_+(B) = \cS_B$ by Eq.~\eqref{sm:eq:FPS-effect_set}.}
Since $\Her_{m_B} = \Span \cS_B = \Span \Eff_+(B) = \Span \Eff(B)$,
the effects $\Eff(B)$ span $\Her_{m_B}$.

To verify that the construction above defines an OPT, we check
Assumptions~\ref{sm:ass:OPT-tests}--\ref{sm:ass:OPT-symmetric_composition} of
Sec.~\ref{sm:subsec:OPT-assumptions}.
\begin{description}[style=multiline,leftmargin=3.5em,font=\normalfont]
 \item[\ref{sm:ass:OPT-tests}]
       The systems and tests between systems were defined above, and every test is a finite family of CP maps.
 \item[\ref{sm:ass:OPT-trivial_system}]
       $C_1 = I$ is the trivial system.
       Since $m_I = 1$ and $\cS_I = \Realp$, tests from $I$ to $I$ are precisely finite probability
       distributions $\{ p_i \in \Realp \}_{i=1}^k$ with $\sum_{i=1}^k p_i = 1$.
 \item[\ref{sm:ass:OPT-composition}]
       For any systems $A,B$, we have $\cS_\AB = \cone \{ F \ot G \mid F \in \cS_A,~G \in \cS_B \}$.
       Hence the cone-preservation condition on adjoints is stable under sequential and parallel composition,
       as are complete positivity and trace preservation of the sum of a test.
       Moreover, parallel composition of systems concatenates sequences of primitive factors and is therefore
       associative with unit $I$.
       Composition and tensor product of maps are each associative and satisfy the interchange law.
       The identity channel $\id_A$ on $A$ is the identity map $\id_{\Her_{m_A}}$.
 \item[\ref{sm:ass:OPT-operational_identity}]
       Consider two transformations $f,g \in \Trans(A,B)$ whose operational statistics agree in the presence
       of any ancillary system.
       Thus, for every system $E$, $\rho \in \StN(AE)$, and $a \in \Eff(BE)$,
       we have $\craket{a|f|\rho} = \craket{a|g|\rho}$.
       The set $\StN(A) = \Den_{m_A}$ spans $\Her_{m_A}$, and $\Eff(B)$ spans $\Her_{m_B}$.
       Hence, for all $X \in \Her_{m_A}$ and $Y \in \Her_{m_B}$,
       we have $\tr(Y f(X)) = \tr(Y g(X))$.
       Nondegeneracy of the trace inner product implies $f = g$.
       The converse is immediate.
 \item[\ref{sm:ass:OPT-deterministic_effect}]
       The map $e_A \colon \St_\Real(A) \ni X \mapsto \tr(X) \in \Real$ satisfies
       Eq.~\eqref{sm:eq:FPS-allowed_tests} and is therefore a deterministic effect.
       Moreover, the set $\StN(A) = \Den_{m_A}$ is nonempty.
 \item[\ref{sm:ass:OPT-coarse_graining}]
       Complete positivity and the cone-preservation condition in Eq.~\eqref{sm:eq:FPS-allowed_tests} are
       stable under taking sums, and coarse-graining does not change the sum of a test.
 \item[\ref{sm:ass:OPT-conditioned_tests}]
       Complete positivity and the cone-preservation condition in Eq.~\eqref{sm:eq:FPS-allowed_tests} are
       stable under sequential composition.
       If $\{ T_i \}_i$ is a test from $A$ to $B$ and $\{ S_j^{(i)} \}_j$ is a test from $B$ to $C$ for each
       $i$, then $\sum_{i,j} (S_j^{(i)} \c T_i)^*(I_{m_C}) = \sum_i T_i^*(I_{m_B}) = I_{m_A}$.
       Hence $\{ S_j^{(i)} \c T_i \}_{i,j}$ is also a test.
 \item[\ref{sm:ass:OPT-finite_dimensionality}]
       The space $\St_\Real(A) = \Her_{m_A}$ is finite-dimensional, and $\StN(A) = \Den_{m_A}$ is closed in
       $\St_\Real(A)$.
 \item[\ref{sm:ass:OPT-nontriviality}]
       The informational dimension of the primitive system $C_2$ is at least $2$.
       Indeed, the states $\ket{0}\bra{0}$ and $\ket{1}\bra{1}$ of $C_2$ are perfectly distinguished by the
       measurement $\{ X \mapsto \braket{0|X|0}, X \mapsto \braket{1|X|1} \}$.
 \item[\ref{sm:ass:OPT-symmetric_composition}]
       Conjugation by a unitary swapping the tensor factors maps $\cS_\AB$ bijectively onto $\cS_{BA}$ and
       preserves the identity matrix, and is thus a reversible channel.
       Hence $AB \simeq BA$.
       Choosing the canonical tensor-factor flip as each swap ensures that all the required conditions are satisfied.
\end{description}

Next, we show that the construction satisfies local equivalence.
The argument is the same as in the verification of
Assumption~\ref{sm:ass:OPT-operational_identity} above.
For any systems $A,B$, the set $\StN(A) = \Den_{m_A}$ spans $\Her_{m_A}$, and $\Eff(B)$ spans $\Her_{m_B}$.
Suppose two channels $\Lambda$ and $\cE$ from $A$ to $B$ satisfy
$\tr(a\Lambda(\rho)) = \tr(a\cE(\rho))$ for all $\rho \in \StN(A)$ and $a \in \Eff(B)$.
By these spanning properties and bilinearity,
$\tr(Y\Lambda(X)) = \tr(Y\cE(X))$ holds for all $X \in \Her_{m_A}$ and $Y \in \Her_{m_B}$.
Nondegeneracy of the trace inner product implies $\Lambda = \cE$.
Thus the construction satisfies local equivalence.

Finally, we show that the construction satisfies FPS purification.
For any system $A$, choose $A' \coloneqq C_{m_A}$.
Since $\StN(D) = \Den_{m_D}$ for each $D \in \{ A, A', AA' \}$,
the argument proceeds as for purification in ordinary quantum theory.
Specifically, take any normalized state $\rho$ of $A$ and write its spectral decomposition as
$\rho = \sum_{j=1}^{m_A} \lambda_j\ket{a_j}\bra{a_j}$.
Here $\lambda_j \in \Realp$ for all $j \in \{ 1,\dots,m_A \}$,
$\sum_{j=1}^{m_A} \lambda_j = 1$, and
$\{ \ket{a_j} \}_{j=1}^{m_A}$ is an orthonormal basis of $\Complex^{m_A}$.
With an orthonormal basis $\{ \ket{b_j} \}_{j=1}^{m_A}$ of the Hilbert space $\Complex^{m_A}$
associated with $A'$, define
\begin{alignat}{1}
 \ket\psi &\coloneqq \sum_{j=1}^{m_A} \sqrt{\lambda_j} \ket{a_j} \ot \ket{b_j},
 \qquad \Psi \coloneqq \ket\psi\bra\psi \in \StNP(AA').
 \label{sm:eq:FPS-fixed_purifier_state}
\end{alignat}
Since $(\id_A \ot e_{A'})(\Psi) = \rho$, the state $\Psi$ purifies $\rho$.
Moreover, any purification $\Psi \in \StNP(AA')$ of $\rho$ can be expressed in this form after choosing a
suitable orthonormal basis $\{ \ket{b_j} \}_{j=1}^{m_A}$ of the Hilbert space associated with $A'$.
Thus, given two purifications $\Psi,\Phi \in \StNP(AA')$ of $\rho$, there exists a unitary matrix $V$
taking the orthonormal basis associated with one purification to that associated with the other.
The map $U \colon \St_\Real(A') \ni X \mapsto V X V^\dag \in \St_\Real(A')$ is a reversible channel on
$A'$ and satisfies $\Phi = (\id_A \ot U)(\Psi)$.
Hence the construction satisfies FPS purification.

Furthermore, product standard-basis measurements distinguish $m_A$ normalized states of $A$,
while no POVM on $\Complex^{m_A}$ can distinguish more than $m_A$ states.
Thus $n_A = m_A$, and $d_A = m_A^2$ follows from $\StN(A) = \Den_{m_A}$.
In particular, the chosen purifying system $A' = C_{m_A}$ satisfies
$n_{A'} = n_A$ and $d_{A'} = d_A$.

\subsection{Insufficiency of FPS purification for deriving quantum theory} \label{sm:subsec:FPS-comparison}

We briefly explain why local equivalence, combined with CDP/ES purification, yields quantum theory,
whereas FPS purification does not suffice in general.
A key distinction concerns pure-state transitivity (Lemma~\ref{sm:lemma:CDP_ES-pure_unitary}),
the finite-steering property (Lemma~\ref{sm:lemma:CDP_ES-finite_steering}),
and the existence of reversible dilations (Lemma~\ref{sm:lemma:CDP_ES-reversible_dilation}).
The finite-steering property is also used to prove refinement realization
(Lemma~\ref{sm:lemma:CDP_ES-refinement_realization}).
These properties follow from CDP/ES purification and play important roles in our reconstruction.
FPS purification alone does not guarantee them, which helps explain why it is insufficient.
We show below that all three properties fail in the counterexample above.

First, consider pure-state transitivity (Lemma~\ref{sm:lemma:CDP_ES-pure_unitary}).
In the counterexample, take $A \coloneqq C_2C_2$.
The normalized pure states $P \coloneqq \ket{00}\bra{00}$ and $\Pi$ from Eq.~\eqref{sm:eq:FPS-Bell_state}
cannot be related by any reversible channel on $A$.
Indeed, if a reversible channel $U$ satisfied $U(\Pi) = P$, then, because $U$ is a unitary conjugation,
$U^*(P) = U^{-1}(P) = \Pi$.
But $P \in \cS_A$ whereas $\Pi \notin \cS_A$, contradicting $U^*(\cS_A) \subseteq \cS_A$, which is required by
Eq.~\eqref{sm:eq:FPS-allowed_tests}.
Thus pure-state transitivity fails.

Next, consider finite steering (Lemma~\ref{sm:lemma:CDP_ES-finite_steering}).
In the counterexample choose $A \coloneqq C_2C_2$ and $A' \coloneqq C_4$, and match their computational bases
(the standard bases relabeled with indices starting at $0$).
Define
\begin{alignat}{1}
 \ket\psi &\coloneqq \frac12\sum_{j=0}^3 \ket{j}_A \ot \ket{j}_{A'},
 \qquad \Psi \coloneqq \ket\psi\bra\psi \in \StNP(AA').
 \label{sm:eq:FPS-maximally_entangled_purification}
\end{alignat}
The marginal state on $A'$ is $I_4/4$.
The pair $\{ \Pi^\T / 4, (I_4 - \Pi)^\T / 4 \}$ is a refinement of this state.
Here $\T$ denotes transposition in the chosen computational basis.
If this refinement could be realized by a measurement $\{ F_1, F_2 \}$ on $A$,
with its effects represented by Hermitian matrices via the trace pairing,
then its first outcome would require
$\tr_A((F_1 \ot I_4)\Psi) = F_1^\T / 4 = \Pi^\T / 4$, which implies $F_1 = \Pi$.
However, $\Pi \notin \cS_A$, so no such measurement exists.
Hence the property of Lemma~\ref{sm:lemma:CDP_ES-finite_steering} fails.
Furthermore, since the counterexample satisfies both local equivalence and the quantum state-space property,
Lemma~\ref{sm:lemma:E2} shows that refinement realization also fails in the above counterexample.

Finally, consider reversible dilations (Lemma~\ref{sm:lemma:CDP_ES-reversible_dilation}).
In the counterexample, let $A \coloneqq C_4$ and $B \coloneqq C_2C_2$, and identify their computational bases.
Since $\cS_B \subsetneq \cS_A = \Her_4^+$, the identity map $N$ from $A$ to $B$ is a channel,
whereas its inverse $N^{-1}$ is not, since $(N^{-1})^*(\cS_A) = \cS_A \not\subseteq \cS_B$,
contrary to Eq.~\eqref{sm:eq:FPS-allowed_tests}.
Suppose that $N$ has a reversible dilation, with $E,F,\alpha,W$ as in
Definition~\ref{sm:def:reversible_dilation}, and set
$T \coloneqq W \c (\alpha \ot \id_A) \in \Trans(A,FB)$.
Let $A' \coloneqq C_4$ and choose a maximally entangled normalized pure state $\Phi$ on $AA'$.
The $BA'$-marginal of $(T \ot \id_{A'})(\Phi)$ is $(N \ot \id_{A'})(\Phi)$,
which is pure because $N$ is the identity map on the underlying matrix algebra.
Since a density matrix with a pure marginal factorizes, there exists $\beta \in \StN(F)$ such that
$(T \ot \id_{A'})(\Phi) = \beta \ot (N \ot \id_{A'})(\Phi)$.
The injectivity of the Choi correspondence for a maximally entangled state then gives $T = \beta \ot N$.
The channel $R \coloneqq (e_E \ot \id_A) \c W^{-1} \c (\beta \ot \id_B)$ from $B$ to $A$ satisfies
\begin{alignat}{1}
 R \c N &= (e_E \ot \id_A) \c W^{-1} \c (\beta \ot N)
 = (e_E \ot \id_A) \c (\alpha \ot \id_A) = \id_A.
\end{alignat}
Here the second equality follows from $\beta \ot N = T = W \c (\alpha \ot \id_A)$.
But $R \c N = \id_A$ implies $R = N^{-1}$, a contradiction.

\section{Proof of Lemma~\ref{sm:lemma:OPT-probability_product}} \label{sm:sec:probability_product}

We first prove two preparatory lemmas.

\begin{lemma} \label{sm:lemma:OPT-nonnegative_additive_linearity}
 Let $F \colon [0,1] \to \Realp$ be a function satisfying $F(u + v) = F(u) + F(v)$ for all
 $u,v \in [0,1]$ with $u + v \le 1$.
 Then $F(t) = t F(1)$ for every $t \in [0,1]$.
\end{lemma}
\begin{proof}
 Additivity gives $F(0) = F(0) + F(0)$, so $F(0) = 0$.
 For all integers $m,n$ with $0 \le m < n$, repeated use of additivity gives
 $F(1) = nF(1/n)$ and hence $F(m/n) = m F(1/n) = (m/n)F(1)$.
 Moreover, if $0 \le a \le b \le 1$, nonnegativity and additivity give
 $F(b) = F(a) + F(b - a) \ge F(a)$, so $F$ is nondecreasing.
 By the density of the rational numbers, for every $t \in [0,1]$ we can choose sequences of rational numbers
 $\{ a_k \}_{k=1}^\infty$ and $\{ b_k \}_{k=1}^\infty$ satisfying
 $0 \le a_k \le t \le b_k \le 1$,
 $a_k \xrightarrow{k \to \infty} t$, and $b_k \xrightarrow{k \to \infty} t$.
 By monotonicity and the rational case, for every positive integer $k$ we have
 $a_k F(1) = F(a_k) \le F(t) \le F(b_k) = b_k F(1)$.
 Letting $k \to \infty$, the squeeze theorem gives $F(t) = t F(1)$.
\end{proof}

\begin{lemma} \label{sm:lemma:OPT-distribution_preserving_function}
 Let $H \colon [0,1] \to \Realp$ be a function such that, for every positive integer $k$
 and every finite probability distribution $\{ p_i \}_{i=1}^k$, we have
 $\sum_{i=1}^k H(p_i) = 1$.
 Then $H(t) = t$ for every $t \in [0,1]$.
\end{lemma}
\begin{proof}
 Fix any $u,v \in [0,1]$ with $u + v \le 1$, and set $w \coloneqq 1 - u - v$.
 Since $\{ u,v,w \}$ and $\{ u + v,w \}$ are both probability distributions,
 $H(u) + H(v) + H(w) = 1 = H(u + v) + H(w)$.
 Thus $H(u + v) = H(u) + H(v)$.
 Lemma~\ref{sm:lemma:OPT-nonnegative_additive_linearity} therefore gives $H(t) = t H(1)$.
 Applying the condition to the probability distribution $\{ 1 \}$ gives $H(1) = 1$.
 Consequently, $H(t) = t$.
\end{proof}

We now prove Lemma~\ref{sm:lemma:OPT-probability_product} using the two preceding lemmas.
The interchange law gives
$p \ot q = (p \c \id_I) \ot (\id_I \c q)
= (p \ot \id_I) \c (\id_I \ot q) = p \c q$.
It therefore suffices to show that $p \c q = pq$.
Fix any test $\cQ \coloneqq \{ q_j \}_{j=1}^l$ from $I$ to $I$, and for each $t \in [0,1]$ set
$H_{\cQ}(t) \coloneqq \sum_{j=1}^l (q_j \c t)$.
Each sequential composition $q_j \c t$ is a transformation from $I$ to $I$, so it is a nonnegative real number.
For any probability distribution $\cP \coloneqq \{ p_i \}_{i=1}^k$, the sequentially composed test
$\cQ \c \cP = \{ q_j \c p_i \}_{(i,j)=(1,1)}^{(k,l)}$ is a probability distribution.
Hence $\sum_{i=1}^k H_{\cQ}(p_i) = \sum_{i=1}^k \sum_{j=1}^l (q_j \c p_i) = 1$.
Lemma~\ref{sm:lemma:OPT-distribution_preserving_function} then yields
\begin{alignat}{1}
 \sum_{j=1}^l (q_j \c t) &= t, \qquad \forall t \in [0,1].
 \label{sm:eq:OPT-probability_marginal_sum}
\end{alignat}
Fix any $u,v \in [0,1]$ with $u + v \le 1$ and set $w \coloneqq 1 - u - v$.
Applying Eq.~\eqref{sm:eq:OPT-probability_marginal_sum} to the two probability distributions
$\{ u,v,w \}$ and $\{ u + v,w \}$ gives, for every $t \in [0,1]$,
$(u \c t) + (v \c t) + (w \c t) = t = ((u + v) \c t) + (w \c t)$.
Consequently, $(u + v) \c t = (u \c t) + (v \c t)$.
Fix any $q \in [0,1]$ and consider the function
$F_q \colon [0,1] \ni x \mapsto x \c q \in \Realp$.
The preceding equality shows that $F_q$ is additive.
Applying Lemma~\ref{sm:lemma:OPT-nonnegative_additive_linearity} to $F = F_q$ gives
$F_q(p) = p F_q(1)$ for every $p \in [0,1]$.
Since $\id_I = 1$, we have $F_q(1) = 1 \c q = q$.
Hence $p \c q = F_q(p) = p F_q(1) = pq$.
The final identification with ordinary multiplication is also a consequence of
Proposition~63 of Ref.~\cite{Wes-Wes-Wet-2020}, once the effect-monoid distributive laws are established.
\QED

\section{Proofs of Lemmas~\ref{sm:lemma:quantum_parallel_tensorator}
 and~\ref{sm:lemma:quantum_test_characterization}} \label{sm:sec:quantum_lemmas}

As a preliminary step, we prove the following lemma.

\begin{lemma} \label{sm:lemma:product_tensorator_preserver}
 Consider an OPT satisfying local tomography.
 Suppose that, for each system $C$, there exist a positive integer $r_C$
 and a linear isomorphism $\Repr_C \colon \St_\Real(C) \to \Her_{r_C}$
 from $\StN(C)$ to $\Den_{r_C}$.
 Fix one such $r_C$ and $\Repr_C$ for each $C$.
 Then $r_\AB = r_A r_B$ for any systems $A,B$, and there exists a unique
 real-linear automorphism $G_\AB$ of $\Her_{r_A r_B}$ satisfying,
 for all $x \in \St_\Real(A)$ and $y \in \St_\Real(B)$,
 \begin{alignat}{1}
  G_\AB(\Repr_A(x) \ot \Repr_B(y)) &= \Repr_\AB(x \ot y).
  \label{sm:eq:quantum_parallel_G}
 \end{alignat}
 Moreover, for each positive integer $r$, let $\tau_r(X) \coloneqq X^\T$,
 $\tau_r^0 \coloneqq \id_{\Her_r}$, and $\tau_r^1 \coloneqq \tau_r$.
 For any systems $A,B$, there exist a unitary matrix $W_\AB$
 and $a_\AB,b_\AB \in \{0,1\}$ such that
 \begin{alignat}{1}
  G_\AB &= \Ad_{W_\AB} \c \left( \tau_{r_A}^{a_\AB} \ot \tau_{r_B}^{b_\AB} \right).
  \label{sm:eq:complex_monoidal_tensorator_G}
 \end{alignat}
 If $r_A = 1$, we may choose $a_\AB \coloneqq 0$;
 if $r_B = 1$, we may choose $b_\AB \coloneqq 0$.
\end{lemma}
\begin{proof}
 Local tomography gives $d_\AB = d_A d_B$, and $d_C = r_C^2$,
 so $r_\AB = r_A r_B$.
 By Lemma~\ref{sm:lemma:basic-product_span}, $G_\AB$ is uniquely determined, with
 $G_\AB = \Repr_\AB \c (\Repr_A^{-1} \ot \Repr_B^{-1})$.
 For all $x \in \St_\Real(A)$ and $y \in \St_\Real(B)$,
 $\tr (\Repr_A(x) \ot \Repr_B(y)) = e_A(x)e_B(y) = \tr \Repr_\AB(x \ot y)$.
 By Lemma~\ref{sm:lemma:basic-product_span}, $G_\AB$ preserves trace.
 If $r_A,r_B \ge 2$, the product of normalized pure states is pure by
 Lemma~\ref{sm:lemma:basic-pure_internal_product}, and hence $G_\AB$ maps products of
 rank-one Hermitian matrices to rank-one Hermitian matrices.
 By the Main Theorem of Ref.~\cite{Xu-Zhe-Fos-2015}, it has the form
 $G_\AB = \lambda_\AB \Ad_{W_\AB} \c (\tau_{r_A}^{a_\AB} \ot \tau_{r_B}^{b_\AB})$,
 where $W_\AB$ is an invertible matrix of order $r_\AB$,
 $\lambda_\AB \in \{ -1,1 \}$, and $a_\AB,b_\AB \in \{ 0,1 \}$.
 Trace preservation gives $\lambda_\AB W_\AB^\dag W_\AB = I_{r_A r_B}$;
 thus $\lambda_\AB = 1$ and $W_\AB$ is unitary.
 This proves Eq.~\eqref{sm:eq:complex_monoidal_tensorator_G}.

 If $r_A \ge 2$ and $r_B = 1$, consider the map
 $H \colon \Her_{r_A} \ni X \mapsto G_\AB(X \ot 1) \in \Her_{r_A}$.
 Since the normalized state of $B$ is pure, Lemma~\ref{sm:lemma:basic-pure_internal_product}
 shows that $H$ is a linear automorphism of $\Her_{r_A}$ preserving trace and rank one.
 By Proposition~2.1 of Ref.~\cite{Xu-Zhe-Fos-2015},
 $H = \lambda \Ad_{W_\AB} \c \tau_{r_A}^v$,
 where $W_\AB$ is invertible, $\lambda \in \{ -1,1 \}$, and $v \in \{ 0,1 \}$.
 Trace preservation gives $\lambda W_\AB^\dag W_\AB = I_{r_A}$;
 hence $\lambda = 1$ and $W_\AB$ is unitary.
 Setting $a_\AB \coloneqq v$ and $b_\AB \coloneqq 0$
 gives Eq.~\eqref{sm:eq:complex_monoidal_tensorator_G}.
 If $r_A = 1$ and $r_B \ge 2$, apply the same proposition to
 $\Her_{r_B} \ni X \mapsto G_\AB(1 \ot X) \in \Her_{r_B}$ and set $a_\AB \coloneqq 0$.
 If $r_A = r_B = 1$, then $G_\AB$ is a trace-preserving one-dimensional automorphism
 and hence the identity; take $W_\AB \coloneqq 1$, $a_\AB \coloneqq 0$, and $b_\AB \coloneqq 0$.
 Thus Eq.~\eqref{sm:eq:complex_monoidal_tensorator_G} holds for every pair of systems.
\end{proof}

\subsection{Proof of Lemma~\ref{sm:lemma:quantum_parallel_tensorator}}
\label{sm:subsec:quantum_lemmas-parallel_tensorator}

Fix a family of linear isomorphisms $\{ \Repr_C \}_{C \in \System}$ satisfying Definition~\ref{sm:def:quantum_theory},
and put $m \coloneqq n_A$, $n \coloneqq n_B$.
Since $d_C = n_C^2$ and condition~\hyperref[sm:item:quantum_condition_1]{(1)} gives
$d_\AB = d_A d_B$, local tomography holds.
Thus Lemma~\ref{sm:lemma:product_tensorator_preserver} applies.
The map $G_\AB$ defined there satisfies Eq.~\eqref{sm:eq:quantum_parallel_G}.
Moreover, there exist $a,b \in \{0,1\}$ and a unitary matrix $W_\AB$ such that
$G_\AB(X \ot Y) = \Ad_{W_\AB}(\tau_m^a(X) \ot \tau_n^b(Y))$
for all $X \in \Her_m$ and $Y \in \Her_n$.
If $m \ge 2$, take any $\sigma \in \StN(B)$ and put $Q \coloneqq \Repr_B(\sigma)$.
Since $\id_A \ot \sigma \in \Trans(A,AB)$,
condition~\hyperref[sm:item:quantum_condition_2]{(2)} implies that
$\Phi_\sigma \coloneqq \Repr_\AB \c (\id_A \ot \sigma) \c \Repr_A^{-1}$ is CP.
For any $X \in \Her_m$, Eqs.~\eqref{sm:eq:quantum_parallel_G}
and~\eqref{sm:eq:complex_monoidal_tensorator_G} give
$\Phi_\sigma(X) = \Repr_\AB(\Repr_A^{-1}(X) \ot \sigma)
= G_\AB(X \ot Q) = W_\AB(\tau_m^a(X) \ot \tau_n^b(Q))W_\AB^\dag$.
As $\tr(\tau_n^b(Q)) = 1$, it follows that
$(\id_{\Her_m} \ot \tr)(W_\AB^\dag\Phi_\sigma(X)W_\AB) = \tau_m^a(X)$.
Since $\Phi_\sigma$, $\Ad_{W_\AB^\dag}$, and $\id_{\Her_m} \ot \tr$ are CP,
so is $\tau_m^a$.
But $\tau_m$ is not CP for $m \ge 2$, hence $a = 0$.
If $m = 1$, we may also choose $a \coloneqq 0$ without loss of generality,
since $\tau_m^0 = \tau_m^1 = \id_\Real$.
When $n \ge 2$, the same argument gives $b = 0$.
If $n = 1$, Lemma~\ref{sm:lemma:product_tensorator_preserver} allows us to choose $b \coloneqq 0$.
Therefore Eq.~\eqref{sm:eq:complex_monoidal_tensorator_G} gives
$G_\AB(X \ot Y) = \Ad_{W_\AB}(X \ot Y)$.
Taking $U_\AB \coloneqq W_\AB$ and using Eq.~\eqref{sm:eq:quantum_parallel_G} proves the claim.
\QED

\subsection{Proof of Lemma~\ref{sm:lemma:quantum_test_characterization}}
\label{sm:subsec:quantum_lemmas-test_characterization}

Let $F_i \coloneqq \Repr_B \c f_i \c \Repr_A^{-1}$.
If $\{ f_i \}_{i=1}^k$ is a test, then $\sum_{i=1}^k f_i$ is deterministic, and
condition~\hyperref[sm:item:quantum_condition_2]{(2)} implies that $\sum_{i=1}^k F_i$ is trace preserving.
Conversely, suppose that $\sum_{i=1}^k F_i$ is trace preserving.
By Lemma~\ref{sm:lemma:basic-perfect_discrimination}, there exists a system $E$ with
a family of $k$ normalized states $\{ \eta_i \in \StN(E) \}_{i=1}^k$ and a measurement
$\{ a_j \}_{j=1}^k$ that perfectly distinguishes them.
These satisfy $a_j(\eta_i) = \delta_{i,j}$ $~(\forall i,j \in \{ 1,2,\dots,k \})$.
Using the unitary matrix $U_{BE}$ of Lemma~\ref{sm:lemma:quantum_parallel_tensorator}, define
\begin{alignat}{1}
 G &\colon \Her_{n_A} \ni X \mapsto
 \Ad_{U_{BE}} \left( \sum_{i=1}^k F_i(X) \ot \Repr_E(\eta_i) \right) \in \Her_{n_B n_E}.
\end{alignat}
Then $G$ is a trace-preserving CP map.
Condition~\hyperref[sm:item:quantum_condition_2]{(2)} gives
$h \coloneqq \Repr_{BE}^{-1} \c G \c \Repr_A \in \Trans(A,BE)$.
Moreover, since $e_{BE} \c h = \tr \c G \c \Repr_A = \tr \c \Repr_A = e_A$,
$h$ is a channel.
By Lemma~\ref{sm:lemma:quantum_parallel_tensorator} and linearity, for every $x \in \St_\Real(A)$,
\begin{alignat}{1}
 \Repr_{BE}(h(x))
 &= \Ad_{U_{BE}}\left( \sum_{i=1}^k \Repr_B(f_i(x)) \ot \Repr_E(\eta_i) \right)
 = \Repr_{BE}\left( \sum_{i=1}^k f_i(x) \ot \eta_i \right).
\end{alignat}
Applying $\Repr_{BE}^{-1}$ to both sides gives $h(x) = \sum_{i=1}^k f_i(x) \ot \eta_i$.
Therefore, applying the measurement $\{ a_j \}_{j=1}^k$ on $E$ after $h$ yields
$((\id_B \ot a_j) \c h)(x) = \sum_{i=1}^k a_j(\eta_i) f_i(x) = f_j(x)$
for every $1 \le j \le k$ and $x \in \St_\Real(A)$.
Hence $(\id_B \ot a_j) \c h = f_j$.
Since $\{ (\id_B \ot a_j) \c h \}_{j=1}^k$ is a test, so is $\{ f_j \}_{j=1}^k$.
\QED

\section{Proof of Lemma~\ref{sm:lemma:complex_monoidal_tensorator}} \label{sm:sec:tensorator}

Local equivalence implies local tomography.
For each system $C$, choose an arbitrary representation
$\Repr_C \colon \St_\Real(C) \to \Her_{r_C}$.
These representations will be reselected later so as to satisfy
\begin{alignat}{1}
 \Repr_\AB(\rho \ot \sigma) &= \Ad_{U_\AB}(\Repr_A(\rho) \ot \Repr_B(\sigma)), \qquad
 \forall \rho \in \StN(A), ~\sigma \in \StN(B).
 \tag*{\eqref{sm:eq:complex_monoidal_tensorator}}
\end{alignat}
By Lemma~\ref{sm:lemma:product_tensorator_preserver}, for all systems $A,B$,
$r_\AB = r_A r_B$, and
$G_\AB \coloneqq \Repr_\AB \c (\Repr_A^{-1} \ot \Repr_B^{-1})$
takes the form~\eqref{sm:eq:complex_monoidal_tensorator_G}.
Here $W_\AB$ is unitary; we may choose $a_\AB \coloneqq 0$ if $r_A = 1$
and $b_\AB \coloneqq 0$ if $r_B = 1$.

For the remainder of this proof, given $a,b \in \{0,1\}$,
let $a \oplus b$ denote the sum of $a$ and $b$ modulo $2$.
Fix a system $K$ with $r_K \ge 2$.
A \termdef{path} specifies how to regroup and swap the three systems $K,A,B$ when forming their composite.
For example, one can first swap $A$ and $B$, then form the composite of $K$ and $B$,
and finally form the composite of $KB$ and $A$.
We denote this path by $(KB)A$.
For clarity, if the same system appears more than once, as in $K,A,A$,
we assign distinct labels to its occurrences for the following argument.
For each path $\pi$, let $C_\pi$ denote the output system obtained along $\pi$.
For example, if $\pi$ is denoted by $(KB)A$, then $C_\pi = KBA$.
Let $\gamma_\pi \colon \St_\Real(KAB) \to \St_\Real(C_\pi)$
be the linear map that reorders and regroups the factors of a product state
according to $\pi$, using associativity and reversible swap channels.%
\footnote{More precisely, $\gamma_\pi$ is the linearization of the trilinear map
$\tgamma_\pi$ from $\St_\Real(K) \times \St_\Real(A) \times \St_\Real(B)$
to $\St_\Real(C_\pi)$ that regroups and swaps product states along $\pi$;
that is, it is the linear map $\gamma_\pi$ specified by
$\gamma_\pi(x \ot y \ot z) = \tgamma_\pi(x,y,z)$
$~(\forall x \in \St_\Real(K), y \in \St_\Real(A), z \in \St_\Real(B))$.}
By Lemma~\ref{sm:lemma:basic-product_span}, $\gamma_\pi$ is a linear isomorphism.
Define
\begin{alignat}{1}
 G_\pi &\coloneqq \Repr_{C_\pi} \c \gamma_\pi \c
 \left( \Repr_K^{-1} \ot \Repr_A^{-1} \ot \Repr_B^{-1} \right).
 \label{sm:eq:complex_monoidal_tensorator_G_def}
\end{alignat}
Using $\tau_r \c \Ad_V = \Ad_{\overline V} \c \tau_r$ for any unitary matrix $V$ of order $r$,
where $\overline V$ is the complex conjugate of $V$, $\tau_{mn} = \tau_m \ot \tau_n$,
and the fact that permuting matrix tensor factors can be implemented by unitary conjugation,
we obtain, for each path $\pi$, a unitary matrix $V_\pi$ and exponents
$s_K^\pi,s_A^\pi,s_B^\pi \in \{ 0,1 \}$ such that
\begin{alignat}{1}
 G_\pi &= \Ad_{V_\pi} \c
 \left( \tau_{r_K}^{s_K^\pi} \ot \tau_{r_A}^{s_A^\pi} \ot \tau_{r_B}^{s_B^\pi} \right).
 \label{sm:eq:complex_monoidal_tensorator_G_pi}
\end{alignat}

Choose arbitrary paths $\pi,\pi'$.
Associativity and the reversibility of swap channels give
$\gamma_\pi(\St_+(KAB)) = \St_+(C_\pi)$.
Hence $L \coloneqq (\Repr_{C_\pi} \c \gamma_\pi) \c (\Repr_{C_{\pi'}} \c \gamma_{\pi'})^{-1}$
is a linear automorphism of $\Her_{r_{KAB}}^+$.
By Eq.~\eqref{sm:eq:complex_monoidal_tensorator_G_def}, we also have $G_\pi = L \c G_{\pi'}$.
For each $X \in \{ K,A,B \}$, set $\delta_X \coloneqq s_X^\pi \oplus s_X^{\pi'}$.
Combining this identity with Eq.~\eqref{sm:eq:complex_monoidal_tensorator_G_pi}
and using $\tau_r^2 = \id_{\Her_r}$, we obtain
$\tau_{r_K}^{\delta_K} \ot \tau_{r_A}^{\delta_A} \ot \tau_{r_B}^{\delta_B}
= \Ad_{V_\pi^\dag} \c L \c \Ad_{V_{\pi'}}$.
Since $L$ and the unitary conjugations on the right-hand side
are linear automorphisms of $\Her_{r_{KAB}}^+$, so is the left-hand side.
If $X,Y \in \{ K,A,B \}$ denote distinct labeled factors of rank at least $2$ and
$\delta_X \neq \delta_Y$, then $\tau_{r_K}^{\delta_K} \ot \tau_{r_A}^{\delta_A} \ot \tau_{r_B}^{\delta_B}$
would act as a partial transpose and could not be a linear automorphism of $\Her_{r_{KAB}}^+$, a contradiction.
Therefore $\delta_X = \delta_Y$, or equivalently
$s_X^\pi \oplus s_Y^\pi = s_X^{\pi'} \oplus s_Y^{\pi'}$.
In other words, $s_X^\pi \oplus s_Y^\pi$ does not depend on the path $\pi$.
Apply this observation to the three paths $K(AB)$, $(KA)B$, and $(KB)A$.
For the path $K(AB)$, successive applications of Eq.~\eqref{sm:eq:complex_monoidal_tensorator_G} give
$s_K^{K(AB)} = a_{K,\AB}$, $s_A^{K(AB)} = a_\AB \oplus b_{K,\AB}$, and $s_B^{K(AB)} = b_\AB \oplus b_{K,\AB}$.
For each system $C$, set
\begin{alignat}{1}
 t_C &\coloneqq
 \begin{dcases}
  a_{KC} \oplus b_{KC}, & r_C \ge 2, \\
  0, & r_C = 1.
 \end{dcases}
\end{alignat}
For the path $(KA)B$, we have $s_K^{(KA)B} \oplus s_A^{(KA)B} = a_{KA} \oplus b_{KA} = t_A$;
for the path $(KB)A$, we have $s_K^{(KB)A} \oplus s_B^{(KB)A} = a_{KB} \oplus b_{KB} = t_B$.
As shown above, $s_K^\pi \oplus s_X^\pi$ is path-independent whenever $r_X \ge 2$.
Comparing $K(AB)$ with $(KA)B$ for $X = A$ gives $t_\AB \oplus a_\AB = t_A$ whenever $r_A \ge 2$.
Comparing $K(AB)$ with $(KB)A$ for $X = B$ gives $t_\AB \oplus b_\AB = t_B$ whenever $r_B \ge 2$.
Consequently,
\begin{alignat}{1}
 \tau_{r_A}^{t_\AB \oplus a_\AB \oplus t_A} \ot
 \tau_{r_B}^{t_\AB \oplus b_\AB \oplus t_B} &= \id_{\Her_{r_\AB}},
 \label{sm:eq:complex_monoidal_tensorator_tau}
\end{alignat}
where we have used that $\tau_1$ is the identity when $r_A = 1$ or $r_B = 1$.

Replace the representation $\Repr_C$ of each system $C$ by
$\Repr'_C \coloneqq \tau_{r_C}^{t_C} \c \Repr_C$.
Since the transpose $\tau_{r_C}$ preserves $\Den_{r_C}$, $\Repr'_C$ is again a representation.
Then $G'_\AB \coloneqq \Repr'_\AB \c (\Repr'_A{}^{-1} \ot \Repr'_B{}^{-1})$ is given by
\begin{alignat}{1}
 G'_\AB &= \tau_{r_\AB}^{t_\AB} \c G_\AB \c
 \left( \tau_{r_A}^{t_A} \ot \tau_{r_B}^{t_B} \right)
 = \Ad_{U_\AB} \c \left( \tau_{r_A}^{t_\AB \oplus a_\AB \oplus t_A} \ot
 \tau_{r_B}^{t_\AB \oplus b_\AB \oplus t_B} \right)
 = \Ad_{U_\AB}.
\end{alignat}

The second equality follows from Eq.~\eqref{sm:eq:complex_monoidal_tensorator_G}
using $\tau_r \c \Ad_V = \Ad_{\overline V} \c \tau_r$ and $\tau_r^2 = \id_{\Her_r}$.
Here $U_\AB \coloneqq W_\AB$ if $t_\AB = 0$ and
$U_\AB \coloneqq \overline{W_\AB}$ if $t_\AB = 1$;
in either case, $U_\AB$ is unitary.
The last equality follows from Eq.~\eqref{sm:eq:complex_monoidal_tensorator_tau}.
Thus, by Eq.~\eqref{sm:eq:quantum_parallel_G},
Eq.~\eqref{sm:eq:complex_monoidal_tensorator} holds for every pair of systems $A,B$.
\QED

\section{Comparison with an earlier reconstruction using CDP purification} \label{sm:sec:CDP_comparison}

Ref.~\cite{Chi-Dar-Per-2011} derives quantum theory for causal operational-probabilistic theories
by imposing, in addition to Local Distinguishability and CDP purification,
the three postulates defined there: Perfect Distinguishability, Ideal Compression,
and Pure Conditioning.
By contrast, our derivation proceeds through homogeneity of the state cones and obtains the quantum
state-space property from the Jordan structure of symmetric cones, without needing to assume any of these
three additional postulates along the way.
It is nevertheless useful to identify the stages of our proof at which these properties follow.
We briefly discuss this point below.

Perfect Distinguishability follows once measurement no-restriction has been derived
(Proposition~\ref{sm:pro:CDP_ES-no_restriction}).
Indeed, for any normalized state $\rho \in \StN(A)$ that is not internal,
there exists a nonzero $a \in \Eff_+(A)$ with $a(\rho) = 0$.
Choose $\sigma \in \StN(A)$ such that
$a(\sigma) \ge a(\tau)$ $~(\forall \tau \in \StN(A))$.
Because $a \neq \zero$, we have $a(\sigma) > 0$.
With $b \coloneqq a / a(\sigma)$, the pair $\{ b, e_A - b \}$ is a measurement
that perfectly distinguishes $\sigma$ from $\rho$.

Ideal Compression, however, does not necessarily follow from our reconstruction.
It requires every face of a state space to be identified with the full state space
of a suitable system through appropriate encoding and decoding.
This requirement is used in the reconstruction of Ref.~\cite{Chi-Dar-Per-2011},
whereas our reconstruction establishes the full quantum operational structure
for the systems in the theory without relying on Ideal Compression.
In particular, our assumptions do not guarantee the existence of a system of
informational dimension $n$ for every positive integer $n$.
For example, a theory whose systems have informational dimensions restricted to $1$
and the even integers $2, 4, 6, \dots$, while otherwise retaining the full standard
finite-dimensional quantum operational structure, does not satisfy Ideal Compression
for a rank-three state of a four-level system, because no three-level system exists.

Pure Conditioning follows after the full identification of quantum theory.

\section{Discussion of extensions to the reconstruction of infinite-dimensional quantum theory}
\label{sm:sec:infinite}

This work aims only at the reconstruction of finite-dimensional quantum theory.
In this appendix, we discuss whether our derivation can be extended to infinite-dimensional quantum theory.
Although our results may contribute to such a reconstruction, they cannot be applied directly in infinite dimensions.
In particular, our proof relies on finite dimensionality, for example, in arguments involving
cone geometry and compactness, the classification of simple EJAs, and the counting of ranks and dimensions.
For simplicity, the following discussion concerns countably infinite-dimensional systems,
that is, systems associated with Hilbert spaces possessing countable orthonormal bases.
Systems of uncountably infinite dimension would likely require still more involved analysis.

\subsection{Representative properties that fail in infinite-dimensional quantum theory}

We give two examples of basic properties of finite-dimensional quantum theory that fail in infinite dimensions.
Throughout, let $\cH_A \coloneqq \ell^2(\Natural)$ and
$\cH_B \coloneqq \ell^2(\Natural)$, where
$\ell^2(\Natural)$ denotes the space of square-summable sequences,
and write $\ket n$ for the orthonormal basis vectors in each space.
The symbols $I_A$ and $I_B$ denote the identity operators on the respective Hilbert spaces.

The first example shows that two purifications $\Psi_1, \Psi_2 \in \StNP(AB)$ of a normalized state
$\rho \in \StN(A)$ need not be related by a reversible channel on $B$.
To give a concrete example, consider the normalized state
$\rho \coloneqq \sum_{n=1}^\infty \lambda_n\ket n\bra n$, with
$\lambda_n \coloneqq 2^{-n}$.
Consider its purifications $\Psi_1$ and $\Psi_2$ defined by
\begin{alignat}{1}
 \Psi_j \coloneqq \ket{\psi_j}\bra{\psi_j} ~(j \in \{ 1, 2 \}), \qquad
 \ket{\psi_1} \coloneqq \sum_{n=1}^\infty \sqrt{\lambda_n}\ket n_A \ot \ket n_B, \qquad
 \ket{\psi_2} \coloneqq \sum_{n=1}^\infty \sqrt{\lambda_n}\ket n_A \ot \ket{n + 1}_B.
\end{alignat}
Suppose there were a reversible channel $U$ on $B$ such that
$(\id_A \ot U)(\Psi_1) = \Psi_2$.
In the standard quantum-channel model, $U$ has the form
$X \mapsto VXV^\dag$ for a unitary operator $V$ on $\cH_B$.
The assumed equality would imply, up to an overall complex phase,
$\ket{\psi_2} = (I_A \ot V)\ket{\psi_1}$.
Since $\lambda_n > 0$ for every $n$, it would follow, after absorbing this phase into $V$, that
$V\ket n_B = \ket{n + 1}_B$ for every $n$.
However, the vectors $\{ V\ket n_B \}_{n \in \Natural}$ span a proper subspace
orthogonal to $\ket 1_B$.
Thus $V$ would fail to be surjective and hence could not be unitary.
No reversible channel $U$ on $B$ can therefore satisfy
$(\id_A \ot U)(\Psi_1) = \Psi_2$.

The second example shows that an infinite-dimensional system $A$ has no internal state.
Here $\StN(A)$ consists of the positive trace-class operators of trace one,
$\St(A)$ consists of those with trace at most one, and $\St_+(A)$ consists of all positive
trace-class operators.
Then $\St_\Real(A)$ is the real vector space of self-adjoint trace-class operators on $\cH_A$,
and the set of refinement components of a normalized state $\sigma \in \StN(A)$ is
$D_\sigma \coloneqq \{ X \in \St_\Real(A) \mid \zero \le X \le \sigma \}$.
Suppose for contradiction that some $\sigma \in \StN(A)$ is internal, so that
$\Span D_\sigma = \St_\Real(A)$.
For any unit vector $\ket v \in \cH_A$, set $P_v \coloneqq \ket v\bra v$.
Since $P_v \in \St_\Real(A) = \Span D_\sigma$, it can be written as
$P_v = \sum_{j=1}^m a_j X_j$ $~(a_j \in \Real, ~X_j \in D_\sigma)$.
Each $X_j$ satisfies $\zero \le X_j \le \sigma$, and hence
$P_v \le (\sum_{j=1}^m |a_j|)\sigma$.
Thus, for every unit vector $\ket v$, there exists $c > 0$ such that $P_v \le c\sigma$.
Take a spectral decomposition
$\sigma = \sum_{n=1}^\infty \mu_n\ket{\phi_n}\bra{\phi_n}$,
where $\{ \ket{\phi_n} \}$ is an orthonormal basis of $\cH_A$,
$\mu_n \ge 0$, and $\sum_{n=1}^\infty \mu_n = 1$.
Since $\mu_n \to 0$, we can choose a strictly increasing sequence
of positive integers $\{ n_k \}_{k \in \Natural}$ such that
$\mu_{n_k} < 4^{-k}$ for all $k \in \Natural$.
Choose $\ket v \coloneqq \sum_{k=1}^\infty 2^{-k/2}\ket{\phi_{n_k}}$;
this is a unit vector.
The inequality $P_v \le c\sigma$ would imply, for each positive integer $k$, that
$2^{-k} = \braket{\phi_{n_k}|P_v|\phi_{n_k}}
\le c\braket{\phi_{n_k}|\sigma|\phi_{n_k}} = c\mu_{n_k} < c4^{-k}$.
Thus $2^k < c$ for every positive integer $k$, a contradiction.
We conclude that $A$ has no internal state.
Moreover, $\St_+(A)$ has empty interior in the trace-norm topology.
Indeed, any positive trace-class operator has eigenvalues tending to zero,
so an arbitrarily small trace-norm perturbation can destroy positivity.

\subsection{Reconstruction from finite-dimensional subsystems}

Suppose a countably infinite-dimensional system $A$ is equipped with a structure of
finite-dimensional subsystems.
If the operational theory of each finite-dimensional subsystem satisfies the basic OPT assumptions
of this work, as well as local equivalence and either CDP purification or ES purification,
then our finite-dimensional theorem shows that the operational theory of that subsystem
satisfies the definition of quantum theory.
For countably infinite-dimensional systems, Hardy has stated that his axioms also cover
countably infinite-dimensional spaces because every finite-dimensional subspace obeys quantum theory
\cite[Sec.~9]{Har-2001}.
That section distinguishes ``countable'' from ``continuous dimensional,'' and we use only its claim
concerning the former.
Hardy's ``subspace'' is a support-restricted sector of a system (Sec.~6.14 of Ref.~\cite{Har-2001}),
whose lower-dimensional behavior is postulated in Axiom~3.
Here we instead assume finite-dimensional subsystems with specified OPT structures,
without deriving their existence or mutual compatibility from support restrictions.

This precedent suggests that reconstruction subsystem by subsystem is a natural starting point
in countably infinite dimensions.
Our theorem, however, characterizes each finite-dimensional subsystem separately.
It does not establish mutual compatibility of the quantum representations obtained for different subsystems.
Nor does it establish that those representations determine the states, transformations, and effects
of the entire system $A$.
Therefore, to infer countably infinite-dimensional quantum theory from our finite-dimensional results,
additional arguments are required to make the representations of the finite-dimensional subsystems
mutually compatible and to determine the operational theory of the whole system from those compatible data.
This further step is not merely a formal limiting operation but involves a nontrivial extension problem.
Indeed, D'Ariano discusses deriving a formulation in terms of complex Hilbert spaces in finite dimensions,
whereas in infinite dimensions he discusses a $C^*$-algebraic representation of generalized transformations
via the Gelfand--Naimark--Segal construction~\cite{Dar-2007}.
Westerbaan and van de Wetering explicitly assume directed completeness of predicate spaces in a
categorical reconstruction that also includes infinite-dimensional systems~\cite{Wes-Wet-2022}.
These works illustrate that an independent global argument may be needed to pass from
finite-dimensional subsystems to the entire system.
The two examples in the preceding subsection show that properties of finite-dimensional subsystems
need not remain valid for the entire system.
Thus, in addition to the compatibility of quantum representations, one must separately examine how to
formulate the states, effects, transformations, and purification conditions of the entire system.

\end{document}